\pdfoutput=1
\documentclass[11pt, oneside]{amsart}

\usepackage{tabularx,array,adjustbox}
\newcolumntype{Y}{>{\centering\arraybackslash}X}
\newcommand{\wrapmath}[1]{\adjustbox{max width=\linewidth}{\(\displaystyle #1\)}}
\usepackage[letterpaper,top=2.1cm, bottom=2.1cm, left=1.8cm, right=1.8cm]{geometry}

\usepackage[utf8]{inputenc}
\usepackage[T1]{fontenc}
\usepackage[english]{babel}
\usepackage[foot]{amsaddr}

\usepackage[dvipsnames]{xcolor} %

\usepackage[colorlinks]{hyperref}
\hypersetup{
  linkcolor=blue,
  citecolor=blue,
  urlcolor=blue!50!black
}

\usepackage{xurl}
\usepackage{amssymb}
\usepackage{amsmath}
\usepackage{stackengine}
\usepackage{amsfonts}
\usepackage{amsthm}
\usepackage{mathtools}
\usepackage{mathrsfs}
\usepackage[thinc]{esdiff}

\usepackage{graphicx}
\usepackage{overpic}
\usepackage{caption, subcaption}
\usepackage{placeins}

\usepackage{enumerate}
\usepackage{enumitem}
\usepackage[normalem]{ulem}

\usepackage{latexsym}
\usepackage{bm}
\usepackage{bbm}
\usepackage[dvipsnames]{xcolor}
\usepackage{soul}

\usepackage{empheq}
\usepackage{dirtytalk}
\usepackage[symbol]{footmisc}
\usepackage{csquotes}
\usepackage[makeroom]{cancel}
\usepackage{cancel}

\definecolor{lime}{RGB}{89,223,0}
\definecolor{celestino}{RGB}{102,178,255}
\definecolor{bordeaux}{RGB}{140,0,0}

\def\blue#1{\textcolor{blue}{#1}}

\newcommand{\hide}[1]{}

\makeatletter
\newcommand{\vastt}{\bBigg@{3}}
\newcommand{\vast}{\bBigg@{3.3}}
\newcommand{\Vast}{\bBigg@{4}}
\makeatother

\def\cC{\mathscr C}

\def\cX{\mathscr X}

\def\bT{\mathbb T}

\def\wt{\widetilde}

\DeclareMathOperator{\Id}{Id}
\newcommand{\E}{\mathbb{E}}
\newcommand{\R}{\mathbb{R}}
\newcommand{\N}{\mathbb{N}}

\newcommand{\eps}{\varepsilon}

\newcommand{\sums}{\sum_{\sigma \in \{+,-\}}}
\newcommand{\sumt}{\sum_{\tau \in \{+,-\}}}
\newcommand{\ds}{d_\sigma}
\newcommand{\ms}{m_\sigma}
\newcommand{\cs}{c_\sigma}
\newcommand{\dsb}{d_{\bar{\sigma}}}
\newcommand{\msb}{m_{\bar{\sigma}}}
\newcommand{\csb}{c_{\bar{\sigma}}}
\newcommand{\ka}{\hat{\kappa}}
\newcommand{\quadre}[1]{\left[{#1}\right]}
\newcommand{\tonde}[1]{\left({#1}\right)}

\usepackage{commath}
\renewcommand{\abs}[1]{\left\lvert{#1}\right\rvert}

\newcommand*{\crosssymbol}{%
    \text{%
      \raise 1ex\hbox{%
        \rlap{\vrule height.2pt depth.2pt width .75ex}%
        \hbox to .75ex{\hss\vrule height .5ex depth 1ex\hss}%
      }%
    }%
}
\newcommand*{\crossupsidedown}{%
    \text{%
      \raise .5ex\hbox{%
        \rlap{\vrule height.2pt depth.2pt width .75ex}%
        \hbox to .75ex{\hss\vrule height 1ex depth .5ex\hss}%
      }%
    }%
}

\def\lra{\longrightarrow}
\def\ignore#1{}

\theoremstyle{plain}
\newtheorem{theorem}{Theorem}[section]

\newtheorem{lemma}[theorem]{Lemma}

\theoremstyle{definition}
\newtheorem{definition}[theorem]{Definition}

\newtheorem{remark}[theorem]{Remark}

\theoremstyle{definition}

\numberwithin{equation}{section}

\title[]{High-Frequency Analysis of a Trading Game\\with Transient Price Impact}

\author{Marcel Nutz}
\address[MN]{Departments of Statistics and Mathematics, Columbia University}
\email{mnutz@columbia.edu}
\thanks{MN was supported by NSF Grants DMS-2407074 and DMS-2106056.}

\author{Alessandro Prosperi}
\address[AP]{Department of Statistics, Columbia University}
\email{alessandro.prosperi@columbia.edu}


\date{\today}

\begin{document}

\begin{abstract}
\dots
\end{abstract}

\begin{abstract}
We study the high-frequency limit of an $n$-trader optimal execution game in discrete time. Traders face transient price impact of Obizhaeva--Wang type in addition to quadratic instantaneous trading costs $\theta(\Delta X_t)^2$ on each transaction~$\Delta X_t$. There is a unique Nash equilibrium in which traders choose liquidation strategies minimizing expected execution costs. In the high-frequency limit where the grid of trading dates converges to the continuous interval~$[0,T]$, the discrete equilibrium inventories converge at rate $1/N$ to the continuous-time equilibrium of an Obizhaeva--Wang model with additional quadratic costs $\vartheta_0(\Delta X_0)^2$ and $\vartheta_T(\Delta X_T)^2$ on initial and terminal block trades, where $\vartheta_0=(n-1)/2$ and $\vartheta_T=1/2$. The latter model was introduced by Campbell and Nutz as the limit of continuous-time equilibria with vanishing instantaneous costs. Our results extend and refine previous results of Schied, Strehle, and Zhang for the particular case $n=2$ where $\vartheta_0=\vartheta_T=1/2$. In particular, we show how the coefficients $\vartheta_0=(n-1)/2$ and $\vartheta_T=1/2$ arise endogenously in the high-frequency limit: the initial and terminal block costs of the continuous-time model are identified as the limits of the cumulative discrete instantaneous costs incurred over small neighborhoods of~$0$ and~$T$, respectively, and these limits are independent of $\theta>0$. By contrast, when $\theta=0$ the discrete-time equilibrium strategies and costs exhibit persistent oscillations and admit no high-frequency limit, mirroring the non-existence of continuous-time equilibria without boundary block costs. Our results show that two different types of trading frictions---a fine time discretization and small instantaneous costs in continuous time---have similar regularizing effects and, in the limiting regime, select a canonical continuous-time model with transient price impact and endogenous block costs.
\end{abstract}

\maketitle

\section{Introduction}

Transaction costs are a key consideration for financial institutions. In equity trading, the lion's share of costs is due to price impact, i.e., the fact that buy (sell) orders tend to push prices up (down). Following~\cite{almgren.chriss.01}, price impact is often modeled in reduced form, positing that each atomic trade mechanically leads to a price change. Later models incorporate price resilience (transient impact), meaning that prices revert over time once the buying or selling pressure ceases. The most tractable formulation is the Obizhaeva--Wang model~\cite{ObizhaevaWang.13}, which uses an exponential decay kernel. Starting with \cite{garleanu.pedersen.16,GraeweHorst.17}, numerous works have added quadratic instantaneous costs on the trading rate to the Obizhaeva--Wang impact cost. As illustrated in~\cite{GraeweHorst.17}, this \say{regularizes} the problem and leads to smoother optimal trading strategies; see also~\cite{HorstKivman.24}. We refer to \cite{CarteaJaimungalPenalva.15,Webster.23} for further background and extensive references on price impact. Strategic interactions between several large traders are studied in game-theoretic models. This branch of the literature emerged to study predatory trading, where one trader exploits the price impact of a second trader who needs to unwind a position~\cite{PedersenBrunnermeier.05,CarlinLoboViswanathan.07}. For the Obizhaeva--Wang model regularized by quadratic instantaneous costs, \cite{Strehle.17} shows that there is a unique Nash equilibrium, whose closed form is provided in~\cite{CampbellNutz.24}. While these works follow the optimal execution literature in assuming that the unaffected price is a martingale, they have been generalized in several directions, such as incorporating alpha signals~\cite{NeumanVoss.23}, alpha signals and non-exponential decay kernels~\cite{AbiJaberNeumanVoss.24}, or self-exciting order flow~\cite{FuHorstXia.22}.

The goal of the present paper is to shed light on the Nash equilibria of trading games in the Obizhaeva--Wang model without regularization. Surprisingly, a naive formulation in continuous time does not admit an equilibrium, as established by~\cite{SchiedStrehleZhang.17} and \cite{CampbellNutz.24}. They further show that existence is restored if very specific costs on block trades are added to the Obizhaeva--Wang impact cost. Namely, in a game with $n$ traders, a block trade of size $\Delta X_0$ at the initial time $t=0$ is charged $\vartheta_0(\Delta X_0)^2$, where $\vartheta_0:=(n-1)/2$, and a block trade $\Delta X_T$ at the terminal time $t=T$ is charged $\vartheta_T(\Delta X_T)^2$, where $\vartheta_T:=1/2$ (up to reparametrizing time). On the open interval $(0,T)$, no additional costs are charged. For $n=2$ traders, as studied in~\cite{SchiedStrehleZhang.17}, the initial and terminal costs coincide. For general~$n$, as studied in~\cite{CampbellNutz.24}, the two costs differ, with $\vartheta_0$ depending directly on~$n$, making this adjustment even more surprising. Conversely, these works show that for general initial inventories of the traders, no equilibrium exists unless $\vartheta_0$ and $\vartheta_T$ have exactly the stated values. The two works further motivate their models by asymptotic considerations. On the one hand, \cite{SchiedStrehleZhang.17} shows that their continuous-time equilibrium strategies are the high-frequency limits of \emph{discrete-time} equilibria. The discrete-time models use Obizhaeva--Wang impact and an additional quadratic instantaneous cost $\theta (\Delta X_t)^2$, where $\theta>0$ is arbitrary and fixed. The authors further show that without instantaneous costs, the high-frequency limit does not exist because strategies oscillate. These results build on~\cite{Schoneborn.08,SchonebornSchied.09,SchiedZhang.19}, which documented such oscillations in different contexts; see also~\cite{LuoSchied.19}. On the other hand, \cite{CampbellNutz.24} shows that their equilibrium is the limit of \emph{continuous-time} equilibria with quadratic instantaneous costs $\eps (dX_t/dt)^2$ as $\eps\to0$.

The present work refines and extends the analysis of~\cite{SchiedStrehleZhang.17} in several ways. First, we generalize from $n=2$ to an arbitrary number~$n$ of traders. We show that the high-frequency limits of discrete-time equilibria with instantaneous costs $\theta (\Delta X_t)^2$ recover the continuous-time model of~\cite{CampbellNutz.24} with the block cost coefficients~$\vartheta_0$ and~$\vartheta_T$, which are distinct for $n>2$. Second, refining the results of~\cite{SchiedStrehleZhang.17}, we show not only that the total execution costs converge, but also how the different parts of the continuous-time model emerge in the high-frequency limit: The initial block costs are identified as the limits of the instantaneous costs accrued over an initial interval $[0,t_0]$ with arbitrary $0<t_0<T$; similarly, the terminal block costs are the limits of the instantaneous costs accrued over an interval $[t_0,T]$. Moreover, the \say{regular} Obizhaeva--Wang impact costs of the continuous-time model are the limits of the Obizhaeva--Wang costs of the discrete-time models. Third, we not only show the qualitative convergence of the equilibria, but also establish the convergence rate $1/N$ for the trading strategies, where~$N$ is the number of trading periods in~$[0,T]$. Finally, we show that when the discrete-time models are formulated without instantaneous costs ($\theta=0$), the equilibrium strategies oscillate in the high-frequency limit. This yields a one-to-one correspondence between non-existence of the high-frequency limits and non-existence of the continuous-time equilibria without block costs. This correspondence is robust and even extends to fine details: For $n>2$, \cite{CampbellNutz.24} shows that an equilibrium can exist for particular initial inventories of the traders even when only one of the two coefficients $\vartheta_0$ and $\vartheta_T$ has the \say{correct} value---namely, when initial inventories are symmetric or sum to zero, respectively. We further link this to high-frequency limits of discrete-time models where instantaneous costs are charged only on an initial or terminal portion of the time interval.

Our results complement the analysis of~\cite{CampbellNutz.24} for vanishing instantaneous costs in continuous time. Taken together, a high-level picture emerges: discretizing time has the same regularizing effect as adding a small instantaneous cost in continuous time, and yields the same limit. This leads us to conjecture a universality phenomenon: a broad class of trading frictions can be introduced to obtain existence of equilibria in trading games with Obizhaeva--Wang price impact, and the small-friction limits of such regularizations all yield the same model, namely Obizhaeva--Wang price impact with additional block costs as specified in~\cite{SchiedStrehleZhang.17} and~\cite{CampbellNutz.24}.

\vspace{1em}

The remainder of this paper is organized as follows. Section~\ref{The_Problem} formulates and solves the discrete-time models, while Section~\ref{cont_time_section} recalls the corresponding continuous-time results. Section~\ref{hfl_section} states our main results: the high-frequency limits of the discrete-time equilibrium strategies and costs (Theorems~\ref{strat_asympt_thm_theta_positive} and~\ref{costs asymptotics thm}), as well as the corresponding oscillatory asymptotics for $\theta=0$ (Theorems~\ref{strat_osc_thm} and~\ref{cost_asympt_thm_theta_zero}). Appendix~\ref{closed form section} provides a closed-form expression for the discrete-time equilibrium strategies that is used in the high-frequency proofs. Appendix~\ref{exist_unique_Nash} contains the proofs for the discrete-time results in Section~\ref{The_Problem}, while Appendix~\ref{sec_3_proofs_appendix} collects the proofs for the main results in Section~\ref{hfl_section}. Finally, Appendix~\ref{half-grid-block-costs} analyzes the high-frequency asymptotics when instantaneous costs are charged only on an initial or terminal portion of the time interval.

\section{Discrete-Time Equilibrium}\label{The_Problem}

\subsection{Model Specifications}

We consider $n\geq2$ agents trading a single risky asset on the discrete time grid $0=t_0<t_1<\dots<t_N=T$, and a filtered probability space $(\Omega,\mathscr{F},(\mathscr{F}_t)_{t\ge 0},\mathbb{P})$ where $\mathscr{F}_0$ is $\mathbb{P}$-trivial. The \emph{unaffected} price $S^0=(S^0_t)_{t\ge 0}$ is a square-integrable, right-continuous martingale. The definitions below detail how trading generates transient price impact governed by the exponential \emph{decay kernel} $G:\mathbb{R}_+\to\mathbb{R}_+$,
\begin{align*}
    G(t)=e^{-\rho t},
\end{align*}
where $\rho>0$. (A more general form is $G(t)=\lambda e^{-\rho t}$, but we set $\lambda=1$ without loss of generality as all other quantities can be rescaled accordingly.)

\begin{definition}[Admissible trading strategy]
Given a grid $\mathbb{T}=\{t_0,\dots,t_N\}$ and an initial inventory $x\in\mathbb{R}$, an \emph{admissible trading strategy} is a vector $\bm{\xi} = (\xi_0,\dots,\xi_N)^\top$ of random variables such that
\begin{enumerate}[label=(\alph*)]
    \item each $\xi_i$ is $\mathscr{F}_{t_i}$-measurable and bounded;
    \item\label{terminal inv cond} $x=\xi_0+\dots+\xi_N$ $\mathbb{P}$-a.s.
\end{enumerate}
We write $\mathscr{X}(x,\mathbb{T})$ for the set of admissible strategies.
\end{definition}

Intuitively, agent $i$ chooses $\bm{\xi}_i=(\xi_{i,0},\dots,\xi_{i,N})^\top\in\mathscr{X}(x_i,\mathbb{T})$, where $x_i$ denotes the agent's initial inventory and $\xi_{i,k}$ is the number of shares traded at time $t_k$, with the sign convention that $\xi_{i,k}>0$ is a sell and $\xi_{i,k}<0$ is a buy. Condition~\ref{terminal inv cond} enforces liquidation by $t_N=T$. 
Collecting agents' strategies in the matrix $\Xi=[\bm{\xi}_1,\dots,\bm{\xi}_n]$,  the (affected) price process is 
\begin{align*}
    S_t^\Xi = S_t^0 - \sum_{t_k<t} G(t-t_k)\sum_{i=1}^n \xi_{i,k}.
\end{align*}

We fix an \emph{instantaneous cost} parameter $\theta\geq0$ and define the execution cost of agent~$i$ as follows.

\begin{definition}[Execution cost]
Given a grid $\mathbb{T}$ and inventories $(x_1,\dots,x_n)$, the execution cost of $\bm{\xi}_i$ given opponents' strategies $\bm{\xi}_{-i}=[\bm{\xi}_1,\dots,\bm{\xi}_{i-1},\bm{\xi}_{i+1},\dots,\bm{\xi}_n]$ is
\begin{equation}\label{liqcost}
    \mathscr{C}_{\mathbb{T}}(\bm{\xi}_i \mid \bm{\xi}_{-i})=x_iS^0_0 +
    \sum_{k=0}^N\Big(
        \frac{G(0)}{2}\xi_{i,k}^2
        - S_{t_k}^\Xi\xi_{i,k}
        + \frac{G(0)}{2}\sum_{j\neq i}\xi_{i,k}\xi_{j,k}
        + \theta\xi_{i,k}^2
    \Big).
\end{equation}
\end{definition}

In~\eqref{liqcost}, the cross-term describes the standard (symmetric) tie-breaking rule that applies when agents place orders at the same instant; see \cite{SchiedZhang.19, LuoSchied.19}. In addition to the cost of transient impact, each trade incurs quadratic instantaneous (or ``temporary impact'') costs  $\theta\xi_{i,k}^2$; see \cite{Gatheral.10} for a related discussion. %

\subsection{Nash Equilibrium}

Fix a grid $\mathbb{T}$ and initial inventories $(x_1,\dots,x_n)$. Each agent $i$ is risk-neutral and chooses an admissible strategy to minimize the expected execution cost~\eqref{liqcost}, where we may assume $S^0_t\equiv 0$ without loss of generality.
This leads to the standard notion of Nash equilibrium.

\begin{definition}[Nash equilibrium]\label{Nasheqdiscdef}
A \emph{Nash equilibrium} is a profile $(\bm{\xi}^*_1,\dots,\bm{\xi}^*_n)\in\prod_i \mathscr{X}(x_i,\mathbb{T})$ such that
\[
    \mathbb{E}[\mathscr{C}_{\mathbb{T}}(\bm{\xi}^*_i \mid \bm{\xi}^*_{-i})]
    =
    \underset{\bm{\xi}\in\mathscr{X}(x_i,\mathbb{T})}{\min}
    \mathbb{E}[\mathscr{C}_{\mathbb{T}}(\bm{\xi} \mid \bm{\xi}^*_{-i})],\qquad\text{for every }i=1,\dots,n.
\]
\end{definition}

To state a more explicit expression for the objective functional, let $\delta_{ij}$ denote the Kronecker delta and define, for $i,j=0,\dots,N$,
\begin{align}\label{gammatheta}
    \Gamma_{ij}^\theta
    := G(|t_i - t_j|) + 2\theta\delta_{ij},
    \qquad
    \widetilde{\Gamma}_{ij}
    :=
    \begin{cases}
        0, & i<j,\\
        \frac12 G(0), & i=j,\\
        \Gamma_{ij}^0, & i>j.
    \end{cases}
\end{align}
Moreover, we introduce the vectors
\begin{equation}\label{optimal sol}
    \bm{v}
    :=
    \frac{(\Gamma^{\theta} + (n-1)\widetilde{\Gamma})^{-1}\bm{1}}
         {\bm{1}^\top(\Gamma^{\theta} + (n-1)\widetilde{\Gamma})^{-1}\bm{1}},
    \qquad
    \bm{w}
    :=
    \frac{(\Gamma^{\theta} - \widetilde{\Gamma})^{-1}\bm{1}}
         {\bm{1}^\top(\Gamma^{\theta} - \widetilde{\Gamma})^{-1}\bm{1}}.
\end{equation}

\begin{remark}\label{w_indep_n}
We observe that $\bm{w}$ does \emph{not depend on $n$}, whereas $\bm{v}$ depends on $n$ through $\Gamma^{\theta}+(n-1)\widetilde{\Gamma}$. An interpretation for~$\bm{v}$ and~$\bm{w}$ will be given in Remark~\ref{vwSpecialCases}.
\end{remark}

The next lemma ensures that \eqref{optimal sol} is well-defined. We call a (possibly non-symmetric) square matrix $A$ \emph{positive} if $\bm{x}^\top A\bm{x}>0$ for all nonzero $\bm{x}$. Then, $A$ is invertible, and $A^{-1}$ is positive as well. 

\begin{lemma}\label{matrices positive definite}
    For all $\theta\ge 0$, the matrices $\Gamma^\theta$ and $\Gamma^\theta+(n-1)\widetilde{\Gamma}$ and $\Gamma^\theta-\widetilde{\Gamma}$ are positive. In particular, the denominators in \eqref{optimal sol} are strictly positive.
\end{lemma}

The proof is analogous to \cite[Lemma~3.2]{SchiedZhang.19} and omitted. The next result gives an explicit expression for agent $i$'s  objective functional.

\begin{lemma}[Explicit objective]\label{formoffunctionaldisc}
    For $\bm \xi_i\in\mathscr{X}(x_i,\mathbb{T})$ and competitors' strategies $\bm \xi_j\in\mathscr{X}(x_j,\mathbb{T})$,
    \begin{align*}
        \mathbb{E}[\mathscr{C}_{\mathbb{T}}(\bm{\xi}_i \mid \bm{\xi}_{-i})]
        &= \mathbb{E}\Bigl[\frac{1}{2}\bm{\xi}_i^\top\Gamma^\theta\bm{\xi}_i
        + \bm{\xi}_i^\top\widetilde{\Gamma}\Bigl(\sum_{j\neq i}\bm{\xi}_j\Bigr)\Bigr].
    \end{align*}
\end{lemma}

The proof follows \cite[Lemma~3.1]{LuoSchied.19} and is omitted. The final result of this section establishes existence and uniqueness of a Nash equilibrium; it is deterministic and described by~$\bm{v}$ and~$\bm{w}$ of~\eqref{optimal sol}.

\begin{theorem}[Discrete  equilibrium]\label{exuniqnasheqdiscrete}
For any grid $\mathbb{T}$, $\theta\ge 0$, and initial inventories $(x_1,\dots,x_n)\in\mathbb{R}^n$, there exists a unique Nash equilibrium
$(\bm{\xi}^*_1,\dots,\bm{\xi}^*_n)\in \prod_i \mathscr{X}(x_i,\mathbb{T})$.
The equilibrium strategies are deterministic and given by
\begin{align}\label{NashEqStratDisc}
    \bm{\xi}^*_i = \bar{x}\bm{v} + (x_i-\bar{x})\bm{w},
    \qquad \text{where} \quad \bar{x} = \frac{1}{n}\sum_{j=1}^n x_j.
\end{align}
(Theorem~\ref{omega and nu closed form thm} in Appendix~\ref{closed form section} provides fully explicit expressions for $\bm{v}$ and $\bm{w}$, for equidistant grids~$\mathbb{T}$.) 
\end{theorem}

The proof is detailed in Appendix~\ref{exist_unique_Nash}.

\begin{remark}\label{vwSpecialCases}
We observe the following special cases of Theorem~\ref{exuniqnasheqdiscrete}. In the symmetric case $x_1=\dots=x_n$, we have $\bm{\xi}^*_i=x_1\bm{v}$ for all $i$, whereas in the case $x_1+\dots+x_n=0$ of zero net supply, $\bm{\xi}^*_i=x_i\bm{w}$ for all $i$. Thus, $\bm{v}$ and $\bm{w}$ can be interpreted as the strategies for an agent with unit initial inventory in each of those cases.
\end{remark}

\section{Continuous-Time Equilibrium}\label{cont_time_section}

This section recalls the continuous-time setting with \emph{boundary block costs}. We refer to \cite[Section~2]{CampbellNutz.24} for further details and proofs.

\subsection{Model Specifications}

There are $n$ traders with inventory processes $X^{i}=(X^{i}_t)_{t\in[0,T]}$. An \emph{admissible inventory} $X^i$ is c\`adl\`ag, predictable, has (essentially) bounded total variation, and satisfies $X^i_{0-}=x_i\in\R$ and $X^i_T=0$.
The \emph{unaffected price} $S=(S_t)_{t\ge0}$ is a c\`adl\`ag local martingale with $\E[[S,S]_T]<\infty$. By risk-neutrality (see \cite[Proposition~2.2]{CampbellNutz.24} for a detailed proof), we may assume $S\equiv0$. As in the discrete-time model, trading generates transient impact $I=(I_t)_{t\ge0}$ with the Obizhaeva–Wang dynamics
\begin{align*}
  dI_t = -\rho I_t\,dt + \lambda \sum_{i=1}^n dX^i_t, 
  \qquad I_{0-}=0.
\end{align*}
Collecting agents' inventories in the vector $\bm X = (X^1,\dots,X^n)$ and setting $\lambda =1$ without loss of generality, the (affected) price process is
\begin{align*}
    S_t^{\bm X} = \int_0^t e^{-\rho(t-s)} \sum_{i=1}^n dX^i_s.
\end{align*}
In addition to the cost of transient impact, we charge quadratic \emph{boundary block costs} at $t=0$ and $t=T$ with coefficients $\vartheta_0,\vartheta_T\geq0$. Given opponents’ strategies $\boldsymbol{X}^{-i}$, the execution cost of $X^i$ is then defined as
\begin{equation}\label{cost_func_cont}
  \cC(X^i\mid\boldsymbol{X}^{-i})
  = \E\left[\int_0^T S^{\bm X}_{t-}\,dX^i_t + \frac12\sum_{t\in[0,T]}\Delta S_t\,\Delta X^i_t + \frac12\big(\vartheta_0(\Delta X^i_0)^2 + \vartheta_T(\Delta X^i_T)^2\big) \right].
\end{equation}
Thus, block trades at the initial and terminal times incur an additional quadratic cost governed by~$\vartheta_0$ and~$\vartheta_T$, respectively. The cross-term has the same interpretation as the discrete-time model. %

\subsection{Nash Equilibrium}

A profile $\boldsymbol{X}^*=(X^{*,1},\dots,X^{*,n})$ is a Nash equilibrium if each $X^{*,i}$ is admissible and
\begin{align*}
    \cC(Z\mid\boldsymbol{X}^{*,-i}) \ge \cC(X^{*,i}\mid\boldsymbol{X}^{*,-i})
    \qquad \text{for all admissible } Z.
\end{align*}

Existence of an equilibrium depends crucially on the initial and terminal block cost coefficients~$\vartheta_0$ and~$\vartheta_T$---there is a single choice yielding existence for general initial inventories. 

\begin{theorem}[Continuous  equilibrium, {\cite[Theorem~4.4, Corollary~4.6]{CampbellNutz.24}}]\label{thm:equil.block.cost}
Let $\vartheta_0=({n-1})/{2}$ and $\vartheta_T={1}/{2}$. Then the unique Nash equilibrium is given by
  \begin{align}\label{eq:contEquilib}
    X^{*,i}_t = \mathbbm{f}(t) (x_i-\bar{x}) + \mathbbm{g}(t) \bar{x},
    \qquad t\in[0,T],\ i=1,\dots,n,\qquad\text{where }\bar{x}=\frac1n\sum_{j=1}^nx_j%
  \end{align}
  and %
\begin{align}
    \mathbbm{f}(t) := \frac{\rho(T-t)+1}{\rho T+1},
    \quad t\in[0,T), 
    \qquad
    \mathbbm{f}_{0-}=1,
    \qquad
    \mathbbm{f}_T=0,\label{def_of_f_fundamental}\\
    \mathbbm{g}(t) := 1-\frac{n(\rho t +1)(n+1)e^{\rho\frac{n+1}{n-1}T}+2ne^{\rho\frac{n+1}{n-1}t}-(n-1)}
    {n((\rho T+1)(n+1)+2)e^{\rho\frac{n+1}{n-1}T}-(n-1)},
    \quad t\in[0,T],
    \qquad
    \mathbbm{g}_{0-}=1.\label{def_of_g_fundamental}    
\end{align}
Moreover, the equilibrium execution cost is given by 
\begin{align}\label{cost_cont_eq}
    \cC(X^{*,i}\mid\boldsymbol{X}^{*,-i}) = \mathscr{I} + \mathscr{B}_0 + \mathscr{B}_T,
\end{align}
where $\mathscr{I}$ is \emph{impact cost} 
    \begin{align}
        \mathscr{I} := \frac{n}{\rho T+1} \bar{x}(x_i-\bar{x})
  + \frac{\bar{x}^2 n^3 (n+1)  \left(  \left((\rho T + \frac{1}{2})(n+1) + 3\right) e^{\frac{2(n +1)\rho T}{n -1}} - \frac{2(n -1)}{n^2}   \left(n e^{\frac{(n +1)\rho T}{n -1}} + \frac{1}{4}\right)   \right)}
  {[n  ((\rho T + 1)(n+1) + 2)e^{\frac{(n +1)\rho T}{n -1}} - (n -1)]^{2}} \label{impact_cost_continuous_time}
    \end{align}
    and $\mathscr{B}_0,\mathscr{B}_T$ are the \emph{initial and terminal block costs}
    \begin{equation}\label{block_costs}
    \begin{aligned}
        \mathscr{B}_0 &:= \frac{(n - 1)(n + 1)^2 (1 + n e^{\rho \frac{n + 1}{n - 1} T})^2 \bar{x}^2}
        {4 [n((\rho T + 1)(n + 1) + 2)e^{\rho \frac{n + 1}{n - 1} T} - (n - 1)]^2}, \\[2pt] 
        \mathscr{B}_T &:= \frac{(x_i - \bar{x})^2}{4 (\rho T + 1)^2}.
    \end{aligned}
\end{equation}
\end{theorem}

\begin{remark}\label{non_existence_for_diff_theta_values}
The stated values for the block cost coefficients $\vartheta_0, \vartheta_T$ are the unique \say{correct} choice for this model. Indeed, for different values of $\vartheta_0, \vartheta_T$, equilibrium does not exist, except for particular initial inventories. Specifically, \cite[Theorem~4.4]{CampbellNutz.24} shows that if $\vartheta_T={1}/{2}$ but $\vartheta_0\neq({n-1})/{2}$, an equilibrium exists if and only if $\bar{x}=0$, and if $\vartheta_0=({n-1})/{2}$ but $\vartheta_T\neq{1}/{2}$, an equilibrium exists if and only if $x_1=\cdots=x_n$. Thus, if both $\vartheta_0\neq({n-1})/{2}$ and $\vartheta_T\neq{1}/{2}$, then the only case with equilibrium is $x_i\equiv0$, which yields the trivial no-trade solution $X^{*,i}\equiv0$. In the cases with existence, the equilibrium is given by~\eqref{eq:contEquilib}.
\end{remark}

\section{High-Frequency Limits}\label{hfl_section}

We can now present our main results on the high-frequency limits of the discrete equilibrium strategies and costs. In the case $\theta>0$ of non-zero instantaneous costs, we show that the discrete equilibria converge to the continuous-time equilibrium of Theorem~\ref{thm:equil.block.cost} including the particular boundary block costs. Whereas for $\theta=0$, the limit does not exist, and this will be linked to the non-existence of a continuous-time equilibrium when there are no boundary block costs (Remark~\ref{non_existence_for_diff_theta_values}). 
We fix initial inventories $({x_1,\dots,x_n})\in\mathbb{R}^n$ and consider equidistant grids
\begin{align}\label{time_grid_equidistant}
    \mathbb{T}_N := \{ kT/N \mid k = 0,1,\dots,N \}, \quad N = 2,3,\dots
\end{align}
For $t\in[0,T]$, define
\begin{equation}\label{fund_strat_rescaled}
    n_t = \lceil Nt/T \rceil, \qquad 
    V_t^{(N)} = 1 - \sum_{k=1}^{n_t} v_k, \qquad 
    W_t^{(N)} = 1 - \sum_{k=1}^{n_t} w_k,\qquad X^{(N),i}_t = \bar{x}V^{(N)}_t + (x_i-\bar x)W^{(N)}_t,
\end{equation}
where $\bm v$ and $\bm w$ are the vectors from \eqref{optimal sol}. 
Note that time $t$ corresponds to the $n_t$-th trading date in $\mathbb{T}_N$. In view of Remark~\ref{vwSpecialCases}, $V_t^{(N)}$ is the time-$t$ inventory of an agent with unit initial inventory in the symmetric case $x_1=\dots=x_n$. Similarly, $W_t^{(N)}$ is the time-$t$ inventory of an agent with unit initial inventory in the case of zero net supply. Finally, $X^{(N),i}_t$ is the time-$t$ inventory of agent $i$ with initial inventory $x_i$.

We first focus on the case $\theta>0$. The first result states the convergence of the strategies. More precisely, the time-$t$ inventory $X^{(N),i}_t$ converges pointwise to the continuous-time inventory $X^{*,i}_t$ of Theorem~\ref{thm:equil.block.cost} for $t\in(0,T)$, and we establish the rate $1/N$. Given the form of the strategies, convergence boils down to $V^{(N)}_t \to \mathbbm{g}(t)$ and $W^{(N)}_t \to \mathbbm{f}(t)$, where $\mathbbm{f}$ and $\mathbbm{g}$ are defined in~\eqref{def_of_f_fundamental} and~\eqref{def_of_g_fundamental}. At each of the boundaries ($t=0$ and $t=T$), one of these limits fails, whence the convergence of the strategies only on the open interval~$(0,T)$.

\begin{theorem}[Convergence of strategies for $\theta>0$]\label{strat_asympt_thm_theta_positive}
If $\theta>0$, we have 
\begin{align*}
    X^{(N),i}_t\lra X^{*,i}_t,\qquad\text{for any }t\in(0,T).
\end{align*}
More precisely:
    \begin{enumerate}[label=(\alph*)]
        \item\label{conv_of_V} For any $t\in(0,T]$, the sequence $N|V^{(N)}_t - \mathbbm{g}(t)|$ is bounded, and $V_0^{(N)}=1$ for all $N$.
        \item\label{conv_of_W} For any $t\in[0,T)$, the sequence $N|W^{(N)}_t - \mathbbm{f}(t)|$ is bounded, and $N|W_T^{(N)} - \frac{1}{(2\theta + \frac{1}{2})(\rho T + 1)}|=\mathcal{O}(1)$.
    \end{enumerate}
\end{theorem}

We emphasize that the limits are \emph{independent} of the specific value of $\theta>0$. 

\begin{figure}[hbp]
    \centering
    \includegraphics[width=0.54\linewidth]{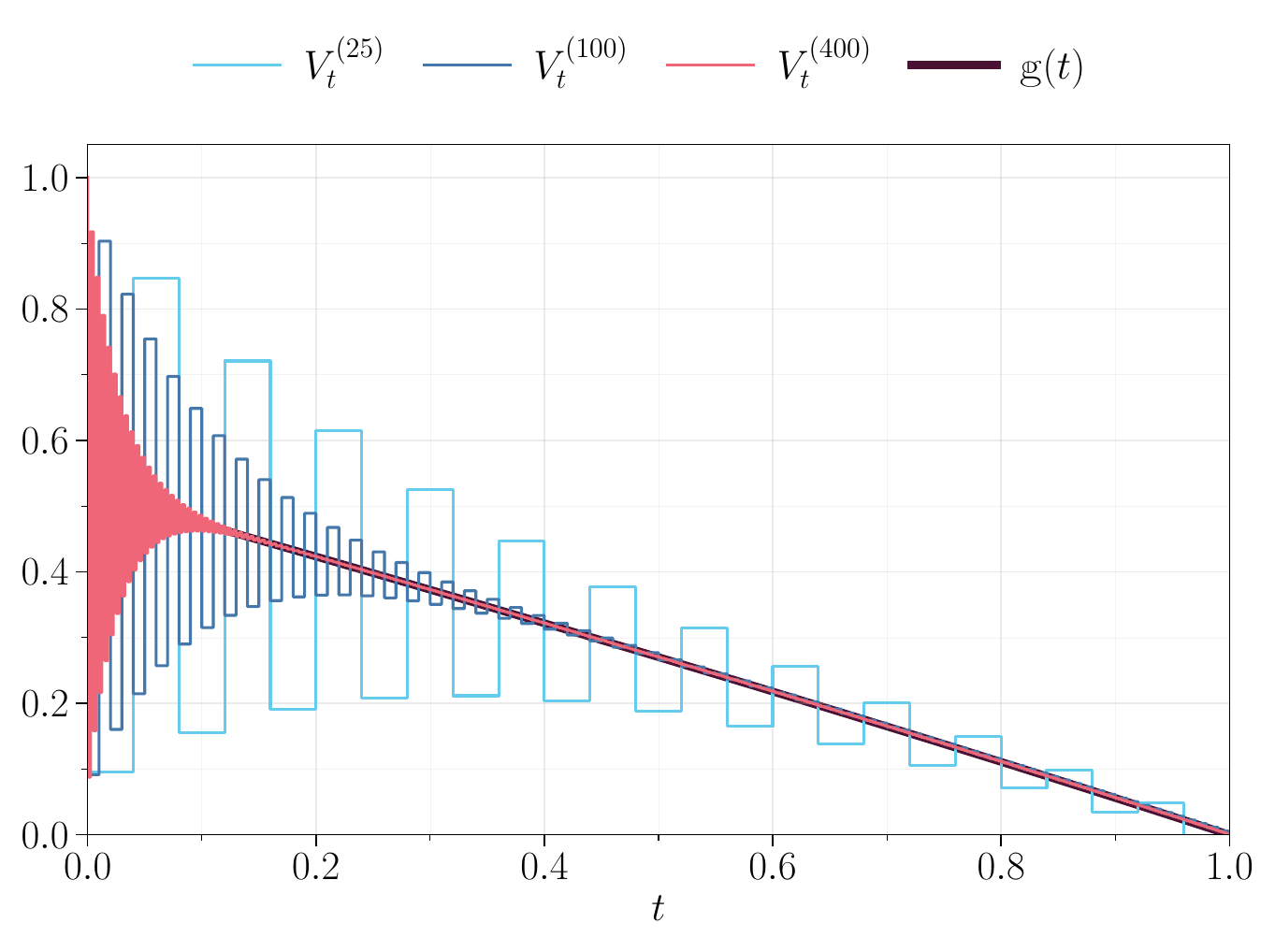} %
    \caption{Convergence of $V^{(N)}_t$ for $\theta=0.1$, $n=10$, and $\rho=1$.}
    \label{V_conv_img}
\end{figure}

\begin{figure}[htbp]
    \centering
    \includegraphics[width=0.54\linewidth]{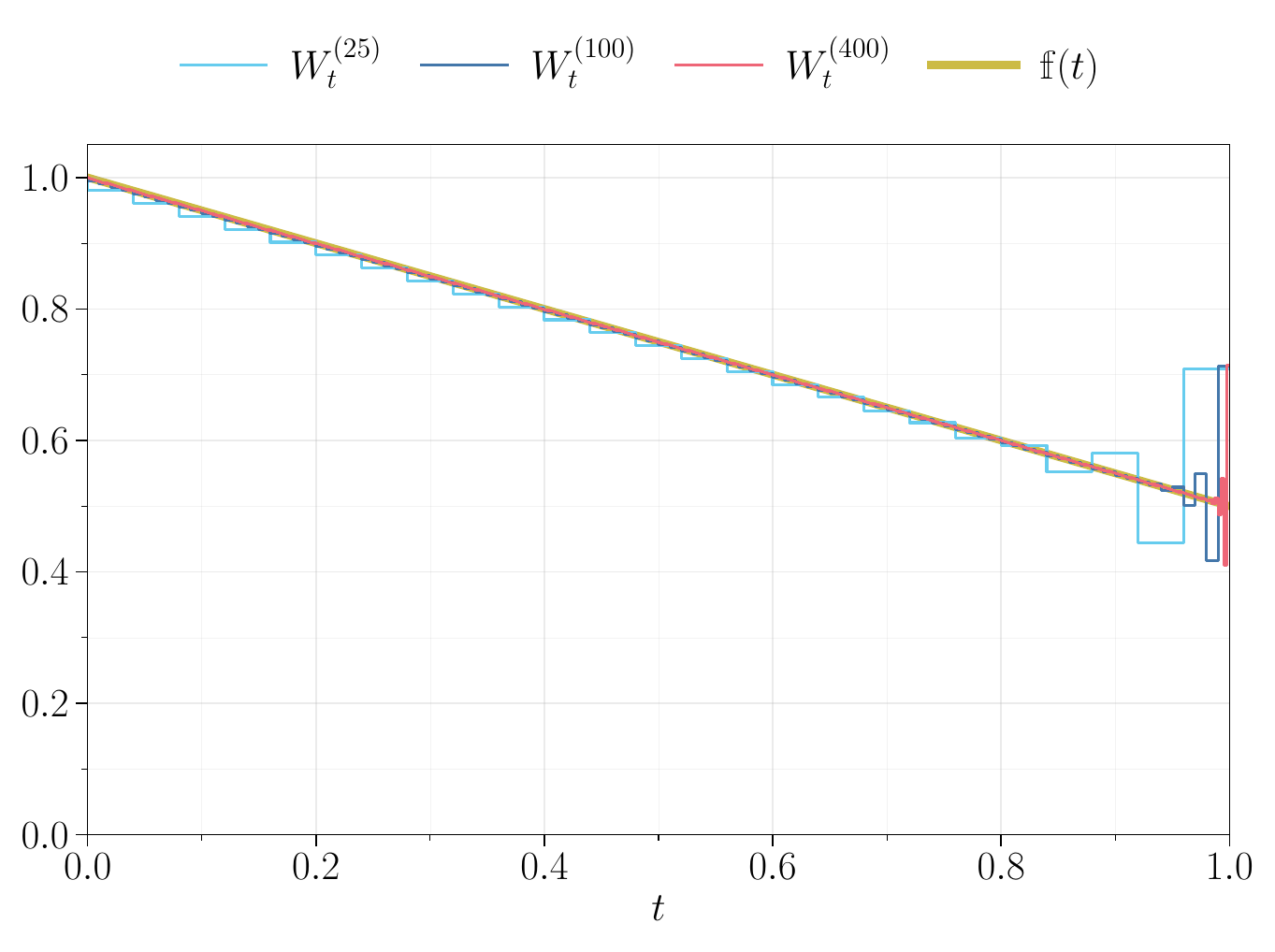}
    \caption{Convergence of $W^{(N)}_t$ for $\theta=0.1$, $n=10$, and $\rho=1$.}
    \label{W_conv_img}
\end{figure}

A similar conclusion holds for the costs. We show not only the convergence of the total cost, but also that the specific boundary block costs $\mathscr{B}_0$ and $\mathscr{B}_T$ arise as the limits of the instantaneous costs incurred near the boundaries $t=0$ and $t=T$, respectively. Hence, the coefficients $\vartheta_0$ and $\vartheta_T$ in Theorem~\ref{thm:equil.block.cost} arise naturally from the high-frequency limit, and they are canonical in that the limit does not depend on the value of~$\theta$ as long as $\theta>0$.

\begin{theorem}[Convergence of costs for $\theta>0$]\label{costs asymptotics thm}
    Let $\theta>0$ and let $({{\bm\xi}^*}^{({N})}_1,\dots,{\bm\xi^*}^{(N)}_n)\in\prod_{i=1}^n \mathscr{X}(x_i,\mathbb{T}_N)$ be the equilibrium strategies. 
The discrete execution cost converges to its continuous counterpart,
    \begin{align}\label{an}
        \lim_{N \uparrow \infty} \mathbb{E}\bigl[\mathscr{C}_{\mathbb{T}_N}({\bm\xi^*}^{(N)}_i \mid {\bm\xi^*}^{(N)}_{-i})\bigr]
        &= \cC(X^{*,i}\mid\boldsymbol{X}^{*,-i}) = \mathscr{I} + \mathscr{B}_0 + \mathscr{B}_T.
    \end{align}
    More precisely, for any split $m_N:=\lceil cN\rceil$ with $c\in({0,1})$, the cumulative initial/terminal instantaneous costs converge to the initial/terminal block costs of the continuous equilibrium,
    \begin{align}\label{inst_cost_to_block}
        \theta\sum_{k=0}^{m_N-1}\bigl({\xi^{*\,(N)}_{i,k}}\bigr)^2 \lra \mathscr{B}_0, \qquad \theta\sum_{k=m_N}^{N}\bigl({\xi^{*\,(N)}_{i,k}}\bigr)^2\lra \mathscr{B}_T,
    \end{align}
    and the discrete impact cost%
\begin{align*}
    \mathscr{I}_N({\bm\xi^*}_i^{(N)} \mid {\bm\xi^*}_{-i}^{(N)}) := \mathbb{E}\bigl[\mathscr{C}_{\mathbb{T}_N}({\bm\xi^*}_i^{(N)} \mid {\bm\xi^*}_{-i}^{(N)}) -\theta\sum_{k=0}^{N}\bigl(\xi^{*\,(N)}_{i,k}\bigr)^2\bigr]
\end{align*}
    converges to its continuous counterpart,
    \begin{align}\label{impact_cost_to_cont}
        \mathscr{I}_N({\bm\xi^*}^{(N)}_i \mid {\bm\xi^{*}}^{(N)}_{-i})\to \mathscr{I}.
    \end{align}
\end{theorem}

Figures~\ref{V_conv_img} and~\ref{W_conv_img} illustrate the persistent oscillations of the inventories $V^{(N)}$ and $W^{(N)}$ near the boundaries $t=0$ and $t=T$, whereas in the interior $({0,T})$ the jumps of the inventories are $\mathcal{O}({1/N})$; see Theorem~\ref{strat_asympt_thm_theta_positive}. The cumulative instantaneous costs of the oscillations near the boundaries tend to $\mathscr{B}_0$ and $\mathscr{B}_T$; see \eqref{inst_cost_to_block}. 

Theorems~\ref{strat_asympt_thm_theta_positive} and \ref{costs asymptotics thm} show convergence to a limit (independent of $\theta$) whenever $\theta>0$. By contrast, without instantaneous costs ($\theta=0$), the optimal strategies and the costs both oscillate and do not converge. The following theorems make this precise; for brevity, we use the constants $\mathfrak{d}_1,\mathfrak{d}_2, \mathfrak{a}_\pm(t), \mathfrak{b}(t),\mathfrak{c}(t)$ detailed in 
Table~\ref{tab:constants}. %

\begin{table}[!htbp]
\centering
\renewcommand{\arraystretch}{1.15}
\setlength{\extrarowheight}{2pt}
\begin{tabularx}{\linewidth}{|c|Y|}
\hline
\textbf{Constant} & \textbf{Definition} \\
\hline
$\mathfrak{d}_1$ &
\wrapmath{ n e^{2\frac{n+1}{n-1}\rho T}((n+1)\rho T + n+3) + (n-1)^2 e^{\frac{n+1}{n-1}\rho T} + (n+1)\rho T + 3n+1 }\\
\hline
$\mathfrak{d}_2$ &
\wrapmath{ n e^{2\frac{n+1}{n-1}\rho T}((n+1)\rho T + n+3) + (1-n^2) e^{\frac{n+1}{n-1}\rho T} - (n+1)\rho T - 3n-1 }\\
\hline
$\mathfrak{a}_\pm(t)$ &
\wrapmath{ \pm (n+1)e^{\frac{n+1}{n-1}\rho (T-t)} \pm n(n+1)e^{\frac{n+1}{n-1}\rho(2T-t)} }\\
\hline
$\mathfrak{b}(t)$ &
\wrapmath{ e^{2\frac{n+1}{n-1}\rho T}(n(n+1)\rho(T-t)+2n) - 2n e^{\frac{n+1}{n-1}\rho(T+t)} }\\
\hline
$\mathfrak{c}(t)$ &
\wrapmath{ (n+1)\rho(T-t) + n(n-1)e^{\frac{n+1}{n-1}\rho T} + 2n e^{\rho\frac{n+1}{n-1}t} + n+1 }\\
\hline
\end{tabularx}
\caption{Constants for oscillatory limits.}
\label{tab:constants}
\end{table}

\begin{theorem}[Divergence of strategies for $\theta = 0$]\label{strat_osc_thm}
    Assume $\theta = 0$.
    \begin{enumerate}[label=(\alph*)]
        \item\label{V_oscillations} Define the functions $\beta_\pm, \gamma_\pm : [0,T] \to \mathbb{R}$ by
            \[
                \beta_\pm(t) := \frac{\mathfrak{a}_\pm(t) + \mathfrak{b}(t) + \mathfrak{c}(t)}{\mathfrak{d}_1},\qquad
                \gamma_\pm(t) := \frac{\mathfrak{a}_\pm(t) + \mathfrak{b}(t) - \mathfrak{c}(t)}{\mathfrak{d}_2}.
            \]
            Then $V_0^{(N)} = 1$, and for $0 < t \le T$ the subsequence $(V_t^{(2N)})_{N \in \mathbb{N}}$ has exactly the two cluster points $\beta_+(t)$ and $\beta_-(t)$, while $(V_t^{(2N+1)})_{N \in \mathbb{N}}$ has exactly the two cluster points $\gamma_+(t)$ and $\gamma_-(t)$.
        \item\label{W_oscillations} Define the functions $\varphi_\pm, \psi_\pm : [0,T] \to \mathbb{R}$ by 
        \[
            \varphi_\pm(t) := \frac{1 + \rho(T - t) \pm e^{-\rho(T - t)}}{1 + \rho T + e^{-\rho T}},\qquad
            \psi_\pm(t) := \frac{1 + \rho(T - t) \pm e^{-\rho(T - t)}}{1 + \rho T - e^{-\rho T}}.
        \]
        Then $W_0^{(N)} = 1$, and for $0 < t < T$ the sequence $(W_t^{(2N)})_{N \in \mathbb{N}}$ has exactly the two cluster points $\varphi_+(t)$ and $\varphi_-(t)$, while $(W_t^{(2N+1)})_{N \in \mathbb{N}}$ has exactly the two cluster points $\psi_+(t)$ and $\psi_-(t)$. For $t=T$ we have $W_T^{(2N)}\to\varphi_+(T)$ and $W_T^{(2N+1)}\to\psi_+(T)$.
    \end{enumerate}
\end{theorem}

\begin{figure}[bthp]
  \begin{minipage}{.99\linewidth}   %
    \centering
    \begin{overpic}[width=0.5\linewidth]{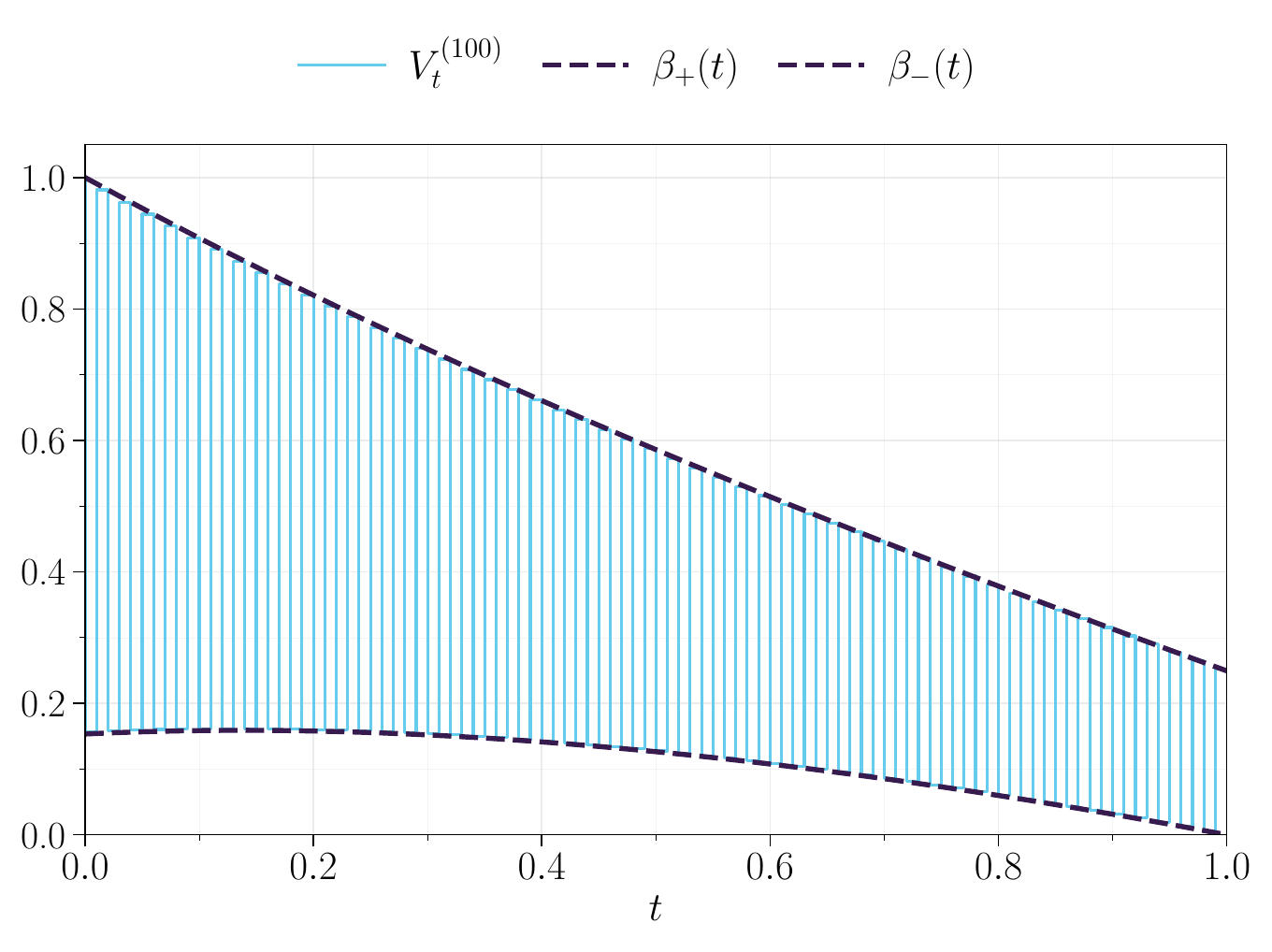}
    \end{overpic}
    \hspace{-10pt}
    \begin{overpic}[width=0.5\linewidth]{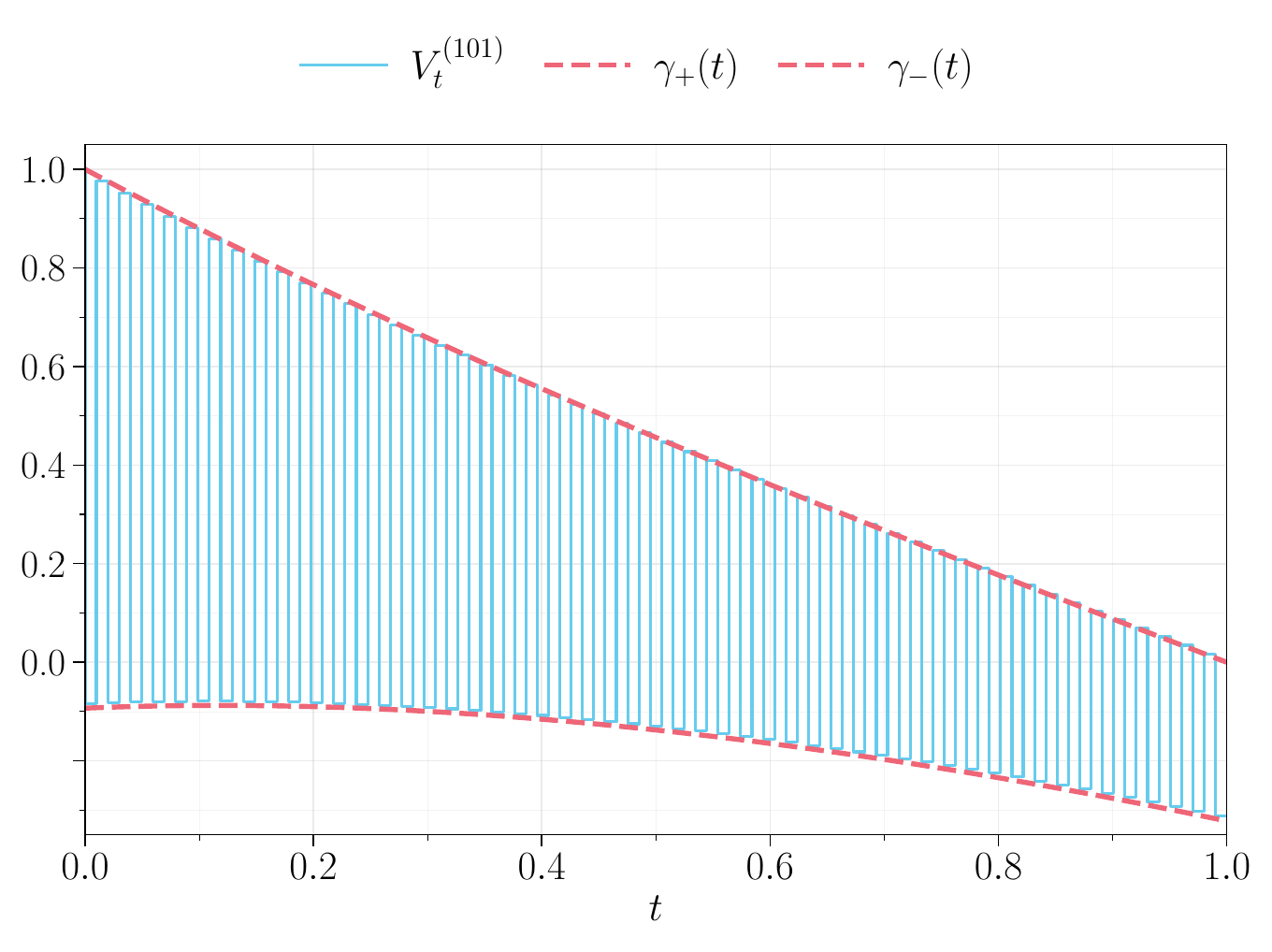}
    \end{overpic}
  \end{minipage}
  \caption{$V_t^{(100)}$ (left) and $V_t^{(101)}$ (right) for $n=10$, $\theta=0$, and $\rho=1$, together with the corresponding cluster points from Theorem~\ref{strat_osc_thm}\,\ref{V_oscillations}.}
\end{figure}

Likewise, the equilibrium costs oscillate when $\theta=0$, approaching different limits along subsequences of time grids with an even or odd number of steps.

\begin{theorem}[Divergence of costs for $\theta = 0$]\label{cost_asympt_thm_theta_zero}
    Using the same notation as in Theorem~\ref{costs asymptotics thm}, the equilibrium costs satisfy
    \begin{align*}
        \lim_{N \uparrow \infty} \mathbb{E}[\mathscr{C}_{\mathbb{T}_{2N}}({\bm\xi^*}_i^{(2N)} \mid {\bm\xi^*}_{-i}^{(2N)})]
        &= \frac{n \bar{x}^2 ((n + 1)n e^{2\rho \frac{n + 1}{n - 1}T} + n + 1)}{\mathfrak{d}_1}
        + \frac{n \bar{x} (x_i - \bar{x})}{e^{-\rho T} + \rho T + 1}
    \end{align*}
    and
    \begin{align*}
        \lim_{N \uparrow \infty} \mathbb{E}[\mathscr{C}_{\mathbb{T}_{2N+1}}({\bm\xi^*}_i^{(2N+1)} \mid {\bm\xi^*}_{-i}^{(2N+1)})]
        &= \frac{n \bar{x}^2 ((n + 1)n e^{2\rho \frac{n + 1}{n - 1}T} - n - 1)}{\mathfrak{d}_2}
        + \frac{n \bar{x} (x_i - \bar{x})}{\rho T + 1 - e^{-\rho T}}.
    \end{align*}
\end{theorem}

Recall that Remark~\ref{non_existence_for_diff_theta_values} identified two special configurations of the initial inventories where continuous-time equilibrium exists even though one of the two boundary block costs $\vartheta_0,\vartheta_T$ does not have the \say{correct} value. That phenomenon has no analogue in Theorems~\ref{strat_osc_thm} and~\ref{cost_asympt_thm_theta_zero}, which feature a single parameter $\theta$ for the entire time interval. Appendix~\ref{half-grid-block-costs} analyzes the behavior of the equilibrium inventories in a richer discrete-time model where the cost functional is modified by charging instantaneous costs only on either the first or the second half of the time interval; this amounts to a time-dependent~$\theta$. When the cost is charged only on the second half (and there is no cost on the first half), the discrete-time inventories converge to the continuous-time equilibrium in the zero-net-supply case, whereas when the cost is charged only on the first half, convergence holds in the symmetric case. Thus, for those special configurations of the initial inventories, convergence of the discrete-time models with costs on the first/second half is in one-to-one correspondence with the existence of a continuous-time equilibrium when the initial/terminal block cost is specified correctly. 

This completes the overall picture that emerges from the preceding theorems: any positive instantaneous costs give rise to the \say{correct} boundary block costs in the limit, whereas absence of instantaneous costs leads to failure of convergence, corresponding to non-existence of equilibrium in the continuous-time setting.

\FloatBarrier

\appendix

\section{Closed Form of the Discrete-Time Equilibrium}\label{closed form section}
The goal of this section is to obtain an explicit formula for the discrete-time equilibrium of Theorem~\ref{exuniqnasheqdiscrete}. For that, we only need to derive a formula for $\bm{v}$. The formula for $\bm{w}$ is the same as in~\cite{SchiedStrehleZhang.17} which treats the case of $n=2$ traders; indeed, by Remark~\ref{w_indep_n}, $\bm{w}$ does not depend on~$n$.  Recall the time grid $\bT_N$ in \eqref{time_grid_equidistant} and the matrices $\Gamma^\theta$, $\wt{\Gamma}$ in \eqref{gammatheta}.
Define the following column vectors of length $N+1$,
\begin{align}\label{optsolvecangprod}
    \bm{\nu} := ( \Gamma^{\theta} + (n-1)\wt{\Gamma} )^{-1}\bm{1},
    \qquad
    \bm{\omega} := ( \Gamma^{\theta} - \wt{\Gamma} )^{-1}\bm{1}.
\end{align}
Then, by \eqref{optimal sol},
\begin{align}\label{v and w}
    \bm{v} = \frac{1}{\bm{1}^\top \bm{\nu}} \bm{\nu},
    \qquad
    \bm{w} = \frac{1}{\bm{1}^\top \bm{\omega}} \bm{\omega}.
\end{align}
Generalizing \cite{SchiedStrehleZhang.17}, we set
\[
\alpha := e^{-\rho T/N},
\qquad
\kappa := 2\theta + ({n-1})/{2},\qquad \Gamma := \Gamma^0, 
\]
and introduce the matrix
\[
    B := (1-\alpha^2)( \Id + \Gamma^{-1}((n-1)\wt{\Gamma} + 2\theta\Id)).
\]
From \eqref{optsolvecangprod} we then have
\begin{align}\label{nu.reduction}
    \bm{\nu} = (\Gamma^\theta + (n-1)\wt{\Gamma})^{-1}\bm{1} = (1-\alpha^2)B^{-1}\Gamma^{-1}\bm{1}.
\end{align}
Moreover, the inverse of the Kac--Murdock--Szeg\H{o} matrix $\Gamma$ is the tridiagonal matrix
\begin{equation*}%
\Gamma^{-1} = \frac{1}{1-\alpha^{2}}\begin{pmatrix}
  1 & -\alpha & 0 & \cdots & \cdots & 0 \\
  -\alpha & 1+\alpha^2 & -\alpha & 0 & \cdots & 0 \\
  0 & \ddots & \ddots & \ddots & \ddots & \vdots \\
  \vdots & \ddots & \ddots & \ddots & \ddots & \vdots \\
  \vdots & \ddots & \ddots & -\alpha & 1+\alpha^2 & -\alpha \\
  0 & \cdots & \cdots & 0 & -\alpha & 1
 \end{pmatrix};
\end{equation*}
see, e.g., \cite[Section~7.2, Problems~12--13]{HornJohnson.85}. From this expression, we find that
\begin{align}\label{mi_serve_nella_proof}
    (1-\alpha^2)\Gamma^{-1}\bm{1} = (1-\alpha, (1-\alpha)^2, \dots, (1-\alpha)^2, 1-\alpha)^\top.
\end{align}
In view of~\eqref{nu.reduction}, our main task is then to determine $B^{-1}$. To that end, we first compute that
{\footnotesize\begin{align*}
B &= (1-\alpha^2)\Id \\[2pt]
&+ 
 \begin{pmatrix}
  1 & -\alpha & 0 & \cdots & \cdots & 0 \\
  -\alpha & 1+\alpha^2 & -\alpha & 0 & \cdots & 0 \\
  0 & \ddots & \ddots & \ddots & \ddots & \vdots \\
  \vdots & \ddots & \ddots & \ddots & \ddots & \vdots \\
  \vdots & \ddots & \ddots & -\alpha & 1+\alpha^2 & -\alpha \\
  0 & \cdots & \cdots & 0 & -\alpha & 1
 \end{pmatrix}
 \begin{pmatrix}
  \kappa & 0 & \cdots & \cdots & \cdots & 0 \\
  (n-1)\alpha & \kappa & 0 & \cdots & \cdots & 0 \\
  (n-1)\alpha^2 & (n-1)\alpha & \ddots & \ddots & \ddots & \vdots \\
  \vdots & \vdots &  &  & 0 & \vdots \\
  (n-1)\alpha^{N-1} & (n-1)\alpha^{N-2} & \ddots & \ddots & \kappa & 0 \\
  (n-1)\alpha^N & (n-1)\alpha^{N-1} & \cdots & \cdots & (n-1)\alpha & \kappa
 \end{pmatrix} \\
&= \begin{pmatrix}
     1 - n\alpha^2 + \kappa & -\alpha\kappa & 0 & \cdots & \cdots & 0 \\
     -\alpha(\kappa + 1 - n) & 1 + \alpha^2(\kappa - n) + \kappa & -\alpha\kappa & 0 & \cdots & 0 \\
     0 & \ddots & \ddots & \ddots & \ddots & \vdots \\
     \vdots & \ddots & \ddots & \ddots & \ddots & \vdots \\
     \vdots & \ddots & 0 & -\alpha(\kappa + 1 - n) & 1 + \alpha^2(\kappa - n) + \kappa & -\alpha\kappa \\
     0 & \cdots & \cdots & 0 & -\alpha(\kappa + 1 - n) & 1 - \alpha^2 + \kappa
 \end{pmatrix}.
\end{align*}}

\begin{lemma}\label{lemma1}

For $k \le N$, the $k$\emph{th} leading principal minor $\delta_k$ of $B$ is given by
\begin{align*}
    \delta_k &= c_+ m_+^k + c_- m_-^k,
\end{align*}
where, defining the real number
\begin{align*}
    R := \sqrt{\alpha^4(\kappa-n)^2 - 2\alpha^2(\kappa(\kappa+1) + n(1-\kappa)) + (\kappa + 1)^2},
\end{align*}
the real numbers $c_\pm$ and $m_\pm$ are given by
\begin{align*}
    c_\pm = \frac{ \pm (1 - \alpha^2(\kappa+n) + \kappa) + R}{2R},
    \qquad
    m_\pm = \frac{1 + \alpha^2(\kappa - n) + \kappa \pm R}{2}.
\end{align*}
\hide{Moreover,
\begin{equation}\label{d}
\begin{aligned}
\det B &= \delta_{N+1}\\ 
&= \left( 1 -\alpha^2 +\kappa \right)\left(c_+\left(m_+\right)^N+c_-\left(m_-\right)^N\right)\\
&\quad -\alpha^2\kappa \left( \kappa+1-n \right)\left(c_+\left(m_+\right)^{N-1}+c_-\left(m_-\right)^{N-1}\right)\\
&=\left( 1 -\alpha^2 +\kappa \right)\delta_N -\alpha^2\kappa \left( \kappa+1-n \right)\delta_{N-1}
\end{aligned}
\end{equation}}
\end{lemma}

\begin{proof}
    We have
    \begin{equation}\label{a}
        \delta_1 = 1 - n\alpha^2 + \kappa
    \end{equation}
    and
    \begin{equation}\label{b}
        \delta_2 = -n\alpha^4(\kappa-n) - n\alpha^2(\kappa + 2) + (\kappa+1)^2 .
    \end{equation}
    \hide{Indeed for $\delta_2$ we have
        \begin{eqnarray*}
            \delta_2&=&(1-\alpha^2n+\kappa)(1 +\alpha^2(\kappa -n) +\kappa)-\alpha\kappa(\alpha(\kappa +1-n))\\
            &=& (1-\alpha^2n+\kappa)(1 +\alpha^2\kappa -\alpha^2n +\kappa)-\alpha^2\kappa(\kappa +1-n)\\
            &=&1+\cancel{\alpha^2\kappa} - \alpha^2n+\kappa-\alpha^2n-\alpha^4n\kappa+\alpha^4n^2-\cancel{\alpha^2n\kappa} + \kappa + \cancel{\alpha^2\kappa^2}-\alpha^2n\kappa+\kappa^2 - \cancel{\alpha^2\kappa^2}-\cancel{\alpha^2\kappa}+\cancel{\alpha^2n\kappa}\\
            &=& 1  - \alpha^2n+\kappa-\alpha^2n-\alpha^4n\kappa+\alpha^4n^2 + \kappa -\alpha^2n\kappa+\kappa^2\\
            &=& -n\alpha^4(\kappa-n) -n\alpha^2(\kappa+2)   + (\kappa+1)^2
        \end{eqnarray*}}
    For $k\in\{3,\dots,N\}$, the $k$th principal minor $\delta_k$ satisfies the recursion
    \begin{align*}
        \delta_k
        = (1 + \alpha^2(\kappa - n) + \kappa)\delta_{k-1}
          - \alpha^2 \kappa(\kappa + 1 - n)\delta_{k-2}.
    \end{align*}
    This is a homogeneous linear difference equation of second-order. Its characteristic equation is
    \begin{align}\label{careqprincmin}
        m^2 - (1 + \alpha^2(\kappa - n) + \kappa)m
        + \alpha^2\kappa(\kappa + 1 - n) = 0.
    \end{align}
    \hide{Let now 
    \begin{eqnarray*}
        R&:=&\sqrt{(1 +\alpha^2(\kappa -n) +\kappa)^2-4\alpha^2\kappa(\kappa+1-n)}\\
        &=&\sqrt{\alpha^4(\kappa-n)^2 -2\alpha^2(\kappa(\kappa+1) + n(1-\kappa)) + (\kappa +1)^2},
    \end{eqnarray*}}
    The roots of \eqref{careqprincmin} are 
    \[
        m_{\pm} = \frac{1 + \alpha^2(\kappa - n) + \kappa \pm R}{2}.
    \]
    We first claim that $m_+$ and $m_-$ are real for $\alpha\in[0,1]$ and $\kappa\ge \frac{n-1}{2}$. This is equivalent to the nonnegativity of the radicand in $R$, which in turn is equivalent to
    \[
        f(t):=t^2(\kappa-n)^2 - 2t(\kappa(\kappa+1) + n(1-\kappa)) + (\kappa+1)^2 \ge 0,
        \quad 0\le t\le 1,
    \]
    after setting $t=\alpha^2$. The claim is clear for $\kappa=n$ since
    \[
        -2t(n^2 + 2n - n^2) + n^2 + 2n + 1
        = n^2 + 1 + 2n(1-2t)
        \ge (n-1)^2 \ge 0 .
    \]
    Otherwise, $f$ is minimized at
    \[
        t_0 := \frac{\kappa(\kappa+1) + n(1-\kappa)}{(\kappa-n)^2}.
    \]
    We have $t_0<1$ iff $\kappa<\frac{n^2-2}{n+1}$. In this case, $f(t)\ge f(t_0)
    = \frac{-(\kappa(\kappa+1) + n(1-\kappa))^2 + (\kappa-n)^2 (\kappa+1)^2}{(\kappa-n)^2}>0$ for all $t$. \hide{Indded, we need to show that
    \begin{eqnarray*}
        (\kappa-n)^2(\kappa+1)^2&>&(\kappa(\kappa+1) + n(1-\kappa))^2\\
        (\kappa-n)^2(\kappa+1)^2&>&(\kappa^2+\kappa + n-n\kappa)^2\\
        (\kappa-n)^2(\kappa^2 + 1 + 2\kappa)&>&(\kappa(\kappa-n) + (\kappa +n))^2\\
        (\kappa-n)^2(\cancel{\kappa^2} + 1 + 2\kappa)&>&\cancel{\kappa^2}(\kappa-n)^2 + (\kappa+n)^2 + 2\kappa(\kappa^2-n^2)\\
        (\kappa-n)^2 + 2\kappa(\cancel{\kappa^2} + 2n^2 -2\kappa n)&>&(\kappa+n)^2 + 2\kappa\cancel{\kappa^2}\\
        \cancel{\kappa^2}+\cancel{n^2}-2\kappa n + 2\kappa(2n^2 -2\kappa n)&>&\cancel{\kappa^2}+\cancel{n^2}+2\kappa n\\
        n &>&\kappa +1   
    \end{eqnarray*}
    Now just compute
    \begin{eqnarray*}
        n &>& (n^2-1 + n)/(n+1)\\
        n^2 + n &>& n^2+n-1\\
        0&>&-1
    \end{eqnarray*}
    now since by assumption $(n^2-1 + n)/(n+1)>\kappa +1$, we have our claim.
    } For $\kappa\ge \frac{n^2-2}{n+1}$, we have $t_0\ge 1$ and, in turn, $f'(t)\le 0$ for $0\le t\le 1$, so $f(t)\ge f(1)=(n-1)^2>0$. This proves that $m_\pm$ are real. \hide{Indeed, 
    \begin{eqnarray*}
        f'(t)&=&2t(\kappa-n)^2 - 2(\kappa(\kappa+1) + n(1-\kappa))\\
        &\le&2(\kappa-n)^2 - 2(\kappa(\kappa+1) + n(1-\kappa))
    \end{eqnarray*}
    Now we just need to prove that
    \begin{eqnarray*}
        (\kappa-n)^2 - (\kappa(\kappa+1) + n(1-\kappa))&\le&0\\
        \cancel{\kappa^2} + n^2 - 2n\kappa &\le& \cancel{\kappa^2} + \kappa + n -n\kappa\\
        n^2-n\kappa-\kappa-n&\le&0\\
    \end{eqnarray*}
    Now by assumption $-\kappa\le -(n^2-2)/(n+1)$, thus we get
    \begin{eqnarray*}
        n^2 -n (n^2-2)/(n+1) -(n^2-2)/(n+1) -n&\le&0\\
        n^2(n+1) - n (n^2-2)-(n^2-2)-n(n+1)&\le&0\\
        \cancel{n^3} + \cancel{n^2} - \cancel{n^3} +2n-\cancel{n^2}+2-n^2-n&\le&0\\
        -n^2 +n +2&\le&0\\
        2+n&\le&n^2\\
        2&\le&n
    \end{eqnarray*}
    that gives us our claim} 

    By the theory of second-order linear difference equations, every solution of \eqref{careqprincmin} has the form
    \(
      c_1 m_+^k + c_2 m_-^k
    \)
    with real constants \(c_1,c_2\); see \cite[Theorem~3.7]{KelleyPeterson.01}.
    Imposing the initial conditions \eqref{a}--\eqref{b} yields \(c_1=c_+\) and \(c_2=c_-\), as stated.
\end{proof}

\begin{lemma}\label{lemma2}
    Define the sequence $\phi_k$ recursively by
    \[
        \phi_{N+2} = 1, \qquad \phi_{N+1} = 1-\alpha^2+\kappa,
    \]
    and, for $k = N, N-1, \dots, 2$,
    \[
        \phi_k
    = (1+\alpha^2(\kappa-n)+\kappa)\phi_{k+1}
      - \alpha^2\kappa(\kappa+1-n)\phi_{k+2}.
    \]
    Then, for $k \in \{2,\dots, N+2\}$,
    \[
        \phi_k = d_+ m_+^{N+2-k} + d_- m_-^{N+2-k},
    \]
    where $m_\pm$ are as in Lemma~\ref{lemma1} and
    \[
        d_\pm := \frac{\pm(1+(1-\alpha^2)\kappa - \alpha^2(2-n)) + R}{2R}.
    \]
\end{lemma}

\begin{proof}
    Let
    \begin{align}\label{g}
        \psi_0 = 1, \qquad \psi_1 = 1-\alpha^2+\kappa,
    \end{align}
    and, for $l \in \{2,\dots,N\}$, set
    \begin{align}\label{f}
        \psi_l
    = (1+\alpha^2(\kappa-n)+\kappa)\psi_{l-1}
      - \alpha^2\kappa(\kappa+1-n)\psi_{l-2}.
    \end{align}
    Then $\psi_k = \phi_{N+2-k}$. As in the proof of Lemma~\ref{lemma1}, the general solution to \eqref{f} is $d_1 m_+^{l} + d_2 m_-^{l}$ with $m_\pm$ as above. Choosing $d_1=d_+$ and $d_2=d_-$ satisfies the initial conditions \eqref{g} and completes the proof.
\end{proof}

\begin{lemma}\label{B inverse lemma}
    The matrix $B$ is nonsingular and its inverse is
    \begin{align}\label{B inverse eq}
        (B^{-1})_{ij}
        &=
        \begin{cases}
            (\alpha\kappa)^{j-i}\delta_{i-1}\phi_{j+1}\delta_{N+1}^{-1}, & \text{if } i\le j,\\
            (\alpha(\kappa+1-n))^{i-j}\delta_{j-1}\phi_{i+1}\delta_{N+1}^{-1}, & \text{if } i\ge j,
        \end{cases}
    \end{align}
    where $\delta_0=1$ and $\delta_{N+1} = \det B$.
\end{lemma}

\begin{proof}
    We have shown in Lemma~\ref{matrices positive definite} that both $\Gamma$ and $\Gamma^\theta + (n-1)\wt{\Gamma}$ are invertible. Thus
    \[
        B = (1-\alpha^2)\Gamma^{-1}(\Gamma^\theta+(n-1)\wt{\Gamma})
    \]
    is also invertible. Hence $\delta_{N+1}\neq 0$, so the right–hand side of \eqref{B inverse eq} is well defined.
    In view of Lemmas~\ref{lemma1} and \ref{lemma2}, the explicit form follows from Usmani’s formula for the inverse of a tridiagonal Jacobi matrix \cite{Usmani.94}.
\end{proof}

\begin{theorem}[Explicit form of $\bm\omega$ and $\bm\nu$]\label{omega and nu closed form thm}
Recall that the vectors $\bm{v}$ and $\bm{w}$ of the discrete-time equilibrium in Theorem~\ref{exuniqnasheqdiscrete} have been written as 
\begin{align*}%
    \bm{v} = \frac{1}{\bm{1}^\top \bm{\nu}} \bm{\nu},
    \qquad
    \bm{w} = \frac{1}{\bm{1}^\top \bm{\omega}} \bm{\omega}.
\end{align*}
Let $\tilde{\kappa} = 2\theta + \frac{1}{2}$. Then the components of $\bm\omega$ are
\begin{align}\label{omi formula}
    \omega_i
    &= \frac{(1-\alpha)\tilde{\kappa}
        + \alpha\left(\frac{\alpha(\tilde{\kappa}-1)}{\tilde{\kappa}}\right)^{N+1-i}}
        {\tilde{\kappa}(\tilde{\kappa}-\alpha(\tilde{\kappa}-1))},
    \qquad i \in \{1,\dots,N+1\},
\end{align}
and in particular $\omega_{N+1} = 1/\tilde{\kappa}$.
The components of $\bm\nu$ are 
\begin{align*}
    \nu_1 &= \frac{1-\alpha}{\delta_{N+1}}
            \Bigg(\phi_2
                  + (1-\alpha)\sum_{j=2}^N (\alpha\kappa)^{j-1}\phi_{j+1}
                  + (\alpha\kappa)^{N}\Bigg),\\
    \nu_{N+1} &= \frac{1-\alpha}{\delta_{N+1}}
            \Bigg((\alpha(\kappa+1-n))^{N}
                  + (1-\alpha)\sum_{j=2}^N (\alpha(\kappa+1-n))^{N+1-j}\delta_{j-1}
                  + \delta_N\Bigg),
\end{align*}
and, for $i=2,\dots,N$,
\begin{equation*}
\begin{split}
    \nu_i &= \frac{1-\alpha}{\delta_{N+1}}
            \Bigg((\alpha(\kappa+1-n))^{i-1}\phi_{i+1}
                  + (1-\alpha)\sum_{j=2}^{i-1} (\alpha(\kappa+1-n))^{i-j}\delta_{j-1}\phi_{i+1}\\
          &\qquad\qquad\qquad\qquad\qquad\quad
                  + (1-\alpha)\sum_{j=i}^{N} (\alpha\kappa)^{j-i}\delta_{i-1}\phi_{j+1}+ (\alpha\kappa)^{N+1-i}\delta_{i-1}\Bigg).
\end{split}
\end{equation*}
\end{theorem}

\begin{proof} 
    The representation \eqref{omi formula} for the components of $\bm\omega$ was proved in \cite[Equation~(16)]{SchiedZhang.19} in the case $n=2$. (Note that our vector $\bm\omega$ is denoted by $\bm u$ in \cite{SchiedZhang.19}, our $\alpha$ corresponds to $a^{1/N}$ there, and we have $\lambda=1$ here.) By Remark~\ref{w_indep_n}, $\bm\omega$ does not depend on $n$, so the same formula holds for any~$n$. For $\bm\nu$, recall from \eqref{nu.reduction} that
    \[
        \bm{\nu}
        = (\Gamma^\theta + (n-1)\wt{\Gamma})^{-1}\bm{1}
        = (1-\alpha^2)B^{-1}\Gamma^{-1}\bm{1}.
    \]
    Using the explicit expression for $(1-\alpha^2)\Gamma^{-1}\bm{1}$ from~\eqref{mi_serve_nella_proof} and the formula for $B^{-1}$ in Lemma~\ref{B inverse lemma}, we obtain the stated formulas for the components of~$\bm\nu$.
\end{proof}

\section{Proofs for Section~\ref{The_Problem}}\label{exist_unique_Nash}

We first show uniqueness. 
\hide{
but before the proof we recall some results from \cite{CampbellNutz.24}. First of all, we recall that the cost-functional $\cC$ of the continuous-time defined in \eqref{cost_func_cont} is a particular case of a more general class of functional. In the same setting of Section~\ref{cont_time_section}, let us charge a (deterministic but possibly time-dependent) cost $\vartheta_t>0$ on block trades, the resulting functional takes the form
\begin{align*}
    \cC(X^i\mid\boldsymbol{X}^{-i})
  = \E\left[\int_0^T S^{\bm X}_{t-}\,dX^i_t + \frac12\sum_{t\in[0,T]}\Delta S_t\,\Delta X^i_t + \frac12\sum_{t\in[0,T]}\vartheta_t(\Delta X^i_t)^2 \right].
\end{align*}
In particular, \cite[Proposition 2.4]{CampbellNutz.24} guarantees uniqueness of the Nasj equilibrium for this probem for any choice of the time-dependent block cost parameter $\vartheta$
By \cite[Proposition 2.2]{CampbellNutz.24}, $\overline\cC$ has the form
\begin{align*}
    \bar\cC(X^i\mid\boldsymbol{X}^{-i})= &\mathbb{E}\left[\frac{1}{2}\int_0^T\int_0^T e^{-\rho|t-s|}dX_s^{i}dX_t^{i} +\int_0^T\int_0^{t-} e^{-\rho(t-s)}\sum_{j\neq i}dX_s^{j}dX_t^{i}+\frac{1}{2}\sum_{j\neq i}\sum_{t\in[0,T]}\Delta X^j_t\Delta X^i_t\right]\\
&\qquad+\frac{1}{2}\,\mathbb{E}\!\left[\sum_{t\in[0,T]}\vartheta_t(\Delta X_t^{i})^2\right]
\end{align*}
\cite[Proposition 2.4]{CampbellNutz.24}
}

\begin{lemma}\label{le:discreteUniqueness}
For a given time grid $\mathbb{T}$ and initial inventories $(x_1,\dots,x_n)$, there exists at most one Nash equilibrium in the class $\prod_i\mathscr{X}(x_i,\mathbb{T})$.  
\end{lemma}

\begin{proof}
  This is a special case of the uniqueness result stated in \cite[Theorem~5.1]{CampbellNutz.25} for a general class of models. To embed the present discrete-time model in that continuous-time setting, we specify an infinite cost for strategies acting outside the grid~$\bT$; cf.\ \cite[Section~2.3]{CampbellNutz.25}.
\end{proof}
Next, we reduce the existence proof to the class
\begin{align*}
    \mathscr{X}_{det}(x,\bT):=\{ \bm\xi\in\mathscr{X}(x,\bT) \mid \bm\xi\text{ is deterministic} \}
\end{align*}
of deterministic strategies. A Nash equilibrium in the class $\cX_{\rm det}(x_1,\bT)\times\dots\times\cX_{\rm det}(x_n,\bT)$ is defined in the same way as in Definition~\ref{Nasheqdiscdef} and called a deterministic Nash equilibrium.

\begin{lemma}\label{detstratonlydisc}
    A Nash equilibrium in the class $\cX_{\rm det}(x_1,\bT)\times\dots\times\cX_{\rm det}(x_n,\bT)$ of deterministic strategies is also a Nash equilibrium in the class $\cX(x_1,\bT)\times\dots\times\cX(x_n,\bT)$ of adapted strategies.
\end{lemma}

\begin{proof}
    We follow \cite[Lemma 3.4]{SchiedZhang.19}. Assume that $(\bm\xi_1^*,\dots,\bm\xi_n^*)$ is a Nash equilibrium in the class $\cX_{\rm det}(x_1,\bT)\times\dots\times\cX_{\rm det}(x_n,\bT)$ of deterministic strategies. We need to show that $\bm\xi_i^*$ minimizes $\E[\cC_\bT(\bm\xi  \lvert  \bm\xi^*_{-i})]$ over $\cX(x_i,\bT)$, for any~$i$. To this end, fix $\bm\xi\in\cX(x_i,\bT)$ and define $\overline{\bm\xi}\in\cX_{\rm det}(x_i,\bT)$ by $\overline{\xi}_k=\E[\xi_k]$ for $k=0,1,\dots,N$.
    Applying Jensen’s inequality to the convex map $\mathbb{R}^{N+1}\ni\bm x\mapsto\bm x^\top\Gamma^\theta\bm x$ gives
    \begin{align*}
        \E[\mathscr{C}_{\mathbb{T}}(\bm{\xi} \mid \bm{\xi}^*_{-i})]
        &= \E\Big[\frac{1}{2}\bm\xi^\top\Gamma^\theta\bm\xi + \bm\xi^\top\widetilde{\Gamma}\sum_{j\neq i}\bm\xi^*_j\Big] 
        = \E\Big[\frac{1}{2}\bm\xi^\top\Gamma^\theta\bm\xi\Big] + \overline{\bm\xi}^\top\widetilde{\Gamma}\sum_{j\neq i}\bm\xi^*_j\\
        &\ge \frac{1}{2}\overline{\bm\xi}^\top\Gamma^\theta\overline{\bm\xi} + \overline{\bm\xi}^\top\widetilde{\Gamma}\sum_{j\neq i}\bm\xi^*_j
        = \E[\mathscr{C}_{\mathbb{T}}(\overline{\bm{\xi}} \mid \bm{\xi}^*_{-i})]\\
        &\ge \E[\mathscr{C}_{\mathbb{T}}(\bm{\xi}_i^* \mid \bm{\xi}^*_{-i})],
    \end{align*}
    showing that $\bm{\xi}_i^*$ minimizes $\E[\mathscr{C}_{\mathbb{T}}(\bm{\xi} \mid \bm{\xi}^*_{-i})]$ over $\bm\xi\in\cX(x_i,\bT)$. %
\end{proof}

We can now establish the main theorem on the discrete-time equilibrium.

\begin{proof}[Proof of Theorem \ref{exuniqnasheqdiscrete}]
    We adapt \cite[Theorem 2.4]{LuoSchied.19}. Recall that uniqueness was shown in Lemma~\ref{le:discreteUniqueness}. By Lemma \ref{detstratonlydisc}, it then suffices to show that the strategies stated in Theorem~\ref{exuniqnasheqdiscrete} form a deterministic Nash equilibrium. In view of Lemma \ref{formoffunctionaldisc}, we thus need to show that
    \[
        \E[\cC_\bT(\bm\xi^*_i\mid\bm\xi^*_{-i})]
        = \min_{\bm{\xi}_i \in \mathscr{X}_{\mathrm{det}}(x_i, \mathbb{T})}
          \Big( \frac{1}{2}\bm{\xi}_i^{\top}\Gamma^{\theta}\bm{\xi}_i
               + \bm{\xi}_i^{\top}\widetilde{\Gamma} \sum_{j\neq i}\bm{\xi}^*_j \Big).
    \]
    The constraint $\bm{\xi}_i\in \mathscr{X}_{\mathrm{det}}(x_i, \mathbb{T})$ is the linear equality $\bm{1}^{\top}\bm{\xi}_i=x_i$. By Lagrange multiplier theory, a necessary condition for $(\bm\xi^*_1,\dots,\bm\xi^*_n)$ to be a deterministic Nash equilibrium is the existence of $\alpha_i\in\mathbb{R}$, $i=1,\dots,n$, such that
    \begin{equation}\label{pf:optimalityconditions}
        \left\{
        \begin{array}{l}
              \Gamma^{\theta}\bm{\xi}_i^* + \widetilde{\Gamma} \displaystyle\sum_{j\neq i}\bm{\xi}_j^* = \alpha_i \bm{1}, \\
              \bm{1}^{\top}\bm{\xi}_i^* = x_i .
        \end{array} 
        \right.
    \end{equation}
    We will show below that these equations are also sufficient for our optimization problem. Summing the first line of \eqref{pf:optimalityconditions} over $i$ yields
    \[
        (\Gamma^{\theta} + (n-1)\widetilde{\Gamma})\sum_{j=1}^{n}\bm{\xi}_j^*
        = \Big(\sum_{j=1}^{n}\alpha_j\Big)\bm{1}.
    \]
    By Lemma~\ref{matrices positive definite}, $\Gamma^{\theta}+ (n-1)\widetilde{\Gamma}$ is invertible. Hence
    \begin{equation}\label{pf:maineq1}
        \begin{split}
            \sum^{n}_{j=1}\bm{\xi}_j^* 
            & = \sum^{n}_{j=1}\alpha_j (\Gamma^{\theta} + (n-1) \widetilde{\Gamma})^{-1}\bm{1} \\
            & = \frac{\bm{1}^{\top} \sum^{n}_{j=1}\alpha_j (\Gamma^{\theta}+ (n-1) \widetilde{\Gamma})^{-1}\bm{1}}{\bm{1}^{\top}(\Gamma^{ \theta}+ (n-1) \widetilde{\Gamma})^{-1}\bm{1}} (\Gamma^{ \theta}+ (n-1) \widetilde{\Gamma})^{-1}\bm{1} \\
            & =  \sum^{n}_{j=1}\frac{\bm{1}^{\top}\bm{\xi}_j^*}{\bm{1}^{\top}(\Gamma^{ \theta}+ (n-1) \widetilde{\Gamma})^{-1}\bm{1}} (\Gamma^{ \theta}+ (n-1) \widetilde{\Gamma})^{-1}\bm{1}\\
            &= \sum^{n}_{j=1} x_j \bm{v},
        \end{split}
    \end{equation}
    using the second line of \eqref{pf:optimalityconditions} in the last step. Next, take the $i$th equation in \eqref{pf:optimalityconditions}, multiply by $n-1$, and subtract the sum of the remaining $n-1$ equations. This gives
        \[
        \Gamma^{ \theta} \Big( (n-1)\bm{\xi}_i^* - \sum_{j\neq i}\bm{\xi}_j^* \Big)  - \widetilde{\Gamma}  \Big( (n-1)\bm{\xi}_i^* - \sum_{j\neq i}\bm{\xi}_j^* \Big)  = \Big( (n-1)\alpha_i - \sum_{j\neq i}\alpha_j \Big) \bm{1}.
        \]
    Further simplifications show that
        \[
        (\Gamma^{ \theta} - \widetilde{\Gamma})\Big( n\bm{\xi}_i^* - \sum^{n}_{j=1}\bm{\xi}_j^* \Big) = \Big( n\alpha_i - \sum^{n}_{j=1}\alpha_j \Big) \bm{1}.
        \]
    Since $\Gamma^{\theta}-\widetilde{\Gamma}$ is invertible (Lemma~\ref{matrices positive definite}), we obtain
    \begin{equation}\label{pf:maineq2}
        n\bm{\xi}_i^*-\sum_{j=1}^{n}\bm{\xi}_j^* = \Big(nx_i-\sum_{j=1}^{n}x_j\Big)\bm{w}.
    \end{equation}
    Combining \eqref{pf:maineq1} and \eqref{pf:maineq2} yields
    \[
        \bm{\xi}_i^*  =  \bar{x} \bm{v} + (x_i-\bar{x}) \bm{w}.
    \]
    It remains to show that \eqref{pf:optimalityconditions} is sufficient. Write
    \[
        \frac{1}{2}{\bm{\xi}_i^*}^{\top}\Gamma^{\theta}\bm{\xi}_i
        + {\bm{\xi}_i^*}^{\top}\widetilde{\Gamma} \sum_{j\neq i}\bm{\xi}^*_j
        = \frac{1}{2}{\bm{\xi}_i^*}^{\top}\Gamma^{\theta}\bm{\xi}_i^* + \bm{g}_i^{\top}\bm{\xi}^*_i,
        \qquad \bm{g}_i := \widetilde{\Gamma}\sum_{j\neq i}\bm{\xi}^*_j.
    \]
    For any $\bm{\eta}_i\in\mathscr{X}_{\mathrm{det}}(x_i,\mathbb{T})$, using \eqref{pf:optimalityconditions} and the symmetry of $\Gamma^\theta$,
    \begin{equation}\label{optimality check eq}
        \begin{split}
            \frac{1}{2}\bm{\eta_i}^{\top}\Gamma^{\theta}\bm{\eta_i} + \bm{g_i}^{\top}\bm{\eta_i} -  \frac{1}{2}{\bm{\xi}_i^*}^{\top}\Gamma^{  \theta}{\bm{\xi}_i^*} - \bm{g_i}^{\top}{\bm{\xi}_i^*} 
            & = \frac{1}{2}(\bm{\eta_i} + {\bm{\xi}_i^*})^{\top}\Gamma^{  \theta}(\bm{\eta_i} - {\bm{\xi}_i^*}) + \bm{g_i}^{\top}(\bm{\eta_i} - {\bm{\xi}_i^*}) \\
            & = \Big(\frac{1}{2} (\Gamma^{  \theta})^{\top}(\bm{\eta_i} + {\bm{\xi}_i^*}) +  \bm{g_i}\Big)^{\top}(\bm{\eta_i} - {\bm{\xi}_i^*}) \\
            & = \Big((\Gamma^{  \theta}{\bm{\xi}_i^*} +  \bm{g_i}) + \frac{1}{2}(\Gamma^{  \theta})^\top(\bm{\eta_i} - {\bm{\xi}_i^*}) \Big)^{\top}(\bm{\eta_i} - {\bm{\xi}_i^*}) \\
            & = \Big(\alpha_i \bm{1} + \frac{1}{2}(\Gamma^{  \theta})^\top(\bm{\eta_i} - {\bm{\xi}_i^*}) \Big)^{\top}(\bm{\eta_i} - {\bm{\xi}_i^*}) \\
            & = \alpha_i \bm{1}^{\top}(\bm{\eta_i} - {\bm{\xi}_i^*}) + \frac{1}{2}(\bm{\eta_i} - {\bm{\xi}_i^*}) ^{\top}\Gamma^{  \theta}(\bm{\eta_i} - {\bm{\xi}_i^*}) \\
            & \geq 0,
        \end{split}
    \end{equation}
    with equality if and only if $\bm{\eta_i} = {\bm{\xi}_i^*}$. Therefore, the strategy profile defined by \eqref{NashEqStratDisc} is a deterministic Nash equilibrium and the proof is complete.
\end{proof}

We mention that the proofs in this section remain valid if the exponential kernel $G$ is generalized to an arbitrary positive definite kernel (in the sense of Bochner). 

\section{Proofs for Section~\ref{hfl_section}}\label{sec_3_proofs_appendix}

The proofs for the high-frequency asymptotics of Section~\ref{hfl_section} involve rather lengthy expressions. We start with some abstract remarks and notation intended to make the exposition more concise. While the quantities introduced in Appendix~\ref{closed form section} (e.g., \(\alpha\), \(\bm\nu\), \(\bm\omega\)) depend on the trading frequency \(N\), we usually suppress this dependence for brevity. Throughout, we let \(N\uparrow\infty\), so, for example, we write
\[
\lim_{N\uparrow\infty}\alpha
=
\lim_{N\uparrow\infty} e^{-\rho T/N}
=1.
\]
For \(t\in[0,T]\) we recall the discrete trading index \(n_t=\lceil Nt/T\rceil\) and denote the distance between $Nt/T$ and the subsequent grid point by
\[
\eta_t^{N}:=n_t-\frac{Nt}{T}\in[0,1).
\]
This will appear, for example, when first-order terms depend on \(n_t\).

Rather than expanding every expression directly in powers of \(N^{-1}\), it will be often convenient to introduce the small parameter
\(
\Delta:=1-\alpha=1-e^{-\rho T/N}.
\)
A Taylor expansion at \(0\) yields
\[
\Delta=\frac{\rho T}{N}-\frac{(\rho T)^2}{2N^2}+\mathcal{O}(N^{-3}),
\qquad
\frac{1}{N}=\frac{\Delta}{\rho T}+\frac{\Delta^2}{2\rho T}+\mathcal{O}(\Delta^{3}).
\]
Hence \(o(N^{-p})\) and \(o(\Delta^{p})\) are interchangeable; we switch between these two symbols as convenient.

All the functions we manipulate %
are real-analytic in the neighborhoods we consider. Two consequences, often used without further comment, are the following.
\begin{enumerate}
\item \emph{Stability under algebraic operations.} If \(A_N=a_0+o(N^{-p})\) and \(B_N=b_0+o(N^{-p})\), then
\[
A_N\pm B_N=(a_0\pm b_0)+o(N^{-p}),\qquad A_NB_N=a_0b_0+o(N^{-p}),
\]
and, provided \(b_0\neq 0\),
\[
\frac{A_N}{B_N}=\frac{a_0}{b_0}+o(N^{-p}).
\]
Thus sums, products, and quotients preserve the error order.
\item \emph{Stability under composition.} If \(X_N=x_0+r_N\) with \(r_N=o(N^{-p})\) and \(h\) is real-analytic on a neighborhood \(U\) of \(x_0\), then by %
Taylor’s formula with Lagrange remainder, for any fixed \(q\in\mathbb{N}\) and all sufficiently large \(N\),
\[
h(X_N)=\sum_{k=0}^{q}\frac{h^{(k)}(x_0)}{k!}r_N^{k}
+\frac{h^{(q+1)}(x_0+\zeta_N r_N)}{(q+1)!}r_N^{q+1},
\qquad \zeta_N\in(0,1).
\]
Hence, if \(r_N=\mathcal{O}(N^{-m})\) and \((q+1)m>p\), the remainder is \(o(N^{-p})\). In particular, in our setting, compositions of finitely many analytic maps preserve \(o(N^{-p})\) remainders (equivalently \(o(\Delta^{p})\)). Typical uses below include \(h(x)=1/x\) (with \(x_0\neq0\)), \(h(x)=\sqrt{x}\) (with \(x_0>0\)), \(\log x\) (with \(x_0>0\)), and their compositions.
\end{enumerate}

Whenever a quotient of two analytic expansions is required, we identify the coefficients via the standard series-division rule below; see also \cite[§ 67]{ChurchillBrown.84}. This will be used repeatedly when taking quotients of closed forms and extracting leading constants.

\begin{lemma}[Quotient of analytic Taylor series]
Let \(I\subset\mathbb{R}\) be an open interval containing \(a\), and let \(f,g\) be real-analytic on \(I\) with
\[
f(x)=\sum_{k\ge 0} a_k (x-a)^{k},
\qquad
g(x)=\sum_{k\ge 0} b_k (x-a)^{k},
\]
both converging on some interval \((a-R,a+R)\subset I\). If \(b_0=g(a)\neq 0\), then \(f/g\) is real-analytic on \((a-r,a+r)\) for some \(r\in(0,R)\) with
\[
\frac{f(x)}{g(x)}=\sum_{k\ge 0} c_k (x-a)^{k},
\]
and the coefficients \(\{c_k\}_{k\ge0}\) are uniquely determined by
\[
  b_0c_0=a_0,\qquad
  \sum_{j=0}^m b_jc_{m-j}=a_m\quad(m\ge1).
\]
\end{lemma}

\begin{remark}
In particular, at first order one has
\begin{align}\label{quotient_expansion}
\frac{f(x)}{g(x)}
=
\frac{a_0}{b_0}
+
\frac{a_1 b_0 - a_0 b_1}{b_0^{2}}(x-a)
+\text{higher-order terms}.
\end{align}
\end{remark}

The subsequent proofs proceed by expanding all discrete objects using the conventions above, together with \eqref{quotient_expansion}, to organize remainders into \(o(N^{-p})\) at the target order. 

\subsection{Proof of Theorem~\ref{strat_asympt_thm_theta_positive}~\ref{conv_of_W}}

We remark that the convergence of $W^{(N)}_t$ to $\mathbbm{f}(t)$ for $t\in[0,T)$, without a rate, already follows from \cite[Theorem~3.1(c)]{SchiedStrehleZhang.17} as $W^{(N)}$ is independent of $n$ by Remark~\ref{w_indep_n}. Next, we establish the $1/N$ rate and recover their result as a by-product. We observe that (in contrast to the statement in \cite[Theorem~3.1(c)]{SchiedStrehleZhang.17}, which seems to have a glitch) the sequence $W_T^{(N)}$ does not converge to $\mathbbm{f}(T)$.

\begin{proof}[Proof of Theorem~\ref{strat_asympt_thm_theta_positive}~\ref{conv_of_W}]
    Using the closed-form expression in Theorem~\ref{omega and nu closed form thm},
    \begin{align}\label{W^N_t_expl_form}
        W^{(N)}_t
        =
        \frac{1}{\mathbf{1}^\top\bm{\omega}}
        \sum_{k = n_t+1}^{N+1}\omega_k
        =
        \frac{(N+1-n_t)\frac{(1-\alpha)}{\tilde{\kappa}-\alpha(\tilde{\kappa}-1)}
        +
        \frac{\alpha}{\tilde{\kappa}(\tilde{\kappa}-\alpha(\tilde{\kappa}-1))}
        \frac{1-\left(\frac{\alpha(\tilde{\kappa}-1)}{\tilde{\kappa}}\right)^{N+1-n_t}}{1-\frac{\alpha(\tilde{\kappa}-1)}{\tilde{\kappa}}}}
        {(N+1)\frac{(1-\alpha)}{\tilde{\kappa}-\alpha(\tilde{\kappa}-1)}
        +
        \frac{\alpha}{\tilde{\kappa}(\tilde{\kappa}-\alpha(\tilde{\kappa}-1))}
        \frac{1-\left(\frac{\alpha(\tilde{\kappa}-1)}{\tilde{\kappa}}\right)^{N+1}}{1-\frac{\alpha(\tilde{\kappa}-1)}{\tilde{\kappa}}}} .
    \end{align}
    We first treat $t=T$. With $\tilde{\kappa}=2\theta+\frac{1}{2}$ and using \eqref{quotient_expansion},
    \[
        W^{(N)}_T
        =
        \frac{\omega_{N+1}}{\mathbf{1}^\top\bm{\omega}}
        =
        \frac{1}{\tilde{\kappa}\mathbf{1}^\top\bm{\omega}}
        =
        \frac{1}{\tilde{\kappa}(\rho T+1)}
        -
        \frac{1}{N}
        \frac{(\tilde{\kappa}-\frac{3}{2})\rho^2T^2-2\rho T(\tilde{\kappa}-1)(\rho T+1)}
             {\tilde{\kappa}(\rho T+1)^2}
        + o\left(\frac{1}{N}\right),
    \]
    which yields the stated claim at $t=T$.

    Now fix $t\in[0,T)$ and apply \eqref{quotient_expansion}--\eqref{W^N_t_expl_form}. A direct calculation (whose details we omit for the sake of brevity) yields
    \[
        N|W^{(N)}_t-\mathbbm{f}(t)|
        =
        \frac{\rho T}{\rho T+1}
        \bigg|
           \eta_t^{N}
           + \frac{\rho t(2\theta-1)}{\rho T+1}
        \bigg| + o(1),\qquad N\to\infty.
    \] 
    This proves the claimed $\mathcal{O}(N^{-1})$ rate of convergence to $\mathbbm{f}(t)$ for every fixed $t\in[0,T)$.
\end{proof}

\subsection{Proof of Theorem~\ref{strat_asympt_thm_theta_positive}~\ref{conv_of_V}}

We state separately the proofs for $\kappa = n - 1$ and $\kappa \neq n - 1$. The details are different because the general representation for the sum of the components of $\bm\nu$ in \eqref{ae} involves denominators that vanish exactly at $\kappa = n - 1$, and therefore is not well defined at this value.

\subsubsection{Proof for $\kappa = n - 1$}

Adapting arguments from \cite{SchiedStrehleZhang.17}, we first consider the case $\kappa=n-1$, which corresponds to $\theta=\frac{n-1}{4}$. The proof of Theorem~\ref{strat_asympt_thm_theta_positive}\ref{conv_of_V} for this particular value of $\kappa$ will be given after the following lemma.

\begin{lemma}
    Let $\kappa=n-1$. Then, for $m\in\{1,\dots,N+1\}$,
    \begin{align}\label{aa}
        \sum_{i=1}^m \nu_i
        = \frac{1}{n+\alpha}\left(
            (1-\alpha)m+\alpha
            +\frac{\alpha(\alpha^2-n)}{n(n+\alpha)}\left(\frac{\alpha(n-1)}{n-\alpha^2}\right)^{N+1}
            +\frac{\alpha(1+\alpha)}{n+\alpha}\left(\frac{\alpha(n-1)}{n-\alpha^2}\right)^{N+1-m}
        \right).
    \end{align}
\end{lemma}

\begin{proof}
Plugging in $\kappa=n-1$ yields $R=n-\alpha^2$, $\delta_k=n(1-\alpha^2)(n-\alpha^2)^{ k-1}$ for $k\in\{1,\dots,N+1\}$, and $\phi_k=(n-\alpha^2)^{N+2-k}$ for $k\in\{2,\dots,N+2\}$. Therefore,
\begin{equation}\label{expl_form_nu_i_kappa_n-1}
\begin{aligned}
    \nu_1
    &= \frac{1}{n+\alpha}\left(
        1+\left(\frac{n-\alpha^2}{n(n-1)}\right)\left(\frac{\alpha(n-1)}{n-\alpha^2}\right)^{N+1}
    \right),\\
    \nu_i
    &= \frac{1}{n+\alpha}\left(
        1-\alpha
        +\left(\frac{\alpha(n-1)}{n-\alpha^2}\right)^{N+2-i}\left(\frac{1-\alpha^2}{n-1}\right)
    \right),
    \qquad i\in\{2,\dots,N+1\}.
\end{aligned}
\end{equation}
Summing \eqref{expl_form_nu_i_kappa_n-1} over $i=1,\dots,m$ yields \eqref{aa}.
\end{proof}

\begin{proof}[Proof of Theorem~\ref{strat_asympt_thm_theta_positive}~\ref{conv_of_V} for $\kappa=n-1$]

Recall that $\alpha=e^{-\rho T/N}$. Therefore,
\begin{align*}
    (1-\alpha)n_t
    &= \rho t + \frac{1}{N}\big(\eta_t^{N}\rho T - \tfrac{\rho^2 T t}{2}\big)
       + o\!\left(\frac{1}{N}\right),\\
    \left(\frac{\alpha(n-1)}{n-\alpha^2}\right)^{N+1}
    &= e^{-\rho T\frac{n+1}{n-1}}
       \bigg(1+\frac{1}{N}\Big(- \rho T\frac{n+1}{n-1}
       + \frac{2n\rho^2 T^2}{(n-1)^2}\Big)\bigg)
       + o\!\left(\frac{1}{N}\right),\\
    \left(\frac{\alpha(n-1)}{n-\alpha^2}\right)^{N+1-n_t}
    &= e^{-\rho\frac{n+1}{n-1}(T-t)}
       \bigg(1+\frac{1}{N}\Big(\frac{2n\rho^2 T(T-t)}{(n-1)^2}
       - (1-\eta_t^{N}) \rho T\frac{n+1}{n-1}\Big)\bigg)
       + o\!\left(\frac{1}{N}\right),\\
    \frac{1}{n+\alpha}
    &= \frac{1}{n+1} + \frac{\rho T}{N(n+1)^{2}} + o\!\left(\frac{1}{N}\right),\\
    \alpha
    &= 1 - \frac{\rho T}{N} + o\!\left(\frac{1}{N}\right),\\
     \frac{\alpha(\alpha^{2}-n)}{n(n+\alpha)} &= \frac{1-n}{n(n+1)} + \frac{1}{N} \frac{\rho T(n^{2}-3n-2)}{n(n+1)^{2}} + o\!\left(\frac{1}{N}\right),\\
    \frac{\alpha(1+\alpha)}{n+\alpha}
    &= \frac{2}{n+1}
       - \frac{1}{N} \frac{\rho T(3n+1)}{(n+1)^{2}}
       + o\!\left(\frac{1}{N}\right)
\end{align*}
for all $t\in(0,T]$. Moreover,
\[
({1-\alpha}) (N+1)
= \rho T + \frac{\rho T - \frac{\rho^{2}T^{2}}{2}}{N}
  + o  \tonde{\frac{1}{N}}.
\]

Putting everything together in \eqref{aa} yields
\begin{align}\label{sumnuitotal}
\sum_{i=1}^{N+1}\nu_i
= \frac{e^{-\rho \frac{n+1}{n-1}T}}{(n+1)^2 n}
   \Big( n\big((\rho T+1)(n+1)+2\big)e^{\rho \frac{n+1}{n-1}T} - (n-1)\Big)
   + \mathscr{Q} \frac{1}{N} +o  \tonde{\frac{1}{N}},
\end{align}
where
\begin{align*}
    \mathscr{Q}=\frac{\rho T}{2({n+1})^{3}}
   \tonde{(\rho T)({1-n^{2}}) - 4(n-1)}
  + e^{-\rho T\frac{n+1}{n-1}}\left( \frac{2\rho T(n-1)}{(n+1)^3} - \frac{2\rho^2T^2}{(n+1)^2(n-1)} \right)
\end{align*}
and
\begin{align}\label{sumnuiupto_t}
\sum_{i=1}^{n_t}\nu_i
= \frac{e^{-\rho \frac{n+1}{n-1}T}}{(n+1)^2 n}
  \Big( n(\rho t+1)(n+1)e^{\rho \frac{n+1}{n-1}T}
        - (n-1) + 2n e^{\rho \frac{n+1}{n-1}t}
  \Big)
  + \mathscr{R}_N(t) \frac{1}{N} +o  \tonde{\frac{1}{N}},
\end{align}
where
\[
\mathscr{R}_N(t)
= \frac{1}{n+1} b_1({t,\eta_t^{N}})
  + \frac{\rho T}{({n+1})^{2}} b_0({t})
\]
with
\[
b_0({t})
= \rho t + 1 + \frac{1-n}{n({n+1})} e^{-\rho\frac{n+1}{n-1}T}
  + \frac{2}{n+1} e^{-\rho\frac{n+1}{n-1}({T-t})}
\]
and
\begin{align*}
b_1({t,\eta})
&= \eta \rho T - \frac{\rho^{2}T t}{2} - \rho T
   + e^{-\rho\frac{n+1}{n-1}T}\frac{1-n}{n({n+1})}
     \tonde{-\frac{\rho T({n+1})}{n-1} + \frac{2n\rho^{2}T^{2}}{({n-1})^{2}}}
\\
&\quad
   + e^{-\rho\frac{n+1}{n-1}T} \frac{\rho T({n^{2}-3n-2})}{n({n+1})^{2}}
   + \frac{2}{n+1}e^{-\rho\frac{n+1}{n-1}({T-t})}
     \tonde{\frac{2n\rho^{2}T({T-t})}{({n-1})^{2}}
           - ({1-\eta}) \rho T\frac{n+1}{n-1}}
\\
&\quad
   - \frac{\rho T({3n+1})}{({n+1})^{2}}
     e^{-\rho\frac{n+1}{n-1}({T-t})}.
\end{align*} 
Because $\eta_t^{N}\in[0,1)$ for all $N\in\mathbb{N}$ and $b_1$ depends linearly on $\eta_t^{N}$, the sequence $\mathscr{R}_N(t)$ is bounded for fixed parameters $\theta,T,\rho,n$.

Finally, plugging \eqref{sumnuitotal} and \eqref{sumnuiupto_t} into the definition of $V^{(N)}_t$ and applying \eqref{quotient_expansion} once more yields the claim.
\hide{Indeed, 
\begin{align*}
    \sum_{i=1}^{N+1}\nu_i &\longrightarrow \frac{1}{n+1}\left( \rho T +1 +\frac{1-n}{n(n+1)}e^{-\rho \frac{n+1}{n-1}T} +\frac{2}{n+1} \right)\\
    &=\frac{e^{-\rho \frac{n+1}{n-1}T}}{(n+1)^2n}\Big( n((\rho T +1)(n+1) +2)e^{\rho \frac{n+1}{n-1}T} - (n-1)
    \Big)
\end{align*}
and
\begin{align*}
    \sum_{i=1}^{n_t}\nu_i &\longrightarrow \frac{1}{n+1}\left( \rho t +1 +\frac{1-n}{n(n+1)}e^{-\rho \frac{n+1}{n-1}T} +\frac{2}{n+1}e^{-\rho \frac{n+1}{n-1}(T-t)} \right)\\
    &=\frac{e^{-\rho \frac{n+1}{n-1}T}}{(n+1)^2n} \Big( n(\rho t +1)(n+1)e^{\rho \frac{n+1}{n-1}T}  - (n-1) + 2ne^{\rho \frac{n+1}{n-1}t}
    \Big).
\end{align*}}
\hide{Indeed, by definition of $V^{(N)}_t$, we have 
\begin{align*}
    V^{(N)}_t = 1-\sum_{k=1}^{n_t} v_k \longrightarrow 1-\frac{n(\rho t +1)(n+1)e^{\rho\frac{n+1}{n-1}T}+2ne^{\rho\frac{n+1}{n-1}t}-(n-1)}{n((\rho T+1)(n+1)+2)e^{\rho\frac{n+1}{n-1}T}-(n-1)} = \mathbbm{g}(t)
\end{align*}}
\end{proof}

\subsubsection{Proof for $\kappa\neq n-1$}

We now prepare for the proof of Theorem~\ref{strat_asympt_thm_theta_positive}~\ref{conv_of_V} for the case $\kappa\neq n-1$. We introduce the shorthand 
\[
  [x]^m := \frac{1-\alpha}{\delta_{N+1}} x^m
\]
for $x\in\mathbb{R}$ and $m\in\mathbb{N}$, which is convenient when taking limits of terms such as $\left[x\right]^N$.

\begin{lemma}\label{sum_of_nu_i_lemma}
Let $\kappa \ge \frac{n-1}{2}$ and $\kappa \neq n-1$. Define $\hat{\kappa}:=n-1$ and $C_1  :=  \frac{\alpha(\alpha+1)}{\kappa+1-\alpha\big(\kappa-\hat{\kappa}-1\big)}$.
Then, for $m\in\{1,\dots,N\}$,
\begin{align}\label{ae}
\sum_{i=1}^m \nu_i
&= \sums \frac{d_\sigma\big(m_\sigma-\alpha^2\kappa\big)}{m_\sigma-\alpha\kappa} \big[m_\sigma\big]^N \\
&\quad + (1-\alpha)(m-1) \sums c_\sigma d_\sigma
\left(
  \frac{\alpha\big(\kappa-\hat{\kappa}\big)}{m_\sigma-\alpha\big(\kappa-\hat{\kappa}\big)}
  + \frac{m_\sigma}{m_\sigma-\alpha\kappa}
\right)\big[m_\sigma\big]^N \nonumber \\
&\quad + C_1  \left(1+\sums \frac{c_\sigma m_\sigma\Big(\big(\frac{m_\sigma}{\alpha\kappa}\big)^{m-1}-1\Big)}{m_\sigma-\alpha\kappa}\right)\alpha^N [\kappa]^N \nonumber \\
&\quad + nC_1 \sums \frac{d_\sigma m_\sigma  \left(\frac{\alpha(\kappa-\hat{\kappa})}{m_\sigma}
        - \big(\tfrac{\alpha(\kappa-\hat{\kappa})}{m_\sigma}\big)^m\right)}
      {m_\sigma-\alpha(\kappa-\hat{\kappa})} \big[m_\sigma\big]^N, \nonumber
\end{align}
and
\begin{equation}\label{af}
\nu_{N+1}
= \sums \frac{c_\sigma\big(m_\sigma-\alpha^2(\kappa-\hat{\kappa})\big)}{m_\sigma-\alpha(\kappa-\hat{\kappa})} \big[m_\sigma\big]^N
  + nC_1 \alpha^N [\kappa-\hat{\kappa}]^N .
\end{equation}
\end{lemma}

\begin{proof}
    For $i\in\{3,\dots,N\}$ we have
    \begin{align}\label{ab}
    \sum_{j=2}^{i-1}  (\alpha(\kappa-\hat{\kappa}))^{ i-j}\delta_{j-1}\phi_{i+1}
    &= \alpha(\kappa-\hat{\kappa})\Bigg(
    \sums \frac{c_\sigma d_\sigma}{m_\sigma-\alpha(\kappa-\hat{\kappa})} m_\sigma^{N}
    \\[-1ex]
    &\hspace{2.1cm}
    + \frac{c_+ d_- m_-^{N+1}}{m_+(m_+-\alpha(\kappa-\hat{\kappa}))}
      \left(\frac{m_+}{m_-}\right)^{  i}
    + \frac{c_- d_+ m_+^{N+1}}{m_-(m_--\alpha(\kappa-\hat{\kappa}))}
      \left(\frac{m_-}{m_+}\right)^{  i} \nonumber \\
    &\hspace{2.1cm}
    - \sums \frac{c_\sigma m_\sigma}{m_\sigma-\alpha(\kappa-\hat{\kappa})}
           \sumt \frac{d_\tau m_\tau^{N+1}}{(\alpha(\kappa-\hat{\kappa}))^2}
                 \left(\frac{\alpha(\kappa-\hat{\kappa})}{m_\tau}\right)^{  i}
    \Bigg) \nonumber
    \end{align}
    \hide{Indeed
    \begin{align*}
        \delta_{j-1}\phi_{i+1}&= \sums c_\sigma d_\sigma m_\sigma^Nm_\sigma^{j-i} + c_+m_+^{j-1}d_-m_-^{N+1-i} + c_-m_-^{j-1}d_+m_+^{N+1-i}\\
        &= \sums A_\sigma + B +D
    \end{align*}
    so 
    \begin{align*}
        \sum_{j=2}^{i-1}\sums (\alpha(\kappa-\hat{\kappa}))^{i-j} A_\sigma &= \sums c_\sigma d_\sigma m_\sigma^N\sum_{j=2}^{i-1} \left( \frac{m_\sigma}{\alpha(\kappa-\hat{\kappa})} \right)^{j-i}  \\
        &= (\alpha(\kappa-\hat{\kappa})) \sums \frac{c_\sigma d_\sigma}{m_\sigma - \alpha(\kappa-\hat{\kappa})} m_\sigma^N  \left(1-\left( \frac{m_\sigma}{\alpha(\kappa-\hat{\kappa})} \right)^{2-i}\right)
    \end{align*}
    and
    \begin{align*}
        \sum_{j=2}^{i-1} (\alpha(\kappa-\hat{\kappa}))^{i-j}B &= c_+d_-m_-^{N+1-i}m_+^{-1}(\alpha(\kappa-\hat{\kappa}))^{i}\sum_{j=2}^{i-1}\left(\frac{m_+}{(\alpha(\kappa-\hat{\kappa}))}\right)^{j}\\
        &= (\alpha(\kappa-\hat{\kappa}))\left(\frac{c_+d_-m_-^{N+1}}{m_+(m_+-\alpha(\kappa-\hat{\kappa}))} \left(\frac{m_+}{m_-} \right)^i +\right. \\
        &\quad \left. - \frac{c_+m_+}{m_+-\alpha(\kappa-\hat{\kappa})}\frac{d_-m_-^{N+1}}{(\alpha(\kappa-\hat{\kappa}))^2}\left(\frac{\alpha(\kappa-\hat{\kappa})}{m_-}\right)^i \right)
    \end{align*}
    finally for $D$ is the same as before with $+$ and $-$ inverted.\\
    }
    and
    \begin{align}\label{ac}
    \sum_{j=i}^{N}  (\alpha\kappa)^{ j-i}\delta_{i-1}\phi_{j+1}
    &= \sums \frac{c_\sigma d_\sigma}{m_\sigma-\alpha\kappa} m_\sigma^{N+1}
     + \frac{c_+ d_- m_-^{N+2}}{m_+\big(m_--\alpha\kappa\big)}
       \left(\frac{m_+}{m_-}\right)^{  i}
     + \frac{c_- d_+ m_+^{N+2}}{m_-\big(m_+-\alpha\kappa\big)}
       \left(\frac{m_-}{m_+}\right)^{  i}
    \\[-1ex]
    &\quad
    - \sums \frac{d_\sigma m_\sigma}{m_\sigma-\alpha\kappa}
           \sumt \frac{c_\tau (\alpha\kappa)^{N+1}}{m_\tau}
                 \left(\frac{m_\tau}{\alpha\kappa}\right)^{  i}. \nonumber
    \end{align}
    Using
    \[
        \alpha(\kappa-\hat{\kappa})(m_--\alpha\kappa)+m_-\big(m_+-\alpha(\kappa-\hat{\kappa})\big)
        = \alpha(\kappa-\hat{\kappa})(m_+-\alpha\kappa)+m_+\big(m_--\alpha(\kappa-\hat{\kappa})\big)
        = m_+m_- - \alpha^2\kappa(\kappa-\hat{\kappa}) = 0,
    \]
    the second and third terms in \eqref{ab} and \eqref{ac} cancel. After simplification we obtain, for $i\in\{2,\dots,N\}$,
    \begin{align}\label{expl_form_nu_i}
    \nu_i
    &= (1-\alpha) \sums c_\sigma d_\sigma
    \left(
      \frac{\alpha(\kappa-\hat{\kappa})}{m_\sigma-\alpha(\kappa-\hat{\kappa})}
      + \frac{m_\sigma}{m_\sigma-\alpha\kappa}
    \right) [m_\sigma]^N
    \\
    &\quad
    + nC_1 \sums \frac{d_\sigma m_\sigma}{\alpha(\kappa-\hat{\kappa})} [m_\sigma]^N
      \left(\frac{\alpha(\kappa-\hat{\kappa})}{m_\sigma}\right)^{  i}
    + C_1 \sums \frac{c_\sigma \alpha^{N+1}\kappa}{m_\sigma} [\kappa]^N
      \left(\frac{m_\sigma}{\alpha\kappa}\right)^{  i}. \nonumber
    \end{align}
    \hide{To see the last calculation, we just compute
    \begin{align*}
         C_1 &= 1+(\alpha-1)\left(\frac{d_+m_+}{m_+ -\alpha\kappa} + \frac{d_-m_-}{m_- -\alpha\kappa}\right)\\
         &=(\alpha-1)\left(\frac{m_+m_--\alpha\kappa(1-\alpha^2+\kappa)}{m_+m_- -\alpha\kappa(m_++m_-) +\alpha^2\kappa^2}\right) +1\\
         &= (\alpha-1)\left(\frac{\alpha^2\kappa(\kappa-\hat{\kappa})-\alpha\kappa(1-\alpha^2+\kappa)}{\alpha^2\kappa(2\kappa-\hat{\kappa}) -\alpha\kappa(1+\alpha^2(\kappa-n)+\kappa)}\right) +1\\
         &=\frac{(\alpha-1)(1+\alpha)\alpha}{\alpha(2\kappa-\hat{\kappa})-1-\alpha^2(\kappa-n)-k}\\
         &=\frac{(1+\alpha)\alpha}{\kappa +1 -\alpha(\kappa - \hat{\kappa}-1)}
    \end{align*}
    for the penultimate calculation, we just compute
    \begin{align*}
        &1+(\alpha-1)\left( \frac{c_+m_+}{m_+-\alpha(\kappa-\hat{\kappa})} + \frac{c_-m_-}{m_--\alpha(\kappa-\hat{\kappa})} \right)\\
        &=1+(\alpha-1)\left( \frac{m_+m_- -\alpha(\kappa-\hat{\kappa})(1-n\alpha^2+\kappa)}{m_+m_--\alpha(\kappa-\hat{\kappa})(1+\alpha^2(\kappa-n)+\kappa)+\alpha^2(\kappa-\hat{\kappa})} \right)\\
        &=1 +(\alpha-1) \frac{\alpha\kappa - (1-n\alpha^2 + \kappa)}{\alpha\kappa - (1+\alpha^2(\kappa-n)+\kappa)+\alpha(\kappa-\hat{\kappa})} \\ 
        &=\frac{-\alpha(1-n\alpha^2+\kappa) + \alpha(\kappa-\hat{\kappa})}{(\alpha-1)(\kappa+1-\alpha(\kappa-\hat{\kappa}-1))} = nC_1
    \end{align*}}
    Similar computations give
    \begin{align}
        \nu_1
        &= \sums \frac{d_\sigma(m_\sigma-\alpha^2\kappa)}{m_\sigma-\alpha\kappa} [m_\sigma]^N
           + C_1 \alpha^N [\kappa]^N, \label{nu_one_expl_form}\\
        \nu_{N+1}
        &= \sums \frac{c_\sigma(m_\sigma-\alpha^2(\kappa-\hat{\kappa}))}{m_\sigma-\alpha(\kappa-\hat{\kappa})} [m_\sigma]^N
           + nC_1 \alpha^N [\kappa-\hat{\kappa}]^N. \label{nu_N+1_expl_form}
    \end{align}
    Finally, for $m\in\{2,\dots,N\}$,
    \[
    \sum_{i=2}^{m} \sums \frac{d_\sigma m_\sigma}{\alpha(\kappa-\hat{\kappa})} [m_\sigma]^N
    \left(\frac{\alpha(\kappa-\hat{\kappa})}{m_\sigma}\right)^{  i}
    = \sums \frac{d_\sigma m_\sigma\Big(\big(\frac{m_\sigma}{\alpha(\kappa-\hat{\kappa})}\big)^{N-1}
      - \big(\frac{m_\sigma}{\alpha(\kappa-\hat{\kappa})}\big)^{N-m}\Big)}
     {m_\sigma-\alpha(\kappa-\hat{\kappa})} \alpha^N [\kappa-\hat{\kappa}]^N,
    \]
    and
    \[
    \sum_{i=2}^{m} \sums \frac{c_\sigma \alpha^{N+1}\kappa}{m_\sigma} [\kappa]^N  
    \left(\frac{m_\sigma}{\alpha\kappa}\right)^{  i}
    = \sums \frac{c_\sigma m_\sigma\Big(\big(\frac{m_\sigma}{\alpha\kappa}\big)^{m-1}-1\Big)}
     {m_\sigma-\alpha\kappa} \alpha^N [\kappa]^N,
    \]
    which, together with \eqref{expl_form_nu_i}, \eqref{nu_one_expl_form}, and \eqref{nu_N+1_expl_form}, proves the claim.
\end{proof}

The next lemma collects the limiting behavior of all quantities that appear in the derivation of the limiting strategy and, subsequently, the limiting costs. In addition, for the case $\theta>0$ (equivalently, $\kappa>\hat{\kappa}/2$) we record first or second-order Taylor expansions used to compute the pointwise convergence rate of the strategies.

For a sequence $\tonde{a_N}_{N\in\mathbb{N}}$ and a real number $a$, we use the shorthand
\[
({a_N})^{ n_t}\to \pm a
\quad:\Longleftrightarrow\quad
({a_N})^{ n_t}=({-1})^{n_t} \lvert a_N\rvert^{ n_t}
\ \text{ and }\ 
\lim_{N\to\infty}\lvert a_N\rvert^{ n_t}=a.
\]
Recall that $\Delta=({1-\alpha})$ and note that $\Delta\to0$ as $N\to\infty$. When convenient, we express expansions in the variable $\Delta$; see also the discussion at the beginning of Appendix~\ref{sec_3_proofs_appendix}.

\begin{lemma}\label{lemma_asympt_rate_various_elements}
    For $\kappa>\frac{\hat{\kappa}}{2}$ and $\kappa\neq\hat{\kappa}$, we have the following Taylor expansions as $N\uparrow\infty$. 
    \begin{enumerate}[label = $\mathrm{(\alph*)}$]
        \itemsep1em
        \item 
        \begin{flalign*}
            &\hspace{1.23cm}\begin{aligned}
                &\alpha&&=1-\frac{\rho T}{N}+\frac{\rho^2T^2}{2N^2}+o  \left(\frac{1}{N^2}\right),
                \\
                &\alpha^{n_t}&&=e^{-\frac{\rho T}{N}\left(\frac{Nt}{T}+\eta^N_t\right)}
                =e^{-\rho t}  \left(1-\frac{\rho T \eta^N_t}{N}+\frac{(\rho T \eta^N_t)^2}{2N^2}\right)
                +o  \left(\frac{1}{N^2}\right).
            \end{aligned}&&
        \end{flalign*}
        \item\label{lemma_asympt_rate_various_elements:label_b}
        \begin{flalign*}
            &\hspace{1.23cm}\begin{aligned}
                &R      &&=\hat{\kappa}
                +\Delta 2\tonde{\kappa\frac{n+1}{\hat{\kappa}}-n}
                +\Delta^2\frac{2\kappa^2+3n^2-5\kappa n-\kappa-n-2({\kappa\frac{n+1}{\hat{\kappa}}-n})^2}{\hat{\kappa}}
                +o  ({\Delta^2}), 
                \\
                &c_+    &&=\Delta\frac{2\kappa n}{\hat{\kappa}^2}
                -\Delta^2 \kappa n\frac{4\kappa n+8\kappa-3n^2-2n+5}{\hat{\kappa}^4}
                +o  ({\Delta^2}),
                \\
                &c_-    &&=1-\Delta\frac{2\kappa n}{\hat{\kappa}^2}
                +\Delta^2 \kappa n\frac{4\kappa n+8\kappa-3n^2-2n+5}{\hat{\kappa}^4}
                +o  ({\Delta^2}),
                \\
                &d_+    &&=1-\Delta\frac{2}{\hat{\kappa}^2}({\kappa-\hat{\kappa}})
                +\Delta^2\frac{3\ka^3-11\ka^2\kappa+4\ka^2+8\ka\kappa^2-16\ka\kappa+12\kappa^2}{\ka^4}
                +o  ({\Delta^2}),
                \\
                &d_-    &&=\Delta\frac{2}{\ka^2}({\kappa-\ka})
                +\Delta^2\frac{-3\ka^3+11\ka^2\kappa-4\ka^2-8\ka\kappa^2+16\ka\kappa-12\kappa^2}{\ka^4}
                +o  ({\Delta^2}),
                \\
                &m_+    &&=\kappa+\Delta\frac{2\kappa}{\hat{\kappa}}
                +\Delta^2\kappa\frac{3\ka^2-4\ka\kappa+4\ka-4\kappa}{\ka^3}
                +o  ({\Delta^2}),
                \\
                &m_-    &&=\kappa-\hat{\kappa}+\Delta\frac{2n({\ka-\kappa})}{\ka}
                +\Delta^2 n\frac{-\ka^3+4\kappa^2+\kappa n^2-6\kappa n+5\kappa}{\ka^3}
                +o  ({\Delta^2}).
            \end{aligned}&&
        \end{flalign*}
        \item\label{lemma_asympt_rate_various_elements:label_c}
        \begin{flalign*}
            &\hspace{1.23cm}\begin{aligned}
                &\frac{c_+}{m_+-\kappa}                     &&=\frac{n}{\ka}
                +\Delta\frac{2n({\ka-2\kappa})}{\ka^3}
                +o  ({\Delta}),
                \\
                &\frac{c_+}{m_+-\alpha\kappa}               &&=\frac{2n}{\ka({n+1})}
                +\Delta n \frac{-4\kappa n^2-4\kappa n-8\kappa+3n^3-n^2+n-3}{({n+1})^2\ka^3}
                +o  ({\Delta}),
                \\
                &\frac{c_+}{m_+-\alpha^2\kappa}             &&=\frac{1}{n-1}
                +\Delta \frac{-2\kappa-n(2\kappa-\ka)+n+\ka^2-1}{\ka^3}
                +o({\Delta}),
                \\
                &\frac{c_+}{1-\alpha^2}                     &&=\frac{\kappa n}{\ka^2}
                -\Delta 2\kappa n \frac{\kappa n+2\kappa-n^2+1}{\ka^4}
                +o({\Delta}).
            \end{aligned}&&
        \end{flalign*}
        \item\label{lemma_asympt_rate_various_elements:label_d}
        \begin{flalign*}
            &\hspace{1.23cm}\begin{aligned}
                &\frac{d_-}{m_--({\kappa-\ka})}                &&=-\frac{1}{n\ka}
                +\Delta \frac{8\theta}{({n-1})^3}
                +o({\Delta}),
                \\
                &\frac{d_-}{m_--\alpha({\kappa-\ka})}          &&=-\frac{2}{(n+1)\ka}
                +\Delta \frac{8\kappa n^2+4\kappa n+4\kappa-5n^3+3n^2+n+1}{\ka^3(n+1)^2}
                +o({\Delta}),
                \\
                &\frac{d_-}{m_--\alpha^2(\kappa-\ka)}        &&=-\frac{1}{\ka}
                +\Delta \frac{2\kappa+n(2\kappa-\ka)-n-\ka^2+1}{\ka^3}
                +o({\Delta}),
                \\
                &\frac{d_-}{1-\alpha^2}                            &&=\frac{\kappa-\ka}{\ka^2}
                +\Delta 2 \frac{-\ka^3+3\ka^2\kappa-\ka^2-2\ka\kappa^2+4\ka\kappa-3\kappa^2}{\ka^4}
                +o({\Delta}).
            \end{aligned}&&
        \end{flalign*}
    \end{enumerate}
    \noindent Fix $t\in(0,T]$ and recall $\eta_T^N=0$ for all $N\in\N$. If $\kappa>\hat{\kappa}/2$, the following expansions hold.
	\begin{enumerate}[label = $\mathrm{(\alph*)}$]
	   \setcounter{enumi}{4}
       \item\label{lemma_asympt_rate_various_elements:label_e}
       \begin{flalign*}
            &\hspace{1.23cm}\begin{aligned}
                (1-\alpha) n_t=\rho t+\frac{\rho T}{N}\tonde{\eta^N_t-\frac{\rho t}{2}}
                +o  \tonde{\frac{1}{N}}.
            \end{aligned}&&
        \end{flalign*}
		\item\label{lemma_asympt_rate_various_elements:label_f} 
        \begin{flalign*}
            &\hspace{1.23cm}\begin{aligned}
                &\tonde{\frac{\kappa-\ka}{\kappa}}^{n_t}    &&=o  \tonde{\frac{1}{N}}
                \quad\text{and more generally}\quad
                \tonde{\frac{\kappa-\ka}{\kappa}}^{n_t}=o  \tonde{\frac{1}{N^p}},\ \forall p\in\N,
                \\
                &\tonde{\frac{m_+}{\kappa}}^{n_t}           &&=\exp  \tonde{\frac{2\rho t}{n-1}}
                \tonde{1+\frac{\rho T}{N}\tonde{-\rho t \frac{8n\theta}{(n-1)^3}+\frac{2}{n-1}\eta^N_t}}
                +o  \tonde{\frac{1}{N}},
                \\
                &\tonde{\frac{\kappa-\ka}{m_+}}^{n_t}       &&=o  \tonde{\frac{1}{N}}
                \quad\text{and more generally}\quad
                \tonde{\frac{\kappa-\ka}{m_+}}^{n_t}=o  \tonde{\frac{1}{N^p}},\ \forall p\in\N,
                \\
                &\tonde{\frac{m_-}{\kappa}}^{n_t}           &&=o  \tonde{\frac{1}{N}}
                \quad\text{and more generally}\quad
                \tonde{\frac{m_-}{\kappa}}^{n_t}=o  \tonde{\frac{1}{N^p}},\ \forall p\in\N,
                \\
                &\tonde{\frac{\kappa-\ka}{m_-}}^{n_t}       &&=\exp  \tonde{\frac{2n\rho t}{n-1}}
                \tonde{1+\frac{\rho T n}{N}\tonde{-\rho t \frac{8\theta}{(n-1)^3}+\frac{2}{n-1}\eta^N_t}}
                +o  \tonde{\frac{1}{N}}.
            \end{aligned}&&
        \end{flalign*}
        \item\label{lemma_asympt_rate_various_elements:label_g}
        \begin{flalign*}
            &\hspace{1.23cm}\begin{aligned}
                &\quadre{m_+}^N         &&=\frac{\ka}{2\kappa n}  \left(1+\Delta \frac{\ka^2-8\ka+8\kappa}{2\ka^2}\right)
                +o  \tonde{\Delta},
                \\
                &[m_-]^N         &&=o  \tonde{\frac{1}{N}}
                \quad\text{and more generally}\quad
                [m_-]^N=o  \tonde{\frac{1}{N^p}},\ \forall p\in\N,
                \\
                &\quadre{\kappa}^N      &&=\exp  \tonde{\frac{-2\rho T}{\ka}} 
                \frac{\ka}{2\kappa n}  \left(1+\frac{\rho T}{N}\tonde{\rho T \frac{8n\theta}{\ka^3}+\frac{\ka^2-8\ka+8\kappa}{2\ka^2}}\right)
                +o  \tonde{\frac{1}{N}},
                \\
                &[\kappa-\ka]^N  &&=o  \tonde{\frac{1}{N}}
                \quad\text{and more generally}\quad
                [\kappa-\ka]^N=o  \tonde{\frac{1}{N^p}},\ \forall p\in\N.
            \end{aligned}&&
        \end{flalign*}
        \item\label{lemma_asympt_rate_various_elements:label_h}
        \begin{flalign*}
            &\hspace{1.23cm}\begin{aligned}
                &\frac{\tonde{\frac{\kappa-\ka}{\kappa}}^N}{1-\alpha^2}     &&=o  \tonde{\frac{1}{N}}
                \quad\text{and more generally}\quad
                \frac{\tonde{\frac{\kappa-\ka}{\kappa}}^N}{1-\alpha^2}=o  \tonde{\frac{1}{N^p}},\ \forall p\in\N,
                \\
                &\frac{[m_-]^N}{1-\alpha^2}                          &&=o  \tonde{\frac{1}{N}}
                \quad\text{and more generally}\quad
                \frac{[m_-]^N}{1-\alpha^2}=o  \tonde{\frac{1}{N^p}},\ \forall p\in\N,
                \\
                &\frac{[\kappa-\ka]^N}{1-\alpha^2}                   &&=o  \tonde{\frac{1}{N}}
                \quad\text{and more generally}\quad
                \frac{[\kappa-\ka]^N}{1-\alpha^2}=o  \tonde{\frac{1}{N^p}},\ \forall p\in\N.
            \end{aligned}&&
        \end{flalign*}
	\end{enumerate}
    \noindent If, on the other hand, $\kappa=\hat{\kappa}/2$, then the preceding limits no longer hold. Instead, we have:	
    \begin{enumerate}[label={\normalfont(\alph*')}]
    	\itemsep1em
    	\setcounter{enumi}{5}
    		\item\label{f'} $\left(\frac{\kappa-\hat{\kappa}}{\kappa}\right)^{n_t}\to{\pm1},\quad
        \left(\frac{m_+}{\kappa}\right)^{n_t}\to e^{\frac{2\rho t}{n-1}},\quad
        \left(\frac{\kappa-\hat{\kappa}}{m_+}\right)^{n_t}\to{\pm e^{-\frac{2\rho t}{n-1}}},\quad
        \left(\frac{m_-}{\kappa}\right)^{n_t}\to{\pm e^{-\frac{2n\rho t}{n-1}}}$,
        \\
        and $\left(\frac{\kappa-\hat{\kappa}}{m_-}\right)^{n_t}\to e^{\frac{2n\rho t}{n-1}}$;
    	\item\label{g'} $[m_+]^{2N}\to\frac{1}{e^{-2\frac{n+1}{n-1}\rho T}+n},\quad
        [m_-]^{2N}\to\frac{1}{ne^{2\frac{n+1}{n-1}\rho T}+1},\quad
        [\kappa]^{2N}\to\frac{e^{\frac{2n\rho T}{n-1}}}{ne^{2\frac{n+1}{n-1}\rho T}+1},\quad
        [\kappa-\hat{\kappa}]^{2N}\to\frac{e^{\frac{2n\rho T}{n-1}}}{ne^{2\frac{n+1}{n-1}\rho T}+1}$,
        \\
        $[m_+]^{2N+1}\to\frac{1}{-e^{-2\frac{n+1}{n-1}\rho T}+n},\quad
        [m_-]^{2N+1}\to\frac{1}{-ne^{2\frac{n+1}{n-1}\rho T}+1},\quad
        [\kappa]^{2N+1}\to\frac{e^{\frac{2n\rho T}{n-1}}}{ne^{2\frac{n+1}{n-1}\rho T}-1}$,
        \\
        and $[\kappa-\hat{\kappa}]^{2N+1}\to\frac{e^{\frac{2n\rho T}{n-1}}}{-ne^{2\frac{n+1}{n-1}\rho T}+1}$;
    	\item\label{h'} $\frac{m_++\kappa-\hat{\kappa}}{m_++\alpha(\kappa-\hat{\kappa})} \to  {\frac{2}{n+1}}, \quad  \frac{m_-+\alpha^2\kappa}{m_-+\alpha\kappa} \to  {\frac{2}{n+1}}$, \ and \ $\frac{\kappa+\alpha(\kappa-\hat{\kappa})}{1-\alpha^2} \to  {\frac{n-1}{4}}$.
    \end{enumerate}
\end{lemma}

\begin{proof}
    We start with \ref{lemma_asympt_rate_various_elements:label_b}. Let
    \begin{align*}
        R
        &=\sqrt{\alpha^4(\kappa-n)^2-2\alpha^2\big(\kappa(\kappa+1)+n(1-\kappa)\big)+(\kappa+1)^2}
        \\
        &=\hat{\kappa}
        +\Delta 2\tonde{\kappa\frac{n+1}{\hat{\kappa}}-n}
        +\Delta^2\frac{2\kappa^2+3n^2-5\kappa n-\kappa-n-2({\kappa\frac{n+1}{\hat{\kappa}}-n})^2}{\hat{\kappa}}
        +o  ({\Delta^2}),
        \end{align*}
    for $c_+$, set $R=\hat{\kappa}+\Delta L_R+\Delta^2 B_R+o(\Delta^2)$ and compute
    \begin{align*}
        c_+
        &= \frac{1-(1-\Delta)^2(\kappa+n)+\kappa+R}{2R}\\
        &= \frac{\Delta(2(\kappa+n)+L_R)
           + \Delta^2(B_R-(\kappa+n))
           + o(\Delta^2)}
           {2\hat{\kappa}(1+\Delta L_R/\hat{\kappa}
           + \Delta^2 B_R/\hat{\kappa}
           + o(\Delta^2))}\\
        &= \Delta\,\frac{2\kappa n}{\hat{\kappa}^2}
           - \Delta^2\,\kappa n
             \frac{4\kappa n+8\kappa-3n^2-2n+5}{\hat{\kappa}^4}
           + o(\Delta^2).
    \end{align*}
    Analogous expansions yield $d_\pm$ and $m_\pm$.
    For \ref{lemma_asympt_rate_various_elements:label_c}, write
    \[
        c_+ = \Delta L_{c_+} + \Delta^2 B_{c_+} + o(\Delta^2),
        \qquad
        m_+ - \kappa = \Delta L_{m_+} + \Delta^2 B_{m_+} + o(\Delta^2),
    \]
    and compute
    \begin{align*}
        \frac{c_+}{m_+ - \kappa}
        &= \frac{L_{c_+} + \Delta B_{c_+} + o(\Delta)}
           {L_{m_+}(1 + \Delta B_{m_+}/L_{m_+} + o(\Delta))}= \left(\frac{n}{\hat{\kappa}}
           + \Delta \frac{\hat{\kappa}}{2\kappa} B_{c_+}
           + o(\Delta)\right)
           \left(1 - \Delta \frac{\hat{\kappa}}{2\kappa} B_{m_+}
           + o(\Delta)\right)\\
        &= \frac{n}{\hat{\kappa}}
           + \Delta\left(-\frac{n}{2\kappa} B_{m_+}
           + \frac{\hat{\kappa}}{2\kappa} B_{c_+}\right)
           + o(\Delta)
        = \frac{n}{\hat{\kappa}}
           + \Delta \frac{2n(\hat{\kappa} - 2\kappa)}{\hat{\kappa}^3}
           + o(\Delta).
    \end{align*}
    The remaining ratios in
    \ref{lemma_asympt_rate_various_elements:label_c}--\ref{lemma_asympt_rate_various_elements:label_d}
    follow similarly. Item \ref{lemma_asympt_rate_various_elements:label_e} and the limits in
    \ref{lemma_asympt_rate_various_elements:label_f} are obtained by the same ideas used in the proof of Theorem~\ref{strat_asympt_thm_theta_positive}\ref{conv_of_W}.
    For \ref{lemma_asympt_rate_various_elements:label_g}, note that
    \[
    [m_+]^{N}
    = \frac{\Delta}{\delta_{N+1}} m_+^N
    = \frac{\Delta}{
        c_+\left(1-\alpha^2+\kappa-\frac{\alpha^2\kappa(\kappa-\hat{\kappa})}{m_+}\right)
        + c_-\left(1-\alpha^2+\kappa-\frac{\alpha^2\kappa(\kappa-\hat{\kappa})}{m_-}\right)
          \left(\frac{m_-}{m_+}\right)^{N}
      }.
    \]
    Expanding the denominator in $\Delta$ and observing that the second term decays faster than any power of $1/N$, we only need the expansion of $c_+\Big(1-\alpha^2+\kappa-\frac{\alpha^2\kappa({\kappa-\hat{\kappa}})}{m_+}\Big)$.
    Writing
    \[
    c_+ = \Delta a_1 + \Delta^2 a_2 + o(\Delta^2),
    \qquad
    m_+ = \kappa + \Delta b_1 + \Delta^2 b_2 + o(\Delta^2),
    \]
    we compute
    \begin{align*}
        &c_+\left(1-\alpha^2+\kappa-\frac{\alpha^2\kappa(\kappa-\hat{\kappa})}{m_+}\right)
        = \Delta(a_1 + \Delta a_2 + o(\Delta))
          \left(\hat{\kappa}
          + \Delta\big((\kappa-\hat{\kappa})(\hat{b}+2) + 2\big)
          + o(\Delta)\right),
    \end{align*}
    where $\hat{b} = b_1/\kappa$.
    After some algebra, we arrive at
    \[
        c_+\left(1-\alpha^2+\kappa-\frac{\alpha^2\kappa(\kappa-\hat{\kappa})}{m_+}\right)
        = \Delta\left(\frac{2\kappa n}{\hat{\kappa}}
        + \Delta\,\kappa n\,\frac{10\hat{\kappa}-8\kappa-n^2+1}{\hat{\kappa}^3}
        + o(\Delta)\right),
    \]
    thus 
    \begin{align*}
        [m_+]^N
        &= \frac{\Delta}{
            \Delta\left(
              \frac{2\kappa n}{\hat{\kappa}}
              + \Delta\kappa n\,\frac{10\hat{\kappa}-8\kappa-n^2+1}{\hat{\kappa}^3}
              + o(\Delta)
            \right)
          }
        = 
        \frac{\hat{\kappa}}{2\kappa n}
           \left(1+\Delta\,\frac{\hat{\kappa}^2-8\hat{\kappa}+8\kappa}{2\hat{\kappa}^2}
           + o(\Delta)\right).
    \end{align*}
    The remaining expansions in
\ref{lemma_asympt_rate_various_elements:label_g}--\ref{lemma_asympt_rate_various_elements:label_h} follow analogously.
    Items \ref{f'}--\ref{h'} follow by L'Hôpital's rule and straightforward algebra. %
\end{proof}

\begin{proof}[Proof of Theorem~\ref{strat_asympt_thm_theta_positive}~\ref{conv_of_V} for $\kappa\neq n-1$]
    Let $\kappa>\frac{n-1}{2}$ with $\kappa\neq n-1$. Starting from \eqref{ae}--\eqref{af}, we expand each term using the asymptotics in Lemma~\ref{lemma_asympt_rate_various_elements}.
    
    \noindent\emph{Step 1: Expansion of $\sum_{i=1}^{n_t}\nu_i$.}

    Let $t\in(0,T]$ and consider \eqref{ae} with $m=n_t$;  we expand each of the four terms in \eqref{ae} as $N\uparrow\infty$.
    \begin{enumerate}[label=\textup{\arabic*)},leftmargin=2.2em]
    \item\label{it:term1}
        \begin{align*}
        \sums\frac{d_\sigma(m_\sigma-\alpha^2\kappa)}{m_\sigma-\alpha\kappa}[m_\sigma]^N
        &= \frac{n-1}{\kappa(n+1)}
          + \frac{1}{N}\underbrace{\frac{2\rho T(-n^2+8n\theta+1)}
          {(n-1)(n+1)^2(n-1+4\theta)}}_{=:~\mathscr{Y}^{(1)}}
          + o\left(\frac{1}{N}\right).
      \end{align*}
    \item\label{it:term2}
        \begin{multline*}
        (1-\alpha)(n_t-1)\sums c_\sigma d_\sigma
        \left(
          \frac{\alpha(\kappa-\hat{\kappa})}{m_\sigma-\alpha(\kappa-\hat{\kappa})}
          + \frac{m_\sigma}{m_\sigma-\alpha\kappa}
        \right)[m_\sigma]^N\\
        = \frac{\rho t}{n+1}
          + \frac{1}{N}\underbrace{
              \frac{\rho T}{n+1}
              \left(
                \eta_t^N - 1 - \frac{\rho t}{2}
                + \frac{\rho t(n-4\theta+1)}{2(n+1)}
              \right)
            }_{=:~\mathscr{Y}_N^{(2)}(t)}
          + o\left(\frac{1}{N}\right).
      \end{multline*}
    \item\label{it:term3}
        \begin{multline*}
            C_1\left(
              1+\sums
              \frac{c_\sigma m_\sigma((m_\sigma/(\alpha\kappa))^{n_t-1}-1)}{m_\sigma-\alpha\kappa}
            \right)\alpha^N[\kappa]^N\\
            = \frac{e^{-\rho\frac{n+1}{n-1}T}}{n(n+1)^2}
              (-(n-1)+2n e^{\rho\frac{n+1}{n-1}t})
              + \frac{\mathscr{Y}_N^{(3)}(t)}{N}
              + o\left(\frac{1}{N}\right),
      \end{multline*}
        where
        \begin{align*}
        \mathscr{Y}_N^{(3)}(t)
        &= e^{-\rho T\frac{n+1}{n-1}}
           (c^{(0)} + c_N^{(1)}(t)e^{\rho t\frac{n+1}{n-1}}),\\[2pt]
        c^{(0)}
        &= \frac{\rho T}{n(n-1)^2(n+1)^3}(
             -8Tn^2\rho\theta - 8Tn\rho\theta + n^4 + 4n^3\theta
             - 20n^2\theta - 2n^2 + 12n\theta + 4\theta + 1
           ),\\[2pt]
        c_N^{(1)}(t)
        &= \frac{2\rho T}{(n-1)^3(n+1)^3}(
             \eta_t^N(n^4 - 2n^2 + 1)
             - 8\rho t(n^2\theta + n\theta)
             + 8T\rho n^2\theta + 8T\rho n\theta
             - n^4 - 6n^3\theta\\
        &\hspace{11em}
             + 14n^2\theta + 2n^2 - 10n\theta + 2\theta - 1
           ).
      \end{align*}
    \item\label{it:term4}
        \begin{multline*}
            nC_1\sums\frac{d_\sigma m_\sigma  \left(\frac{\alpha(\kappa-\hat\kappa)}{m_\sigma}
            - \bigl(\frac{\alpha(\kappa-\hat\kappa)}{m_\sigma}\bigr)^{n_t}\right)}
            {m_\sigma-\alpha(\kappa-\hat\kappa)}[m_\sigma]^N\\
            = \frac{\kappa-\hat\kappa}{\kappa(n+1)}
            + \frac{1}{N}\underbrace{\frac{\rho T (n-4\theta-1) (n^2+4n\theta+2n+1)}
            {(n-1)(n+1)^2 (n+4\theta-1)}}_{=:~\mathscr{Y}^{(4)}}
            +o \left(\frac{1}{N}\right).
        \end{multline*}
    \end{enumerate}
    Collecting the first-order coefficients of \ref{it:term1}, \ref{it:term2}, and \ref{it:term4}, set
    \begin{align*}
      \mathscr{Y}_N(t)
      &:= \mathscr{Y}^{(1)} + \mathscr{Y}_N^{(2)}(t) + \mathscr{Y}^{(4)}\\
      &= \frac{2\rho T(-n^2+8n\theta+1)}{(n-1)(n+1)^2(n-1+4\theta)}
        + \frac{\rho T}{n+1}\left(\eta_t^N - 1 - \frac{\rho t}{2}
        + \frac{\rho t(n-4\theta+1)}{2(n+1)}\right)\\
      &\qquad + \frac{\rho T(n-4\theta-1)(n^2+4n\theta+2n+1)}
        {(n-1)(n+1)^2(n+4\theta-1)}.
    \end{align*}

    Hence
    \begin{align}
      \sum_{i=1}^{n_t}\nu_i
      &= \frac{e^{-\rho\frac{n+1}{n-1}T}}{(n+1)^2 n}
         (
           n(\rho t+1)(n+1)e^{\rho\frac{n+1}{n-1}T}
           - (n-1)
           + 2n e^{\rho\frac{n+1}{n-1}t}
         )\notag\\
      &\quad + \frac{\mathscr{Y}_N(t)+\mathscr{Y}_N^{(3)}(t)}{N}
         + o\left(\frac{1}{N}\right),
      \label{sum nuit kappa>1/2, kappa neq 1 eq}
    \end{align}
    and $\mathscr{Y}_N(t)+\mathscr{Y}_N^{(3)}(t)$ is bounded once $\theta,T,\rho,n$ are fixed.

    \noindent\emph{Step 2: Expansion of \(\nu_{N+1}\).}

    From \eqref{af} and Lemma~\ref{lemma_asympt_rate_various_elements},
    \[
      \nu_{N+1}
      = \frac{\mathscr{T}}{N} + o\left(\frac{1}{N}\right),
      \qquad
      \mathscr{T} := \frac{\rho T}{\hat{\kappa}}.
    \]

    \smallskip
    
    \noindent\emph{Step 3: Expansion of \(\sum_{i=1}^{N+1}\nu_i\).}

    Since $n_T = N$, \eqref{sum nuit kappa>1/2, kappa neq 1 eq} at $t=T$ gives $\sum_{i=1}^{N}\nu_i$; adding $\nu_{N+1}$ yields
    \begin{align}\label{sum nui kappa>1/2, kappa neq 1 eq}
      \sum_{i=1}^{N+1}\nu_i
      &= \frac{e^{-\rho\frac{n+1}{n-1}T}}{(n+1)^2 n}
         (
           n((\rho T+1)(n+1)+2)e^{\rho\frac{n+1}{n-1}T}
           - (n-1)
         )
         + \frac{\mathscr{M}}{N}
         + o\left(\frac{1}{N}\right),
    \end{align}
    with $\mathscr{M} := \mathscr{Y}_N(T) + \mathscr{Y}_N^{(3)}(T) + \mathscr{T}$ (note that $\eta_T^N = 0$, so $\mathscr{M}$ is independent of $N$).

    \smallskip
    
    \noindent\emph{Step 4: Expansion of \(V_t^{(N)}\).}

     The limit in \eqref{sum nui kappa>1/2, kappa neq 1 eq} matches the right-hand side of \eqref{sum nuit kappa>1/2, kappa neq 1 eq} at $t=T$, and \eqref{sum nuit kappa>1/2, kappa neq 1 eq} matches the limit from \eqref{sumnuiupto_t} (obtained when $\kappa=n-1$). Although the leading coefficient of the $1/N$ term depends on $\theta$, the convergence order remains $1/N$ for every $\theta>0$. Plugging these into the definition of $V^{(N)}$ and applying \eqref{quotient_expansion} once more yields the claim.
\end{proof}

\subsection{Proof of Theorem~\ref{costs asymptotics thm}}

Let $(\bm\xi_1,\dots,\bm\xi_n)$ be the equilibrium profile; we drop the star superscript and suppress the $N$-dependence of $\bm\xi_i$ and related quantities to keep notation light. We start with a simple lemma (valid for all $\kappa\ge \hat{\kappa}/2$). 

\begin{lemma}
For all $i=1,\dots,n$,
\begin{align}\label{am}
\mathbb{E}\left[\mathscr{C}_\mathbb{T}\left(\bm{\xi}_i \mid \bm{\xi}_{-i}\right)\right]
&= \frac{1}{2}\Bigg(
      \frac{\bar{x}^2}{\bm{1}^\top\bm{\nu}}
    + \frac{\bar{x}(x_i-\bar{x})(\bm{1}^\top\bm{\nu}+\bm{1}^\top\bm{\omega})}
            {(\bm{1}^\top\bm{\nu})(\bm{1}^\top\bm{\omega})}
    + \frac{(x_i-\bar{x})^2}{\bm{1}^\top\bm{\omega}}\\
&\hspace{3em}
    + \hat{\kappa}\left(\frac{\bar{x}}{\bm{1}^\top\bm{\nu}}\right)^2  
      \bm{\nu}^\top\tilde{\Gamma}\bm{\nu}
    + \frac{\bar{x}(x_i-\bar{x})}{(\bm{1}^\top\bm{\nu})(\bm{1}^\top\bm{\omega})}
      \bm{\omega}^\top(\hat{\kappa}\tilde{\Gamma}-\tilde{\Gamma}^\top)\bm{\nu}
    - \left(\frac{x_i-\bar{x}}{\bm{1}^\top\bm{\omega}}\right)^2
      \bm{\omega}^\top\tilde{\Gamma}\bm{\omega}
\Bigg),\nonumber
\end{align}
where $\bar{x} = \frac{1}{n}\sum_{i=1}^n x_i$.
\end{lemma}

\begin{proof}
By Lemma~\ref{formoffunctionaldisc},
\begin{align}\label{formoffunctlemma}
\mathbb{E}[\mathscr{C}_\mathbb{T}(\bm{\xi}_i \mid \bm{\xi}_{-i})]
= \frac{1}{2}\bm{\xi}_i^\top\Gamma^\theta\bm{\xi}_i
  + \bm{\xi}_i^\top\tilde{\Gamma}\sum_{j\neq i}\bm{\xi}_j
=: A+B.
\end{align}
By Theorem~\ref{exuniqnasheqdiscrete}, $\bm{\xi}_i=\bar{x}\bm{v}+(x_i-\bar{x})\bm{w}$, hence
\begin{align*}
A
&= \frac{1}{2}\Big(
    \bar{x}^2\bm{v}^\top\Gamma^\theta\bm{v}
  + \bar{x}(x_i-\bar{x})\bm{v}^\top\Gamma^\theta\bm{w}
  + \bar{x}(x_i-\bar{x})\bm{w}^\top\Gamma^\theta\bm{v}
  + (x_i-\bar{x})^2\bm{w}^\top\Gamma^\theta\bm{w}
\Big).
\end{align*}
Moreover, since $\sum_{j\neq i}\bm{\xi}_j=\hat{\kappa}\bar{x} \bm{v}+(\bar{x}-x_i)\bm{w}$,
\begin{align*}
B
&= \left(\bar{x}\bm{v}+\left(x_i-\bar{x}\right)\bm{w}\right)^\top
   \tilde{\Gamma}  \left(\hat{\kappa}\bar{x} \bm{v}+(\bar{x}-x_i)\bm{w}\right)
\\
&= \hat{\kappa}\bar{x}^{ 2}\bm{v}^\top\tilde{\Gamma}\bm{v}
 + \hat{\kappa}\bar{x}  \left(x_i-\bar{x}\right)\bm{w}^\top\tilde{\Gamma}\bm{v}
 - \bar{x}  \left(x_i-\bar{x}\right)\bm{v}^\top\tilde{\Gamma}\bm{w}
 - \left(x_i-\bar{x}\right)^{  2}\bm{w}^\top\tilde{\Gamma}\bm{w}.
\end{align*}
Substituting into \eqref{formoffunctlemma} and using
\[
(\Gamma^\theta+\hat{\kappa}\tilde{\Gamma})\bm{\nu}=\bm{1},
\qquad
(\Gamma^\theta-\tilde{\Gamma})\bm{\omega}=\bm{1},
\]
together with $\bm{\nu}^\top\bm{1}=\bm{1}^\top\bm{\nu}$,
$\bm{\omega}^\top\bm{1}=\bm{1}^\top\bm{\omega}$, and
$\bm{\nu}^\top\tilde{\Gamma}\bm{\omega}=\bm{\omega}^\top\tilde{\Gamma}^\top\bm{\nu}$,
and writing $\bm{v}=\bm{\nu}/(\bm{1}^\top\bm{\nu})$ and
$\bm{w}=\bm{\omega}/(\bm{1}^\top\bm{\omega})$, we obtain \eqref{am}.
\end{proof}

\begin{lemma}\label{cost functional ausiliral lemma}
For $\kappa>\hat{\kappa}/2$, as $N\uparrow\infty$,
\begin{align*}
{\bm\nu}^\top \tilde{\Gamma} {\bm\nu}
&\longrightarrow
\frac{n-1}{2n^2(n+1)^3}\Biggl(
 -e^{-2\rho\frac{n+1}{ n-1 }T}
 - 4n e^{-\rho\frac{n+1}{ n-1 }T}
 + \frac{2n^2(n+1)}{n-1} \rho T
 + \frac{n^2(n+7)}{n-1}
\Biggr),\\[0.5em]
{\bm\omega}^\top  \bigl(\hat{\kappa}\tilde{\Gamma}-\tilde{\Gamma}^\top\bigr){\bm\nu}
&\longrightarrow
\frac{-(n-1)(2n-1) e^{-\rho\frac{n+1}{ n-1 }T}
      + n(n+4)(n-1)
      + n(n+1)(n-2) \rho T}{n(n+1)^2},\\[0.5em]
{\bm\omega}^\top \tilde{\Gamma} {\bm\omega}
&\longrightarrow \frac{2\rho T+1}{2} .
\end{align*}
\end{lemma}

\begin{proof}
    The third limit follows from \cite[Lemma~A.5]{SchiedStrehleZhang.17}; we prove the first two.
    
    \noindent\emph{Step 1: Case $\kappa=\hat{\kappa}$ (so $\tilde{\kappa}=n/2$).} We have
    \begin{align*}
    {\bm\nu}^\top \tilde{\Gamma} {\bm\nu}
    &= \frac{\nu_1^{ 2}}{2}
     + \frac12\sum_{i=2}^{N+1}\nu_i^{ 2}
     + \nu_1 \sum_{i=2}^{N+1}\nu_i\alpha^{ i-1}
     + \sum_{i=3}^{N+1}\sum_{j=2}^{i-1}\nu_i\nu_j \alpha^{ i-j}\\
    &= \frac{1}{(n+\alpha)^2}\Biggl[
     \frac{(1-\alpha^2)N}{2}
     + \frac{-\alpha^4 + 2\alpha^3(n-2) + \alpha^2(2n-2) + 2n\alpha + n^2}{2\left(n^2-\alpha^2\right)}\\[-0.15em]
    &\hspace{4.2em}
     - \theta^{N} \frac{(n-1)\alpha^2(\alpha+1)}{n(n+\alpha)}
     - \theta^{2N} \frac{\alpha^4\hat{\kappa}^2}{2n^2\left(n^2-\alpha^2\right)}
    \Biggr],
    \end{align*}
    and therefore, as $N\uparrow\infty$,
    \begin{align*}
    {\bm\nu}^\top \tilde{\Gamma} {\bm\nu}
    &\longrightarrow
    \frac{1}{(n+1)^2}
    \Biggl(
     \rho T + \frac{n+7}{2(n+1)}
     - \frac{2(n-1)}{n(n+1)} e^{-\rho\frac{n+1}{ n-1 }T}
     - \frac{n-1}{2n^2(n+1)} e^{-2\rho\frac{n+1}{ n-1 }T}
    \Biggr)\\
    &=
    \frac{n-1}{2n^2(n+1)^3}\Biggl(
     -e^{-2\rho\frac{n+1}{ n-1 }T}
     - 4n e^{-\rho\frac{n+1}{ n-1 }T}
     + \frac{2n^2(n+1)}{n-1} \rho T
     + \frac{n^2(n+7)}{n-1}
    \Biggr)
    \end{align*}
    as well as
    \begin{align*}
        {\bm\omega}^\top \left(\hat{\kappa}\tilde{\Gamma}-\tilde{\Gamma}^\top \right){\bm\nu}
        &= \frac{n-2}{2}\left(\nu_1\omega_1 + \sum_{i=2}^{N+1}\nu_i\omega_i\right)
         + \hat\kappa\nu_1 \sum_{i=2}^{N+1}  \omega_i \alpha^{i-1}
         + \hat{\kappa}\omega_{N+1}\sum_{i=2}^N \nu_i \alpha^{N+1-i} \\
        &\quad + \sum_{i=1}^N  \omega_i\left(\hat{\kappa}\sum_{j=2}^{i-1} \nu_j \alpha^{i-j} - \sum_{j=i+1}^{N+1} \nu_j \alpha^{j-i}\right)\\
        &\to \frac{-(n-1)(2n-1)e^{-\rho\frac{n+1}{n-1}T} + n(n+4)(n-1) + n(n+1)(n-2)\rho T}{n(n+1)^2}.
    \end{align*}

    \noindent\emph{Step 2: General case $\kappa \ge \frac{n-1}{2}$ with $\kappa \neq n-1$.}
    We include the boundary value $\kappa=\frac{n-1}{2}$ because intermediary limits below will also be used for that case. We first compute $\tilde{\Gamma}\bm{\nu}$. Define $C_1= \frac{\alpha\left(\alpha+1\right)}{\kappa +1 -\alpha\left(\kappa -\hat{\kappa}-1\right)}$ as above and
    \begin{align*}
    C_2&:=\sums c_\sigma d_\sigma\left(\frac{\alpha(\kappa-\hat{\kappa})}{m_\sigma-\alpha(\kappa-\hat{\kappa})}+\frac{m_\sigma}{m_\sigma-\alpha\kappa}\right)  [m_\sigma]^N,\\
    C_3&:=-C_2 + \sums d_\sigma\left( \frac{m_\sigma - \alpha^2\kappa}{m_\sigma-\alpha\kappa} + \frac{nC_1(\kappa-\hat{\kappa})}{m_\sigma-(\kappa-\hat{\kappa})} \right)[m_\sigma]^N .
    \end{align*}
    For $\sigma\in\{+,-\}$, write $\bar{\sigma}=-$ if $\sigma=+$ and $\bar{\sigma}=+$ if $\sigma = -$. Then
    \begin{align*}
    (\tilde{\Gamma}{\bm\nu})_1 &= \frac 12\sums\frac{d_\sigma(m_\sigma-\alpha^2\kappa)}{m_\sigma-\alpha\kappa}  [m_\sigma]^N+\frac{C_1 \alpha^N}{2}   [\kappa]^N,\\
    (\tilde{\Gamma}{\bm\nu})_2 &= \sums d_\sigma\left(\frac{\alpha(m_\sigma-\alpha^2\kappa)}{m_\sigma-\alpha\kappa}+\frac{nC_1\alpha(\kappa-\hat{\kappa}) }{2m_\sigma}\right)  [m_\sigma]^N \\
    &\hspace{1cm}{}+ \frac{C_2(1-\alpha)}{2} + \frac{C_1(1+\alpha^2\left(2\kappa-n\right)+\kappa)\alpha^N}{2\alpha\kappa}  [\kappa]^N,
    \end{align*}
    and, for $i\in \{3,\dots,N\}$,
    \begin{align*}
    (\tilde{\Gamma}{\bm\nu})_i &= \frac{C_2\left(1+\alpha\right)}{2}+\frac{nC_1}{2\alpha(\kappa-\hat{\kappa})} \sums \frac{d_\sigma m_\sigma\left(\hat{\kappa}-\kappa-m_\sigma\right)   [m_\sigma]^N}{\hat{\kappa}-\kappa+m_\sigma}\left(\frac{(\kappa-\hat{\kappa})\alpha}{m_\sigma}\right)^i\\
    &\hspace{1cm}{}+ \frac{C_1\alpha^{N+1}\kappa  [\kappa]^N}{2} \sums \frac{c_\sigma\left(m_\sigma+\alpha^2\kappa\right)}{m_\sigma\left(m_\sigma-\alpha^2\kappa\right)}\left(\frac{m_\sigma}{\alpha\kappa}\right)^i+C_3\alpha^{i-1}.
    \end{align*}
    Moreover,
    \begin{align*}
    (\tilde{\Gamma}{\bm\nu})_{N+1} &= \sums c_\sigma\left(\frac{d_\sigma m_\sigma \alpha }{m_\sigma-\alpha\kappa}+\frac{m_\sigma+(2d_\sigma-1)\alpha^2(\kappa-\hat{\kappa})}{2(m_\sigma-\alpha(\kappa-\hat{\kappa}))}+\frac{C_1\alpha^2\kappa}{m_\sigma-\alpha^2\kappa}\right)  [m_\sigma]^N\\
    &\hspace{.5cm} {}+C_3\alpha^N -C_1\alpha^N\tilde{\kappa}\left[\kappa-\hat{\kappa}\right]^N,
    \end{align*}
    whose last term can also be written as $-\frac{C_1\alpha^N}{2}( 2\kappa-\hat{\kappa}+1 )[\kappa-\hat{\kappa}]^N$.
    
    Next, for $i\in \{3,\dots,N\}$,
    \[
    \nu_i  (\tilde{\Gamma} {\bm\nu})_{i}  =  D_1^i + D_2^i + D_3^i + D_4^i,
    \]
    with $D_1^i := \dfrac{C_2\left(1+\alpha\right)}{2}  \nu_i$ and
    \begin{align*}
    D_2^i := C_2 (1-\alpha) \Bigg(&C_3\alpha^{i-1}+\frac{nC_1}{2\alpha(\kappa-\hat{\kappa})} \sums \frac{d_\sigma m_\sigma(\hat{\kappa}-\kappa-m_\sigma)   [m_\sigma]^N}{\hat{\kappa}-\kappa+m_\sigma}\left(\frac{(\kappa-\hat{\kappa})\alpha}{m_\sigma}\right)^i\\
    &\hspace{3.5cm}{}+ \frac{C_1\alpha^{N+1}\kappa  [\kappa]^N}{2} \sums \frac{c_\sigma(m_\sigma+\alpha^2\kappa)}{m_\sigma(m_\sigma-\alpha^2\kappa)}\left(\frac{m_\sigma}{\alpha\kappa}\right)^i\Bigg),
    \end{align*}
    \begin{align*}
    D_3^i := \frac{C_1C_3}{\alpha} \left(n\sums \frac{d_\sigma m_\sigma   [m_\sigma]^N}{\alpha(\kappa-\hat{\kappa})}\left(\frac{\alpha^2(\kappa-\hat{\kappa})}{m_\sigma}\right)^i +\sums \frac{c_\sigma \alpha^{N+1}\kappa   [\kappa]^N}{m_\sigma}\left(\frac{m_\sigma}{\kappa}\right)^i\right),
    \end{align*}
    and
    \begin{align*}
    D_4^i&:=\left(C_1\right)^2\Bigg( \frac{n^2}{2(\alpha(\kappa-\hat{\kappa}))^2}\Bigg(\sums\left( d_\sigma m_\sigma [m_\sigma]^N \right)^2\frac{\hat{\kappa}-\kappa-m_\sigma}{\hat{\kappa}-\kappa+m_\sigma}\left(\frac{\alpha(\kappa-\hat{\kappa})}{m_\sigma}\right)^{2i} \\
    &\hspace{4.5cm} + d_+d_-m_+m_-\left[m_+\right]^N\left[m_-\right]^N\left( \frac{2\left((\hat{\kappa}-\kappa)^2-m_+m_-\right)}{n(1-\alpha^2)(\hat{\kappa}-\kappa)} \right)\left(\frac{\left(\alpha(\kappa-\hat{\kappa})\right)^2}{m_+m_-}\right)^{i}\Bigg)\\
    &\hspace{2cm}+\frac{\alpha^{2(N+1)}\kappa^2\left([\kappa]^N\right)^2}{2}\Bigg(\sums \frac{(c_\sigma)^2(m_\sigma+\alpha^2\kappa)}{(m_\sigma)^2(m_\sigma-\alpha^2\kappa)}\left(\frac{m_\sigma}{\alpha\kappa}\right)^{2i}\\
    &\hspace{6cm} +\frac{c_+c_-2\left( (\alpha^2\kappa)^2-m_+m_- \right)}{nm_+m_-\alpha^2(1-\alpha^2)\kappa}\left(\frac{m_+m_-}{(\alpha\kappa)^2}\right)^i  \Bigg)\\
    &\hspace{2cm} + \frac{\alpha^N\kappa\left[ \kappa \right]^Nn}{2(\kappa-\hat{\kappa})}\Bigg(\frac{2(\hat{\kappa} - (1-\alpha^2)\kappa)}{n(1-\alpha^2)}\sums d_\sigma c_\sigma\left[ m_\sigma \right]^N\left(\frac{\kappa-\hat{\kappa}}{\kappa}\right)^i\\
    &\hspace{4.7cm} + \sums\frac{c_{\bar{\sigma}}d_\sigma m_\sigma[m_\sigma]^N}{m_{\bar{\sigma}}}\left( \frac{\hat{\kappa} - \kappa -m_\sigma}{\hat{\kappa} - \kappa +m_\sigma} +\frac{m_{\bar{\sigma}} +\alpha^2\kappa }{m_{\bar{\sigma}} -\alpha^2\kappa} \right)\left( \frac{m_{\bar{\sigma}}(\kappa - \hat{\kappa})}{m_\sigma\kappa} \right)^i \Bigg)\Bigg).
    \end{align*}
    Summing over $i$, we obtain
    \begin{align*}
    \sum_{i=3}^N D^i_1 &= \frac{C_2(1+\alpha)}{2}\Bigg( C_2(1-\alpha)(N-2) + C_1\sums\frac{m_\sigma c_\sigma\left( \frac{\alpha\kappa}{m_\sigma}[m_\sigma]^N - \frac{m_\sigma}{\alpha\kappa}\alpha^N[\kappa]^N \right)}{m_\sigma -\alpha\kappa} \\
    &\hspace{2.65cm} +nC_1\sums\frac{d_\sigma m_\sigma\left(  \left( \frac{\alpha(\kappa-\hat{\kappa})}{m_\sigma} \right)^2[m_\sigma]^N - \alpha^N\left[\kappa-\hat{\kappa}\right]^N \right)}{m_\sigma - \alpha(\kappa - \hat{\kappa})} \Bigg),
    \end{align*}
    and 
    \begin{align*}
    \sum_{i=3}^N D_2^i&=C_2C_3(\alpha^2-\alpha^N)\\
    &\quad+\frac{C_2C_1(1-\alpha^2)n}{2(1+\alpha)}\sums\frac{(d_\sigma)^2(\hat{\kappa} - \kappa - m_\sigma)\left( (m_\sigma)^2\alpha^N\left[ \kappa-\hat{\kappa} \right]^N - (\alpha(\kappa-\hat{\kappa}))^2\left[ m_\sigma \right]^N\right)}{d_\sigma m_\sigma(m_\sigma - (\kappa-\hat{\kappa}))((\kappa-\hat{\kappa})\alpha-m_\sigma)}\\
    &\quad+\frac{C_2C_1(1-\alpha^2)}{2(1+\alpha)}\sums\frac{(c_\sigma)^2(m_\sigma + \alpha^2\kappa)((\alpha\kappa)^2\left[ m_\sigma \right]^N - (m_\sigma)^2\alpha^N\left[ \kappa \right]^N)}{c_\sigma\alpha\kappa(m_\sigma-\alpha\kappa)(m_\sigma-\alpha^2\kappa)},
    \end{align*}
    \begin{align*}
    \sum_{i=3}^N D_3^i&=C_1C_3\Bigg( \sums \frac{nd_\sigma \left( (m_\sigma\alpha^N)^2\left[\kappa-\hat{\kappa}\right]^N -\left({\alpha^2}(\kappa-\hat{\kappa})\right)^2[m_\sigma]^N\right)}{m_\sigma(\alpha^2(\kappa-\hat{\kappa})-m_\sigma)} \\
    &\hspace{2.65cm}+\sums\frac{c_\sigma\left( \kappa^2\alpha^N\left[ m_\sigma \right]^N - (m_\sigma)^2\alpha^N[\kappa]^N \right)}{\kappa(m_\sigma-\kappa)}\Bigg),
    \end{align*}
    \begin{align}\label{fa}
        \lefteqn{\sum_{i=3}^N D_4^i}\nonumber\\
        &=(C_1)^2\Bigg(\frac{n^2}{2}\sums\frac{({\ds\ms})^2({\hat{\kappa}-\kappa-\ms})}{({\ms -(\kappa-\hat{\kappa})})({\alpha(\kappa-\hat{\kappa}) - \ms})({\ms +\alpha(\kappa-\hat{\kappa})})}({\alpha^N[{\kappa-\hat{\kappa}}]^N})^2\nonumber\\
        &\hspace{.5cm}- \frac{1}{2(\alpha\kappa)^2}\sums\frac{(\cs)^2(\ms)^4(\ms + \alpha^2\kappa)}{(\ms-\alpha\kappa)(\ms -\alpha^2\kappa)\tonde{\ms +\alpha\kappa}}\tonde{\alpha^N\quadre{\kappa}^N}^2\nonumber\\
        &\hspace{.5cm}+\sums\Bigg(\frac{(\cs)^2\tonde{\ms+\alpha^2\kappa}(\alpha\kappa)^2}{2(\ms-\alpha\kappa)\tonde{\ms-\alpha^2\kappa}\tonde{\ms+\alpha\kappa}}\nonumber\\
        &\hspace{3cm}-\frac{n^2\tonde{\ds}^2\tonde{\alpha(\kappa-\hat{\kappa})}^4\tonde{\hat{\kappa}-\kappa-\ms}}{2(\ms)^2\tonde{\ms -(\kappa-\hat{\kappa})}\tonde{\alpha(\kappa-\hat{\kappa}) - \ms}\tonde{\ms +\alpha(\kappa-\hat{\kappa})}}\Bigg)\tonde{\quadre{\ms}^N}^2\nonumber\\
        &\hspace{.5cm}+\sums\alpha^N\kappa\tonde{ \frac{n\cs\dsb\msb\tonde{\frac{\hat{\kappa}-\kappa-\msb}{\hat{\kappa}-\kappa+\msb} +\frac{\ms+\alpha^2\kappa}{\ms-\alpha^2\kappa}}}{2\tonde{\ms(\kappa-\hat{\kappa})-\msb\kappa}} - \frac{\cs\ds\tonde{\hat{\kappa}-\tonde{1-\alpha^2}\kappa}}{\tonde{1-\alpha^2}\hat{\kappa}}}\quadre{\ms}^N\quadre{\kappa-\hat{\kappa}}^N\nonumber\\
        &\hspace{.5cm}+\sums\frac{\alpha^N(\kappa-\hat{\kappa})^2}{\kappa}\tonde{\frac{\cs\ds\tonde{\hat{\kappa}-\tonde{1-\alpha^2}\kappa}}{\hat{\kappa}\tonde{1-\alpha^2}}-\frac{n\csb\ds\tonde{\msb}^2\tonde{\frac{\hat{\kappa}-\kappa-\ms}{\hat{\kappa}-\kappa+\ms} + \frac{\msb + \alpha^2\kappa}{\msb - \alpha^2\kappa}}}{2\ms\tonde{\msb(\kappa-\hat{\kappa})-\ms\kappa}}}\quadre{\ms}^N\quadre{\kappa}^N\nonumber\\
        &\hspace{.5cm}+\kappa\tonde{\tonde{\alpha^2-1}\kappa+\hat{\kappa}}\tonde{\alpha^N}^2\tonde{ \frac{c_+c_-\tonde{\alpha(\kappa-\hat{\kappa})}^2}{n(\alpha\kappa)^2\tonde{1-\alpha^2}\hat{\kappa}}\tonde{\quadre{\kappa}^N}^2 -\frac{nd_+d_-}{\hat{\kappa}\tonde{1-\alpha^2}}\tonde{\quadre{\kappa-\hat{\kappa}}^N}^2}\nonumber\\
        &\hspace{.5cm}+\frac{\tonde{1-\alpha^2}\kappa-\hat{\kappa}}{n\kappa}\tonde{\frac{c_+c_-\kappa^2}{\hat{\kappa}\tonde{1-\alpha^2}}-\frac{n^2d_+d_-(\kappa-\hat{\kappa})^2}{\hat{\kappa}\tonde{1-\alpha^2}} }\quadre{m_+}^N[m_-]^N\Bigg).\nonumber
    \end{align}
    Note that
    \[
    \frac{c_+}{m_+(\kappa-\hat{\kappa})-m_-\kappa}
    = \left(R\left(\frac{d_-}{1-\alpha^2}\frac{1-\alpha^2}{c_+}\kappa+(\kappa-\hat{\kappa})\right)\right)^{-1}.
    \]
    Using Lemma~\ref{lemma_asympt_rate_various_elements}, the limits of the preceding sums combine to give
    \begin{align*}
    {\bm\nu}^\top \tilde{\Gamma} {\bm\nu}
    &= \nu_1 (\tilde{\Gamma} {\bm\nu})_1 + \nu_2 (\tilde{\Gamma} {\bm\nu})_2 + \sum_{k=1}^4 \sum_{i=3}^N D_k^i + \nu_{N+1} (\tilde{\Gamma} {\bm\nu})_{N+1}\\
    &\to \frac{(n-1)}{2n^2(n+1)^3}\left(-e^{-2\rho\frac{n+1}{n-1}T} - 4n e^{-\rho\frac{n+1}{n-1}T}  +\frac{2n^2(n+1)}{(n-1)} \rho T + \frac{n^2(n+7)}{(n-1)}\right).
    \end{align*}
    \medskip
    We now turn to ${\bm\omega}^\top \left(\hat{\kappa}\tilde{\Gamma}-\tilde{\Gamma}^\top \right){\bm\nu}$. Set
    \begin{align*}
    C_4 &:= \frac{\left(\alpha^2\left(\Tilde{\kappa}-1\right)-\Tilde{\kappa}\right)-\alpha\left(\frac{\alpha\left(\Tilde{\kappa}-1\right)}{\Tilde{\kappa}}\right)^{N+1}}{\left(\Tilde{\kappa}-\alpha\left(\Tilde{\kappa}-1\right)\right)\left(\alpha^2\left(\Tilde{\kappa}-1\right)-\Tilde{\kappa}\right)},\\ 
    C_5 &:= \left(\Tilde{\kappa}+\alpha\left(\Tilde{\kappa}-1\right) - \frac{(n-2)\tilde{\kappa}\left(\alpha^2(\tilde{\kappa}-1)-\tilde{\kappa}\right)}{\tilde{\kappa}-\alpha(\tilde{\kappa}-1)}\right)\frac{\alpha^2(\tilde{\kappa}-1)}{\Tilde{\kappa}^2\left(\alpha^2\left(\Tilde{\kappa}-1\right)-\Tilde{\kappa}\right)},\\
    C_6 &:=\frac{n-2}{2 \tilde{\kappa}\left(\tilde{\kappa}-\alpha(\tilde{\kappa}-1)\right)}.
    \end{align*}
    
    Then, for $i\in\{2,\dots,N\}$,
    \begin{align*}
    \bigl({\bm\omega}^{\top}(\hat{\kappa}\tilde{\Gamma}-\tilde{\Gamma}^{\top})\bigr)_1
    &= \frac{n-2}{2} \omega_1
      + \frac{(n-1)\alpha}{\tilde{\kappa}-\alpha\bigl(\tilde{\kappa}-1\bigr)}
        \left(1-\left(\frac{\alpha(\tilde{\kappa}-1)}{\tilde{\kappa}}\right)^{  N}\right),\\[0.35em]
    \bigl({\bm\omega}^{\top}(\hat{\kappa}\tilde{\Gamma}-\tilde{\Gamma}^{\top})\bigr)_i
    &= \frac{n-2}{2} \omega_i
      + C_4 \alpha^{i}
      + C_5\left(\frac{\alpha(\tilde{\kappa}-1)}{\tilde{\kappa}}\right)^{  N-i}
      + \frac{(n-2)\alpha}{\tilde{\kappa}-\alpha\bigl(\tilde{\kappa}-1\bigr)},
    \end{align*}
    and
    \begin{align*}
    \bigl({\bm\omega}^\top  (\hat{\kappa}\tilde{\Gamma}-\tilde{\Gamma}^\top )\bigr)_{N+1}
    = \frac{\alpha\Big(\alpha^N\left(\tilde{\kappa}-\alpha^2\left(\tilde{\kappa}-1\right)\right)\tilde{\kappa}+\alpha^2\left(\tilde{\kappa}-1\right)\left(\tilde{\kappa}-1+\left(\frac{\alpha^2\left(\tilde{\kappa}-1\right)}{\tilde{\kappa}}\right)^{  N}\right)-\tilde{\kappa}^2\Big)}{\tilde{\kappa}\left(\tilde{\kappa}-\alpha\left(\tilde{\kappa}-1\right)\right)\left(\tilde{\kappa}-\alpha^2\left(\tilde{\kappa}-1\right)\right)} + \frac{n-2}{2}\frac{1}{\tilde{\kappa}}.
    \end{align*}
    For $i \in \{2,\dots,N\}$, write
    \[
    \bigl({\bm\omega}^\top  (\hat{\kappa}\tilde{\Gamma}-\tilde{\Gamma}^\top ) \bigr)_i {\nu}_i = G^i_1 + G^i_2 + G^i_3,
    \]
    where
    \begin{align*}
    G^i_1&=C_2(1-\alpha)\left(\frac{n-2}{2}\omega_i + C_4\alpha^i +C_5\left(\frac{\alpha(\tilde{\kappa}-1)}{\tilde{\kappa}}\right)^{  N}\left(\frac{\tilde{\kappa}}{\alpha(\tilde{\kappa}-1)}\right)^{  i}+\frac{(n-2)\alpha}{\tilde{\kappa} - \alpha(\tilde{\kappa}-1)}\right),\\
    G^i_2&=nC_1\sums \frac{d_\sigma m_\sigma}{\alpha(\kappa-\hat{\kappa})}\Bigg(\frac{n-2}{2}\omega_i\left(\frac{\alpha(\kappa-\hat{\kappa})}{m_\sigma}\right)^{  i} + C_4\left(\frac{\alpha^2(\kappa-\hat{\kappa})}{m_\sigma}\right)^{  i} \\
    &\qquad\qquad+ C_5\left(\frac{\alpha(\tilde{\kappa}-1)}{\tilde{\kappa}}\right)^{  N}\left(\frac{\tilde{\kappa}(\kappa-\hat{\kappa})}{m_\sigma(\tilde{\kappa}-1)}\right)^{  i}
    +\frac{(n-2)\alpha}{\tilde{\kappa}-\alpha(\tilde{\kappa}-1)}\left(\frac{\alpha(\kappa-\hat{\kappa})}{m_\sigma}\right)^{  i}\Bigg)[m_\sigma]^N,\\
    G^i_3&=C_1\sums \frac{c_\sigma \alpha^{N+1}\kappa}{m_\sigma}\Bigg(\frac{n-2}{2}\omega_i\left(\frac{m_\sigma}{\alpha\kappa}\right)^{  i} + C_4\left(\frac{m_\sigma}{\kappa}\right)^{  i} +C_5\left(\frac{\alpha(\tilde{\kappa}-1)}{\tilde{\kappa}}\right)^{  N}\left(\frac{\tilde{\kappa}m_\sigma}{{\kappa}\alpha^2(\tilde{\kappa}-1)}\right)^{  i}\\
    &\qquad\qquad+\frac{(n-2)\alpha}{\tilde{\kappa}-\alpha(\tilde{\kappa}-1)}\left(\frac{m_\sigma}{\alpha\kappa}\right)^{  i} \Bigg)[\kappa]^N .
    \end{align*}
    \hide{
    Notice that
    \begin{align*}
       (\left.{\bm\omega}^\top  (\left.\hat{\kappa}\tilde{\Gamma}-\tilde{\Gamma}^\top \right.)\right.)_i = \frac{n-2}{2}\omega_i + C_4\alpha^i +C_5\tonde{\frac{\alpha\tonde{\tilde{\kappa}-1}}{\tilde{\kappa}}}^{N-i}+\frac{\tonde{n-2}\alpha}{\tilde{\kappa} - \alpha\tonde{\tilde{\kappa}-1}}
    \end{align*}
    }
    
    Summing over $i$,
    \begin{align*}
    \sum_{i=2}^N G_1^i
    &= C_2\left(1-\alpha\right)\vast[
         C_6\left(
            \left(1-\alpha\right)\left(N-1\right)\tilde{\kappa}
            + \frac{\alpha^2\left(\tilde{\kappa}-1\right)}{\tilde{\kappa}-\alpha\left(\tilde{\kappa}-1\right)}
              \left(1-\left(\frac{\alpha\left(\tilde{\kappa}-1\right)}{\tilde{\kappa}}\right)^{  N-1}\right)
         \right) \\
    &\hspace{4.6em}
       +  C_4 \frac{\alpha^{N+1}-\alpha^2}{\alpha-1}
       + \frac{C_5 \tilde{\kappa}}{\tilde{\kappa}-\alpha\left(\tilde{\kappa}-1\right)}
         \left(1-\left(\frac{\alpha\left(\tilde{\kappa}-1\right)}{\tilde{\kappa}}\right)^{  N-1}\right)
       + \frac{\left(n-2\right)\alpha}{\tilde{\kappa}-\alpha\left(\tilde{\kappa}-1\right)}\left(N-1\right)
    \vast].
    \end{align*}
    \hide{
    Notice that
    \begin{align*}
        \sum_{i=2}^N\omega_i = \frac{1}{\tilde{\kappa}\tonde{\tilde{\kappa}-\alpha\tonde{\tilde{\kappa}-1}}}\tonde{\tonde{1-\alpha}\tonde{N-1}\tilde{\kappa}+\frac{\alpha^2\tonde{\tilde{\kappa}-1}}{\tilde{\kappa}-\alpha\tonde{\tilde{\kappa}-1}}\tonde{1-\tonde{\frac{\alpha\tonde{\tilde{\kappa}-1}}{\tilde{\kappa}}}^{N-1}}}
    \end{align*}
    }
    Moreover, 
    \begin{align*}
        {\sum_{i=2}^NG^i_2}&\\
        &\hspace{-1cm}=\sums \frac{nC_1d_\sigma m_\sigma   }{\alpha(\kappa-\hat{\kappa})}\vast[C_6\vastt(\frac{\alpha\tonde{1-\alpha}\tilde{\kappa}(\kappa-\hat{\kappa})}{\alpha(\kappa-\hat{\kappa})-\ms}\tonde{\alpha^N\left[\kappa-\hat{\kappa}\right]^N - \frac{\alpha\tonde{
        \kappa-\hat{\kappa}}}{\ms}\left[ \ms \right]^{N}}\\
        &\hspace{2.8cm}+\frac{\alpha\tilde{\kappa}(\kappa-\hat{\kappa})}{\tilde{\kappa}(\kappa-\hat{\kappa})-\ms\tonde{\tilde{\kappa}-1}}\tonde{\frac{\tilde{\kappa}-1}{\tilde{\kappa}}\alpha^{N+1}\left[\kappa-\hat{\kappa}\right]^N-\frac{\kappa-\hat{\kappa}}{\ms}\alpha^{N+1}\tonde{\frac{\tilde{\kappa}-1}{\tilde{\kappa}}}^N\left[\ms\right]^{N}}\vastt)\\
        &\hspace{2.9cm}+ \frac{C_4\alpha^2(\kappa-\hat{\kappa})}{\alpha^2(\kappa-\hat{\kappa})-\ms}\tonde{\alpha^{2N}\left[\kappa-\hat{\kappa}\right]^N - \frac{\alpha^2(\kappa-\hat{\kappa})}{\ms}\left[\ms\right]^{N}} \\
        &\hspace{2.9cm}+ \frac{C_5\tilde{\kappa}(\kappa-\hat{\kappa})}{\tilde{\kappa}(\kappa-\hat{\kappa})-\ms\tonde{\tilde{\kappa}-1}}\tonde{\alpha^N\left[\kappa-\hat{\kappa}\right]^N-\frac{\kappa-\hat{\kappa}}{\ms}\alpha^N\tonde{\frac{\tilde{\kappa}-1}{\tilde{\kappa}}}^{N-1}\left[ \ms \right]^{N}} \\
        &\hspace{2.9cm}+ \frac{\tonde{n-2}\alpha^2(\kappa-\hat{\kappa})}{\tonde{\tilde{\kappa}-\alpha\tonde{\tilde{\kappa}-1}}\tonde{\alpha(\kappa-\hat{\kappa})-\ms}}\tonde{\alpha^N\left[\kappa-\hat{\kappa}\right]^N - \frac{\alpha(\kappa-\hat{\kappa})}{\ms}\left[\ms\right]^{N}}\vast].
    \end{align*}
    Finally,
    \begin{align*}
        {\sum_{i=2}^NG^i_3}&=C_1\sums \frac{c_\sigma\kappa   }{m_\sigma}\vast[C_6\vastt(\frac{\tonde{1-\alpha}\tilde{\kappa}\ms}{\ms-\alpha\kappa} \tonde{ \alpha\left[\ms \right]^N -\frac{\alpha^N\ms}{\kappa}[\kappa]^N } \\
        &\hspace{3.2cm}+ \frac{\alpha\ms\tilde{\kappa}}{\ms\tilde{\kappa}-\kappa\alpha^2\tonde{\tilde{\kappa}-1}}\tonde{\frac{\alpha^2\tonde{\tilde{\kappa}-1}}{\tilde{\kappa}} \left[\ms\right]^N -\tonde{\frac{\alpha^2\tonde{\tilde{\kappa}-1}}{\tilde{\kappa}}}^N\frac{\ms}{\kappa}[\kappa]^N }\vastt) \\
        &\hspace{2.5cm}+\frac{C_4\ms}{\ms-\kappa}\tonde{\alpha^{N+1}\left[\ms\right]^N - \frac{\ms}{\kappa}\alpha^{N+1}[\kappa]^N} \\
        &\hspace{2.5cm}+ C_5\frac{\alpha\tilde{\kappa}\ms}{\tilde{\kappa}\ms-\kappa\alpha^2\tonde{\tilde{\kappa}-1}}\tonde{\left[\ms\right]^N - \tonde{\frac{\alpha^2\tonde{\tilde{\kappa}-1}}{\tilde{\kappa}}}^{N-1}\frac{\ms}{\kappa}[\kappa]^N}\\
        &\hspace{2.5cm}+\frac{\tonde{n-2}\alpha}{\tilde{\kappa}-\alpha\tonde{\tilde{\kappa}-1}}\frac{\ms}{\ms-\alpha\kappa}\tonde{\alpha\left[\ms\right]^N - \frac{\ms}{\kappa}\alpha^N[\kappa]^N} \vast].
    \end{align*}
    Again, Lemma~\ref{lemma_asympt_rate_various_elements} and \cite[Lemma~A.3]{SchiedStrehleZhang.17} yield all necessary limits, and therefore
    \begin{align*}
        {\bm\omega}^\top (\left.\hat{\kappa}\tilde{\Gamma}-\tilde{\Gamma}^\top \right.){\bm\nu} &= ({\bm\omega}^\top (\left.\hat{\kappa}\tilde{\Gamma}-\tilde{\Gamma}^\top \right.))_1{\nu_1} + \sum_{k=1}^3\sum_{i=2}^NG^i_k + ({\bm\omega}^\top (\left.\hat{\kappa}\tilde{\Gamma}-\tilde{\Gamma}^\top \right.))_{N+1}{\nu_{N+1}}\\
        &\to \frac{\tonde{n-1}^2\tonde{n+e^{-\rho\frac{n+1}{n-1}T}}}{n\kappa(n+1)} +\vast(\frac{(n+1)n\kappa\tonde{1+(n-2)\rho T - e^{-\rho\frac{n+1}{n-1}T}} }{(n+1)^2n\kappa}\\
        &+\frac{(n+1)(\kappa-\hat{\kappa})\tonde{n(n-1) +(n-1)e^{-\rho\frac{n+1}{n-1}T} }+2(n-2)n\kappa\tonde{1-e^{-\rho\frac{n+1}{n-1}T}}}{(n+1)^2n\kappa}\vast) + 0\\
        &= \frac{-(n-1)(2n-1)e^{-\rho\frac{n+1}{n-1}T} + n(n+4)(n-1) + n(n+1)(n-2)\rho T}{n(n+1)^2}.
    \end{align*}
\end{proof}

Before proving Theorem~\ref{costs asymptotics thm}, we recall that $v_k$ corresponds to the $k$-th element of the vector $\bm v=(v_1,\dots,v_{N+1})\in\R^{N+1}$, whereas $\xi_k$ corresponds to the $(k+1)$-th element of the vector $\bm \xi=(\xi_0,\dots,\xi_N)\in\R^{N+1}$. 

\begin{proof}[Proof of Theorem~\ref{costs asymptotics thm}] By \cite[(23)]{SchiedStrehleZhang.17} we have
\begin{equation}\label{1top omega convergence eq}
\mathbf{1}^\top {\bm\omega}
= \sum_{i=1}^{N+1} \omega_i
 \lra  \rho T + 1
\quad\text{as } N\uparrow\infty.
\end{equation}
Moreover, the limit of $\mathbf{1}^\top {\bm\nu}=\sum_{i=1}^{N+1}\nu_i$ is given by \eqref{sumnuitotal} when $\kappa=n-1$, and by \eqref{sum nui kappa>1/2, kappa neq 1 eq} when $\kappa\neq n-1$ with $\kappa>\frac{n-1}{2}$. The limits of ${\bm\nu}^\top \tilde{\Gamma} {\bm\nu}$, ${\bm\omega}^\top  (\ka{\tilde{\Gamma}-\tilde{\Gamma}^\top}){\bm\nu}$, and ${\bm\omega}^\top \tilde{\Gamma} {\bm\omega}$ are collected in Lemma~\ref{cost functional ausiliral lemma}. Substituting these into \eqref{am} yields the claim.
Finally, we only need to prove \eqref{inst_cost_to_block}, then \eqref{impact_cost_to_cont} will follow automatically; recall
\[
  \xi_{i,k}  =  \bar x v_k  +  (x_i-\bar x) w_k,
\]
where $\bm w$ and $\bm v$ are defined in \eqref{v and w}. Without loss of generality, and to simplify explicit computations, we can fix $c=1/2$; the same argument remains valid replacing $1/2$ with any $c\in(0,1)$.

\noindent\emph{Step 1: Back window \([\lceil{N}/{2}\rceil,\dots,N]\), recovery of \(\mathscr{B}_T\).}

Near \(t=T\) the \(\bm w\)-contribution dominates, hence (recall the indexing convention for $\bm\xi$ is $\{0,\dots,N\}$ and for $\bm v$ and $\bm w$ is $\{1,\dots,N+1\}$)
\begin{align*}
    \theta \sum_{k=\lceil N/2\rceil}^{N} \bigl(\xi_{i,k}\bigr)^2
    &= \theta \bar{x}^2\sum_{k=\lceil N/2\rceil+1}^{N+1}v_k^2
        +  2\theta \bar{x}\bigl(x_i-\bar{x}\bigr)\sum_{k=\lceil N/2\rceil+1}^{N+1}v_k w_k
        +  \theta\bigl(x_i-\bar{x}\bigr)^2 \sum_{k=\lceil N/2\rceil+1}^{N+1} w_k^2\\
    &= \theta\bigl(x_i-\bar{x}\bigr)^2 \sum_{k=\lceil N/2\rceil+1}^{N+1} w_k^2  +  o(1)
     \longrightarrow  \mathscr{B}_T \qquad (N\to\infty).
\end{align*}
Using the explicit formula in \eqref{omi formula},
\[
  \sum_{k=\lceil N/2\rceil+1}^{N+1}\omega_k^2 \ \longrightarrow\ \frac{1}{2\tilde\kappa-1} = \frac{1}{4\theta},
\] \hide{
To see this notice that: (independently of the choice $c=1/2$)
\begin{align*}
    \sum_{k=\lceil\frac{N}{2}\rceil+1}^{N+1} \omega_k^2 &= \frac{1}{\tilde{\kappa}^2\tonde{\tilde{\kappa}-\alpha\tonde{\tilde{\kappa}-1}}^2}\sum_{k=\lceil\frac{N}{2}\rceil+1}^{N+1} \tonde{1-\alpha}^2\tilde{\kappa}^2 + \alpha^2\left(\frac{\alpha^2\left(\tilde{\kappa}-1\right)^2}{\tilde{\kappa}^2}\right)^{N+1-k} + 2\alpha\tonde{1-\alpha}\tilde{\kappa}\left(\frac{\alpha\left(\tilde{\kappa}-1\right)}{\tilde{\kappa}}\right)^{N+1-k}
\end{align*}
it is not difficult to see that: 
\begin{align*}
    \frac{\sum_{k=\lceil\frac{N}{2}\rceil+1}^{N+1} \tonde{1-\alpha}^2\tilde{\kappa}^2}{\tilde{\kappa}^2\tonde{\tilde{\kappa}-\alpha\tonde{\tilde{\kappa}-1}}^2}\to 0; 
\end{align*}
\begin{align*}
    \frac{\sum_{k=\lceil\frac{N}{2}\rceil+1}^{N+1} \alpha^2\left(\frac{\alpha^2\left(\tilde{\kappa}-1\right)^2}{\tilde{\kappa}^2}\right)^{N+1-k}}{\tilde{\kappa}^2\tonde{\tilde{\kappa}-\alpha\tonde{\tilde{\kappa}-1}}^2}\to \frac{1}{4\theta};
\end{align*}
\begin{align*}
    \frac{\sum_{k=\lceil\frac{N}{2}\rceil+1}^{N+1} 2\alpha\tonde{1-\alpha}\tilde{\kappa}\left(\frac{\alpha\left(\tilde{\kappa}-1\right)}{\tilde{\kappa}}\right)^{N+1-k}}{\tilde{\kappa}^2\tonde{\tilde{\kappa}-\alpha\tonde{\tilde{\kappa}-1}}^2}\to 0;
\end{align*}
}
and, combining this with \eqref{1top omega convergence eq},
\[
  \sum_{k=\lceil N/2\rceil+1}^{N+1} w_k^2
  \ \longrightarrow\ \frac{1}{4\theta (\rho T+1)^2}.
\]
To see that the $\bm v$-part and the cross term vanish as $N\to\infty$, first consider $\kappa=n-1$: by \eqref{expl_form_nu_i_kappa_n-1},
\[
  \nu_i^2  \le  \frac{\rho^2 T^2}{(n-1)^2}\frac{1}{N^2}  +  o  \left(\frac{1}{N^2}\right),
  \qquad i\in\Bigl\{\bigl\lceil\tfrac{N}{2}\bigr\rceil+1,\dots,N+1\Bigr\}.
\] \hide{it is easy to see that this does not depend on the choice $c=1/2$}
For $\kappa\neq n-1$, \eqref{expl_form_nu_i} and \eqref{nu_N+1_expl_form} yield, for $i\in\{\lceil N/2\rceil+1,\dots,N\}$,
\[
  \nu_i^2  \le  \rho^2T^2\Bigl(\frac{2}{(n+1)(n-1)}e^{\rho T\frac{n+1}{n-1}} + \frac{1}{n+1}\Bigr)^2\frac{1}{N^2}  +  o  \left(\frac{1}{N^2}\right),
  \qquad
  \nu_{N+1}^2  \le  \frac{\rho^2T^2}{(n-1)^2}\frac{1}{N^2}  +  o  \left(\frac{1}{N^2}\right).
\]
Together with the limit of $\mathbf{1}^\top\bm\nu$ (from \eqref{sumnuitotal} or \eqref{sum nui kappa>1/2, kappa neq 1 eq}), this implies
\[
  \sum_{k=\lceil N/2\rceil+1}^{N+1} v_k^2  =  \mathcal{O}  \left(\frac{1}{N}\right),
  \qquad \text{for any }\kappa>\frac{n-1}{2}.
\]
By Cauchy--Schwarz,
\[
  \Bigl|\sum_{k=\lceil N/2\rceil+1}^{N+1} v_k w_k\Bigr|
   \le  \Bigl(\sum_{k=\lceil N/2\rceil+1}^{N+1} v_k^2\Bigr)^{1/2}
           \Bigl(\sum_{k=\lceil N/2\rceil+1}^{N+1} w_k^2\Bigr)^{1/2}
   \xrightarrow[N\to\infty]{} 0.
\]
Hence the limit over the back half equals $\mathscr{B}_T$. 

\noindent\emph{Step 2: Front window \([0,\dots,\lceil{N}/{2}\rceil-1]\), recovery of \(\mathscr{B}_0\).}

Near \(t=0\) the \(\bm v\)-contribution dominates, so
\begin{align*}
    \theta \sum_{k=0}^{\lceil N/2\rceil-1} \bigl(\xi_{i,k}\bigr)^2
    &= \theta \bar{x}^2\sum_{k=1}^{\lceil N/2\rceil}v_k^2
        +  2\theta \bar{x}\bigl(x_i-\bar{x}\bigr)\sum_{k=1}^{\lceil N/2\rceil}v_k w_k
        +  \theta\bigl(x_i-\bar{x}\bigr)^2 \sum_{k=1}^{\lceil N/2\rceil} w_k^2\\
    &= \theta \bar{x}^2\sum_{k=1}^{\lceil N/2\rceil}v_k^2  +  o(1).
\end{align*}
Using \eqref{nu_one_expl_form}-\eqref{expl_form_nu_i} for $\kappa\neq n-1$ and \eqref{expl_form_nu_i_kappa_n-1} for $\kappa=n-1$ (in the latter case $\theta=\tfrac{n-1}{4}$ and only the first trade contributes, meaning $\sum_{k=2}^{\lceil N/2\rceil}\nu_k^2\to0$),
\[
  \sum_{k=1}^{\lceil N/2\rceil}\nu_k^2
  \ \longrightarrow\
  \frac{(n-1) e^{-2\rho T\frac{n+1}{n-1}}\Bigl(n e^{\rho T\frac{n+1}{n-1}}+1\Bigr)^2}{(n+1)^2 n^2 4\theta}.
\]
Therefore, combining with the limit in \eqref{sum nui kappa>1/2, kappa neq 1 eq} (that does not depend on $\theta$), we get, for any $\theta>0$,
\[
  \sum_{k=1}^{\lceil N/2\rceil} v_k^2
  \ \longrightarrow\
  \frac{(n-1)(n+1)^2\Bigl(1+n e^{\rho \frac{n+1}{n-1} T}\Bigr)^2}
       {4\theta\Bigl(n\bigl((\rho T+1)(n+1)+2\bigr)e^{\rho \frac{n+1}{n-1} T}-(n-1)\Bigr)^2}.
\]
To show that the $\bm w$-part and the cross term vanish, note from \eqref{omi formula} that, for $i\in\{1,\dots,\lceil N/2\rceil\}$,
\[
  \omega_i^2  \le  4\rho^2T^2 \frac{1}{N^2}  +  o  \left(\frac{1}{N^2}\right).
\]
By \eqref{1top omega convergence eq}, we conclude as in Step~1. 
\end{proof}

\subsection{Proof of Theorem~\ref{strat_osc_thm}~\ref{V_oscillations}}

\begin{proof}[Proof of Theorem~\ref{strat_osc_thm}~\ref{V_oscillations}]
For $\kappa=\frac{n-1}{2}$, the limits in \eqref{ae} and \eqref{af} follow from Lemma~\ref{lemma_asympt_rate_various_elements}. 
We evaluate \eqref{ae} with $m=n_t$ term-by-term as $N\uparrow\infty$.

\begin{enumerate}[label=\arabic*.]
    \item 
    \[
    \sum_{\sigma\in\{+,-\}}
    \frac{d_\sigma\left(m_\sigma-\alpha^2\kappa\right)}{m_\sigma-\alpha\kappa} [m_\sigma]^N
     \longrightarrow 
    \begin{cases}
    \displaystyle \frac{2n}{\left(e^{-2\frac{n+1}{n-1}\rho T}+n\right)(n+1)}, & N=2k,\\[6pt]
    \displaystyle \frac{2n}{\left(-e^{-2\frac{n+1}{n-1}\rho T}+n\right)(n+1)}, & N=2k+1.
    \end{cases}
    \]
    \item 
    \[
    \left(1-\alpha\right)\left(n_t-1\right)  
    \sum_{\sigma\in\{+,-\}} c_\sigma d_\sigma 
    \left(\frac{\alpha(\kappa-\hat{\kappa})}{m_\sigma-\alpha(\kappa-\hat{\kappa})}
    +\frac{m_\sigma}{m_\sigma-\alpha\kappa}\right)[m_\sigma]^N
     \longrightarrow 
    \begin{cases}
    \displaystyle \frac{\rho t}{n+1}, & N=2k,\\[6pt]
    \displaystyle \frac{\rho t}{n+1}, & N=2k+1.
    \end{cases}
    \]
    \item 
    \[
    \begin{aligned}
    & C_1  \left(1+\sum_{\sigma\in\{+,-\}}
    \frac{c_\sigma m_\sigma\Big(\big(\frac{m_\sigma}{\alpha\kappa}\big)^{n_t-1}-1\Big)}
    {m_\sigma-\alpha\kappa}\right)\alpha^N [\kappa]^N\\
    &\qquad\longrightarrow
    \begin{cases}
    \displaystyle
    \frac{e^{-\rho\frac{n+1}{n-1}T}}
    {\big(e^{-2\frac{n+1}{n-1}\rho T}+n\big)(n+1)^2}
    \left(2n e^{\rho\frac{n+1}{n-1}t}
    -(n+1)\big(\pm e^{-\rho\frac{n+1}{n-1}t}\big)-(n-1)\right), & N=2k,\\[8pt]
    \displaystyle
    \frac{e^{-\rho\frac{n+1}{n-1}T}}
    {\big(-e^{-2\frac{n+1}{n-1}\rho T}+n\big)(n+1)^2}
    \left(2n e^{\rho\frac{n+1}{n-1}t}
    -(n+1)\big(\pm e^{-\rho\frac{n+1}{n-1}t}\big)-(n-1)\right), & N=2k+1.
    \end{cases}
    \end{aligned}
    \]
    \item \hide{Let's split this into the single addends
            \begin{enumerate}[label=\roman*.]
                \item 
                    \begin{align*}
                        nC_1\frac{d_+ m_+\left(\frac{\alpha(\kappa-\hat{\kappa})}{m_+}-\left(\frac{\alpha(\kappa-\hat{\kappa})}{m_+}\right)^{n_t}\right)}{m_+-\alpha(\kappa-\hat{\kappa})}  \left[m_+\right]^N\lra
                         \begin{cases} 
                            &\frac{-n\left(1\pm e^{-\rho\frac{n+1}{n-1}t}\right)}{\left(e^{-2\frac{n+1}{n-1}\rho T} +n\right)(n+1)}\qquad\text{if }N=2k\\
                            &\frac{-n\left(1\pm e^{-\rho\frac{n+1}{n-1}t}\right)}{\left(-e^{-2\frac{n+1}{n-1}\rho T} +n\right)(n+1)}\qquad\text{if }N=2k+1
                        \end{cases}
                    \end{align*}
                \item 
                \begin{align*}
                        nC_1\frac{d_- m_-\left(\frac{\alpha(\kappa-\hat{\kappa})}{m_-}-\left(\frac{\alpha(\kappa-\hat{\kappa})}{m_-}\right)^{n_t}\right)}{m_--\alpha(\kappa-\hat{\kappa})}  \left[m_-\right]^N\lra
                         \begin{cases} 
                            &\frac{2n e^{-2\frac{n+1}{n-1}\rho T}\left(1- e^{\rho\frac{n+1}{n-1}t}\right)}{\left(e^{-2\frac{n+1}{n-1}\rho T} +n\right)(n+1)^2}\qquad\text{if }N=2k\\
                            &\frac{-2n e^{-2\frac{n+1}{n-1}\rho T}\left(1- e^{\rho\frac{n+1}{n-1}t}\right)}{\left(-e^{-2\frac{n+1}{n-1}\rho T} +n\right)(n+1)^2}\qquad\text{if }N=2k+1
                        \end{cases}
                    \end{align*}
            \end{enumerate}
            now letting } Define
    \[
        D_+:=\left(ne^{2\frac{n+1}{n-1}\rho T}+1 \right)(n+1)^2,
        \qquad 
        D_-:=\left(ne^{2\frac{n+1}{n-1}\rho T}-1 \right)(n+1)^2.
    \]
    Then
    \[
    \begin{aligned}
    &nC_1 \sum_{\sigma\in\{+,-\}} 
    \frac{d_\sigma m_\sigma\left(\frac{\alpha(\kappa-\hat{\kappa})}{m_\sigma}
    -\left(\frac{\alpha(\kappa-\hat{\kappa})}{m_\sigma}\right)^{n_t}\right)}
    {m_\sigma-\alpha(\kappa-\hat{\kappa})} [m_\sigma]^N\\
    &\qquad\longrightarrow
    \begin{cases}
    \displaystyle \frac{2n-2ne^{\rho\frac{n+1}{n-1}t}-n(n+1)e^{2\frac{n+1}{n-1}\rho T}\left(1\pm e^{-\rho\frac{n+1}{n-1}t}\right)}{D_+},
    & N=2k,\\[10pt]
    \displaystyle \frac{-2n+2ne^{\rho\frac{n+1}{n-1}t}-n(n+1)e^{2\frac{n+1}{n-1}\rho T}\left(1\pm e^{-\rho\frac{n+1}{n-1}t}\right)}{D_-},
    & N=2k+1.
    \end{cases}
    \end{aligned}
    \]
\end{enumerate}
\hide{Summing now the limits of $\mathrm{4}.$ and $\mathrm{1}.$, we get
\begin{align*}
    \begin{cases} 
        &\frac{2n-2ne^{\rho\frac{n+1}{n-1}t} - n(n+1)e^{2\frac{n+1}{n-1}\rho T}\left( -1\pm e^{-\rho\frac{n+1}{n-1}t} \right)}{D_+}\qquad\text{if }N=2k\\
        &\frac{-2n+2ne^{\rho\frac{n+1}{n-1}t} - n(n+1)e^{2\frac{n+1}{n-1}\rho T}\left( -1\pm e^{-\rho\frac{n+1}{n-1}t} \right)}{D_-}\qquad\text{if }N=2k+1\\
    \end{cases}
\end{align*}}
Summing the four contributions yields the limit
\hide{\begin{small}
    \begin{align*}
    \begin{cases} 
        \frac{  2n-2ne^{\rho\frac{n+1}{n-1}t} - n(n+1)e^{2\frac{n+1}{n-1}\rho T}\left( -1\pm e^{-\rho\frac{n+1}{n-1}t} \right)  + (n+1)\left( ne^{2\frac{n+1}{n-1}\rho T}+1 \right)\rho t +e^{\rho\frac{n+1}{n-1}T} \left( 2n e^{\rho\frac{n+1}{n-1}t} -(n+1)\left( \pm e^{-\rho\frac{n+1}{n-1}t}\right)-(n-1) \right)}{D_+}\\\text{ if } N=2k\\
        \frac{  -2n+2ne^{\rho\frac{n+1}{n-1}t} - n(n+1)e^{2\frac{n+1}{n-1}\rho T}\left( -1\pm e^{-\rho\frac{n+1}{n-1}t} \right)  + (n+1)\left( ne^{2\frac{n+1}{n-1}\rho T}-1 \right)\rho t + e^{\rho\frac{n+1}{n-1}T}\left( 2n e^{\rho\frac{n+1}{n-1}t} -(n+1)\left( \pm e^{-\rho\frac{n+1}{n-1}t}\right)-(n-1) \right)}{D_-}\\\text{ if } N=2k+1\\
    \end{cases}
\end{align*}
\end{small}
that rearranging the terms becomes}
\[
    \begin{aligned}
    &\sum_{i=1}^{n_t}\nu_i\\
    &\qquad\longrightarrow
    \begin{cases}
    \displaystyle
    \frac{2n+(n+1)\rho t-2n e^{\rho\frac{n+1}{n-1}t}
    +e^{2\frac{n+1}{n-1}\rho T}\left(n(n+1)+n(n+1)\rho t\right)
    -e^{\frac{n+1}{n-1}\rho(2T-t)}\left(\pm n(n+1)\right)}{D_+}\\
    \displaystyle\qquad
    +\frac{2n e^{\frac{n+1}{n-1}\rho (T+t)}
    -e^{\frac{n+1}{n-1}\rho (T-t)}\left(\pm (n+1)\right)
    -(n-1)e^{\frac{n+1}{n-1}\rho T}}{D_+},
    \quad N=2k,\\[12pt]
    \displaystyle
    \frac{-2n-(n+1)\rho t+2n e^{\rho\frac{n+1}{n-1}t}
    +e^{2\frac{n+1}{n-1}\rho T}\left(n(n+1)+n(n+1)\rho t\right)
    -e^{\frac{n+1}{n-1}\rho(2T-t)}\left(\pm n(n+1)\right)}{D_-}\\
    \displaystyle\qquad
    +\frac{2n e^{\frac{n+1}{n-1}\rho (T+t)}
    -e^{\frac{n+1}{n-1}\rho (T-t)}\left(\pm (n+1)\right)
    -(n-1)e^{\frac{n+1}{n-1}\rho T}}{D_-},
    \quad N=2k+1.
    \end{cases}
    \end{aligned}
\]
Setting $t=T$ in the preceding display gives the limit
\[
    \begin{aligned}
    &\sum_{i=1}^{N}\nu_i\\
    &\qquad\longrightarrow
    \begin{cases}
    \displaystyle
    \frac{n e^{2\frac{n+1}{n-1}\rho T}\left((n+3)+(n+1)\rho T\right)
    +e^{\frac{n+1}{n-1}\rho T}(1-4n-n^2)+(n+1)\rho T+(n-1)}{D_+},
    & N=2k,\\[10pt]
    \displaystyle
    \frac{n e^{2\frac{n+1}{n-1}\rho T}\left((n+3)+(n+1)\rho T\right)
    +e^{\frac{n+1}{n-1}\rho T}(1+2n+n^2)-(n+1)\rho T-(n-1)}{D_-},
    & N=2k+1.
    \end{cases}
    \end{aligned}
\]
Turning to $\nu_{N+1}$, \eqref{af} yields
\[
    \nu_{N+1} \longrightarrow 
    \begin{cases}
    \displaystyle \frac{2(n+1)+2n(n+1)e^{\frac{n+1}{n-1}\rho T}}{D_+}, & N=2k,\\[8pt]
    \displaystyle \frac{-2(n+1)-2n(n+1)e^{\frac{n+1}{n-1}\rho T}}{D_-}, & N=2k+1.
    \end{cases}
\]
\hide{\begin{align*}
    \sums\frac{c_\sigma\left(m_\sigma-\alpha^2(\kappa-\hat{\kappa})\right)}{m_\sigma-\alpha(\kappa-\hat{\kappa})}  [m_\sigma]^N\lra
    \begin{cases}
        \frac{2(n+1)}{D_+}\qquad&\text{if}\qquad N=2k\\[10pt]
        \frac{-2(n+1)}{D_-}\qquad&\text{if}\qquad N=2k+1
    \end{cases}
\end{align*}
and 
\begin{align*}
    nC_1\alpha^N  \left[\kappa-\hat{\kappa}\right]^N\lra
    \begin{cases}
        \frac{2n(n+1)e^{\frac{n+1}{n-1}\rho T}}{D_+}\qquad&\text{if}\qquad N=2k\\[10pt]
        \frac{-2n(n+1)e^{\frac{n+1}{n-1}\rho T}}{D_-}\qquad&\text{if}\qquad N=2k+1.
    \end{cases}
\end{align*}}

Combining the preceding two displays yields the limits of $\mathbf{1}^\top{\bm\nu}$:
\begin{equation}\label{limits sum nui kappa=n-1/2 eq}
\begin{split}
\lim_{\substack{N\uparrow\infty\\ N \mathrm{even}}}\mathbf{1}^\top{\bm\nu}
&=\frac{n e^{2\frac{n+1}{n-1}\rho T}\left((n+1)\rho T+(n+3)\right)
+(n-1)^2 e^{\frac{n+1}{n-1}\rho T}+(n+1)\rho T+(3n+1)}{D_+},\\
\lim_{\substack{N\uparrow\infty\\ N \mathrm{odd}}}\mathbf{1}^\top{\bm\nu}
&=\frac{n e^{2\frac{n+1}{n-1}\rho T}\left((n+1)\rho T+(n+3)\right)
+(1-n^2)e^{\frac{n+1}{n-1}\rho T}-(n+1)\rho T-(3n+1)}{D_-}.
\end{split}
\end{equation}

Finally, substituting these limits into the definition of $V_t^{(N)}$ completes the proof of the oscillation statement in part~\ref{V_oscillations}.
\end{proof}

\subsection{Proof of Theorem~\ref{strat_osc_thm}~\ref{W_oscillations}}
\begin{proof}[Proof of Theorem~\ref{strat_osc_thm}~\ref{W_oscillations}]
By Remark~\ref{w_indep_n}, $W^{\tonde{N}}$ is independent of $n$. Hence the argument of \cite[Theorem~3.1(d)]{SchiedStrehleZhang.17}, established for $n=2$, applies analogously in our setting for any $t\in(0,T)$. For $t=0$ and $t=T$ a straightforward limit computation yields the result.
\end{proof}

\subsection{Proof of Theorem~\ref{cost_asympt_thm_theta_zero}}

\begin{lemma}\label{cost functional ausiliral lemma k=n-1/2}
Let $\kappa=\frac{n-1}{2}$. Then, as $N\uparrow\infty$,
\begin{align*}
\lim_{\substack{N\uparrow\infty\\ N \mathrm{even}}}
{\bm\nu}^\top \tilde{\Gamma} {\bm\nu}
&=\frac{
n e^{2\rho\frac{n+1}{n-1}T}\tonde{(n+1)\rho T+n+3}
+(n-1)^2 e^{\rho\frac{n+1}{n-1}T}
+(n+1)\rho T+3n+1
}{(n+1)D_+},\\[6pt]
\lim_{\substack{N\uparrow\infty\\ N \mathrm{even}}}
{\bm\omega}^\top  \tonde{\hat{\kappa}\tilde{\Gamma}-\tilde{\Gamma}^\top}{\bm\nu}
&=\frac{
n^2 e^{2\rho\frac{n+1}{n-1}T}
- n(n+1)e^{\rho\frac{n+3}{n-1}T}
+ (2n^2-3n-1)e^{\rho\frac{n+1}{n-1}T}
- (n+1)e^{-\rho T}
+ 3n-2
}{\mathscr{D}_+}\\
&\quad+\frac{\rho T\tonde{n-2}\tonde{n e^{2\rho\frac{n+1}{n-1}T}+1}}{\mathscr{D}_+}
 + \frac{2n(n-2)\tonde{e^{\rho\frac{n+1}{n-1}T}-1}^2}{D_+},\\[6pt]
\lim_{\substack{N\uparrow\infty\\ N \mathrm{even}}}
{\bm\omega}^\top \tilde{\Gamma} {\bm\omega}
&=e^{-\rho T}+\rho T+1.
\end{align*}
Moreover,
\begin{align*}
\lim_{\substack{N\uparrow\infty\\ N \mathrm{odd}}}
{\bm\nu}^\top \tilde{\Gamma} {\bm\nu}
&=\frac{
n e^{2\rho\frac{n+1}{n-1}T}\tonde{(n+1)\rho T+n+3}
- (n^2-1)e^{\rho\frac{n+1}{n-1}T}
- (n+1)\rho T-(3n+1)
}{(n+1)D_-},\\[6pt]
\lim_{\substack{N\uparrow\infty\\ N \mathrm{odd}}}
{\bm\omega}^\top  \tonde{\hat{\kappa}\tilde{\Gamma}-\tilde{\Gamma}^\top}{\bm\nu}
&=\frac{
n^2 e^{2\rho\frac{n+1}{n-1}T}
+ n(n+1)e^{\rho\frac{n+3}{n-1}T}
- (2n^2-3n+1)e^{\rho\frac{n+1}{n-1}T}
- (n+1)e^{-\rho T}
- 3n+2
}{\mathscr{D}_-}\\
&\quad+\frac{\rho T\tonde{n-2}\tonde{n e^{2\rho\frac{n+1}{n-1}T}-1}}{\mathscr{D}_-}
 + \frac{2n(n-2)\tonde{e^{2\rho\frac{n+1}{n-1}T}-1}}{D_-},\\[6pt]
\lim_{\substack{N\uparrow\infty\\ N \mathrm{odd}}}
{\bm\omega}^\top \tilde{\Gamma} {\bm\omega}
&=-e^{-\rho T}+\rho T+1.
\end{align*}
Here
\[
D_\pm := \tonde{n e^{2\frac{n+1}{n-1}\rho T}\pm 1}(n+1)^2,
\qquad
\mathscr{D}_\pm := \frac{D_\pm}{(n+1)}.
\]
\end{lemma}

\begin{proof}
    Since ${\bm\omega}$ is independent of $n$, the third limits coincide with the $2$--player case and are given by \cite[Lemma A.6]{SchiedStrehleZhang.17}. Hence it suffices to establish, for $\kappa=\frac{n-1}{2}$, the first two limits for $N$ even and odd. Moreover, as explained in the proof of Lemma~\ref{cost functional ausiliral lemma}, the representations of
    \[
    {\bm\nu}^\top \tilde{\Gamma} {\bm\nu}
    \quad\text{and}\quad
    {\bm\omega}^\top  \tonde{\hat{\kappa}\tilde{\Gamma}-\tilde{\Gamma}^\top}{\bm\nu}
    \]
    obtained there for $\kappa\neq n-1$ also hold for $\kappa=\frac{n-1}{2}$. Plugging in the limits from Lemma~\ref{lemma_asympt_rate_various_elements} yields the claim. For completeness, we record the decomposition and limiting contributions used in the argument.
    
    \noindent\emph{Quadratic form in ${\bm\nu}$}:
    \begin{align*}
    {\bm\nu}^\top \tilde{\Gamma} {\bm\nu}
    &= \nu_1\bigl(\tilde{\Gamma}{\bm\nu}\bigr)_1
      + \nu_2\bigl(\tilde{\Gamma}{\bm\nu}\bigr)_2
      + \sum_{k=1}^4 \sum_{i=3}^{N} D_k^i
      + \nu_{N+1}\bigl(\tilde{\Gamma}{\bm\nu}\bigr)_{N+1}.
    \end{align*}
    The boundary terms cancel asymptotically,
    \[
    \nu_1\bigl(\tilde{\Gamma}{\bm\nu}\bigr)_1
     + 
    \nu_2\bigl(\tilde{\Gamma}{\bm\nu}\bigr)_2  \longrightarrow  0.
    \]
    For the interior contributions,
    \begin{align*}
    \sum_{k=1}^3 \sum_{i=3}^{N} D_k^i  \longrightarrow 
    \begin{cases}
    \displaystyle
    \frac{1}{(n+1)D_+}
    \Bigl[
    \tonde{n e^{2\rho\frac{n+1}{n-1}T}+1}(n+1)\rho T
    \\[-2pt]\displaystyle\qquad\qquad
    \quad+\tonde{(n^2+4n-1)e^{2\rho\frac{n+1}{n-1}T}
    -(n^2+6n-3)e^{\rho\frac{n+1}{n-1}T}+2(n-1)}
    \Bigr], & N=2k,\\[10pt]
    \displaystyle
    \frac{1}{(n+1)D_-}
    \Bigl[
    \tonde{n e^{2\rho\frac{n+1}{n-1}T}-1}(n+1)\rho T
    \\[-2pt]\displaystyle\qquad\qquad
    \quad+\tonde{(n^2+4n-1)e^{2\rho\frac{n+1}{n-1}T}
    +(n+1)^2 e^{\rho\frac{n+1}{n-1}T}-2(n-1)}
    \Bigr], & N=2k+1,
    \end{cases}
    \end{align*}
    and
    \begin{align*}
    \sum_{i=3}^{N}   D_4^i  +  \nu_{N+1}\bigl(\tilde{\Gamma}{\bm\nu}\bigr)_{N+1}
     \longrightarrow 
    \begin{cases}
    \displaystyle
    \frac{1}{(n+1)D_+}
    \Bigl[
    2(n^2-1+2n)e^{\rho\frac{n+1}{n-1}T}
    \\[-2pt]\displaystyle\qquad\qquad
    \quad\ -(n-1)e^{2\rho\frac{n+1}{n-1}T}+n+3
    \Bigr], & N=2k,\\[10pt]
    \displaystyle
    \frac{1}{(n+1)D_-}
    \Bigl[
    -2n(n+1)e^{\rho\frac{n+1}{n-1}T}
    \\[-2pt]\displaystyle\qquad\qquad
    \quad\ -(n-1)e^{2\rho\frac{n+1}{n-1}T}-(n+3)
    \Bigr], & N=2k+1.
    \end{cases}
    \end{align*}
    Adding these two displays gives the limit for ${\bm\nu}^\top \tilde{\Gamma} {\bm\nu}$.

    \noindent\emph{Mixed form ${\bm\omega}^\top  \tonde{\hat{\kappa}\tilde{\Gamma}-\tilde{\Gamma}^\top}{\bm\nu}$}:
    \begin{align*}
    {\bm\omega}^\top  \tonde{\hat{\kappa}\tilde{\Gamma}-\tilde{\Gamma}^\top}{\bm\nu}
    &=
    \bigl({\bm\omega}^\top  \tonde{\hat{\kappa}\tilde{\Gamma}-\tilde{\Gamma}^\top}\bigr)_1 \nu_1
     + \sum_{k=1}^{3}\sum_{i=2}^{N} G_k^i
     + \bigl({\bm\omega}^\top  \tonde{\hat{\kappa}\tilde{\Gamma}-\tilde{\Gamma}^\top}\bigr)_{N+1}\nu_{N+1}.
    \end{align*}
    The boundary terms satisfy
    \begin{align*}
    \bigl({\bm\omega}^\top  \tonde{\hat{\kappa}\tilde{\Gamma}-\tilde{\Gamma}^\top}\bigr)_1 \nu_1
    &\longrightarrow
    \begin{cases}
    \displaystyle
    \frac{2\tonde{-e^{-\rho T}+n-1}\tonde{n e^{2\rho\frac{n+1}{n-1}T}
    + e^{\rho\frac{n+1}{n-1}T}}}{\mathscr{D}_+}, & N=2k,\\[8pt]
    \displaystyle
    \frac{2\tonde{ e^{-\rho T}+n-1}\tonde{n e^{2\rho\frac{n+1}{n-1}T}
    + e^{\rho\frac{n+1}{n-1}T}}}{\mathscr{D}_-}, & N=2k+1,
    \end{cases}\\[6pt]
    \bigl({\bm\omega}^\top  \tonde{\hat{\kappa}\tilde{\Gamma}-\tilde{\Gamma}^\top}\bigr)_{N+1}\nu_{N+1}
    &\longrightarrow
    \begin{cases}
    \displaystyle
    \frac{2\tonde{e^{-\rho T}-e^{-2\rho T}+n-2}\tonde{1+n e^{\rho\frac{n+1}{n-1}T}}}{\mathscr{D}_+}, & N=2k,\\[8pt]
    \displaystyle
    \frac{-2\tonde{e^{-\rho T}+e^{-2\rho T}+n-2}\tonde{1+n e^{\rho\frac{n+1}{n-1}T}}}{\mathscr{D}_-}, & N=2k+1.
    \end{cases}
    \end{align*}
    For the sums, we have
    \begin{align*}
    \sum_{i=2}^{N} G_1^i
    &\longrightarrow
    \begin{cases}
    \displaystyle \frac{\tonde{1-e^{-\rho T}}^{2}+(n-2)\rho T}{n+1}, & N=2k,\\[6pt]
    \displaystyle \frac{1-e^{-2\rho T}+(n-2)\rho T}{n+1}, & N=2k+1,
    \end{cases}\\[6pt]
    \sum_{i=2}^{N} G_2^i
    &\longrightarrow
    \begin{cases}
    \displaystyle
    \frac{n\tonde{-2e^{-\rho T}+2e^{-2\rho T}+1}e^{\rho\frac{n+1}{n-1}T}
    + n\tonde{1-n+(2-n)e^{-\rho T}}e^{2\rho\frac{n+1}{n-1}T}
    - n\tonde{e^{-\rho T}-1}}{\mathscr{D}_+}
    \\[-2pt]\displaystyle\qquad
    -  \frac{2n(n-2)}{D_+} \tonde{e^{\rho\frac{n+1}{n-1}T}-1},
    & \hspace{-3cm}N=2k,\\[10pt]
    \displaystyle
    \frac{n\tonde{2e^{-\rho T}+2e^{-2\rho T}-1}e^{\rho\frac{n+1}{n-1}T}
    + n\tonde{1-n+(n-2)e^{-\rho T}}e^{2\rho\frac{n+1}{n-1}T}
    - n\tonde{1+e^{-\rho T}}}{\mathscr{D}_-}
    \\[-2pt]\displaystyle\qquad
    +  \frac{2n(n-2)}{D_-} \tonde{e^{\rho\frac{n+1}{n-1}T}-1},
    & \hspace{-3cm} N=2k+1,
    \end{cases}\\[6pt]
    \sum_{i=2}^{N} G_3^i
    &\longrightarrow
    \begin{cases}
    \displaystyle
    \frac{\tonde{e^{-\rho T}-e^{-2\rho T}} n e^{2\rho\frac{n+1}{n-1}T}
    -\tonde{e^{-\rho T}-e^{-2\rho T}-1}
    +2e^{-\rho T}e^{\rho\frac{n+1}{n-1}T}
    -(2n-1)e^{\rho\frac{n+1}{n-1}T}}{\mathscr{D}_+}
    \\[-2pt]\displaystyle\qquad
    +  \frac{2n(n-2)}{D_+} e^{\rho\frac{n+1}{n-1}T}\tonde{e^{\rho\frac{n+1}{n-1}T}-1},
    &\hspace{-2cm} N=2k,\\[10pt]
    \displaystyle
    \frac{\tonde{e^{-\rho T}+e^{-2\rho T}} n e^{2\rho\frac{n+1}{n-1}T}
    +\tonde{e^{-\rho T}+e^{-2\rho T}-1}
    -(2n-1)e^{\rho\frac{n+1}{n-1}T}
    -2e^{-\rho T}e^{\rho\frac{n+1}{n-1}T}}{\mathscr{D}_-}
    \\[-2pt]\displaystyle\qquad
    +  \frac{2n(n-2)}{D_-} e^{\rho\frac{n+1}{n-1}T}\tonde{e^{\rho\frac{n+1}{n-1}T}-1},
    & \hspace{-2cm}N=2k+1.
    \end{cases}
    \end{align*}
    Summing the boundary contributions with $\sum_{i=2}^{N}G_k^i$ for $k=1,2,3$ yields the claimed limits for 
    \[
  {\bm\omega}^\top  \tonde{\hat{\kappa}\tilde{\Gamma}-\tilde{\Gamma}^\top} \bm\nu. \qedhere
\]
\end{proof}

\begin{proof}[Proof of Theorem~\ref{cost_asympt_thm_theta_zero}]
We proceed as in the proof of Theorem~\ref{costs asymptotics thm}, now using the limits
\begin{equation}\label{sum omi limit formula kappa=1/2}
\lim_{\substack{N\uparrow\infty\\ N \mathrm{even}}}\bm\omega^\top\bm1
= e^{-\rho T}+\rho T+1,
\qquad
\lim_{\substack{N\uparrow\infty\\ N \mathrm{odd}}}\bm\omega^\top\bm1
= -e^{-\rho T}+\rho T+1.
\end{equation}
These limits are taken from \cite[eq.~(25), Proof of Theorem~3.1(d)]{SchiedStrehleZhang.17}, derived for the $2$-player case; since $\bm\omega$ is independent of $n$ (Remark~\ref{w_indep_n}), the same limits apply here. In addition, we invoke \eqref{limits sum nui kappa=n-1/2 eq} together with Lemma~\ref{cost functional ausiliral lemma k=n-1/2}. Substituting these limits into the cost representation \eqref{am} yields the claim.
\end{proof}

\section{Time-Varying Instantaneous Costs}\label{half-grid-block-costs}

In this appendix we present a numerical analysis of how the equilibrium strategies and their asymptotics change when we charge instantaneous costs only on the first or second half of the time interval. This construction is motivated by the continuous-time game, where we can specify the \say{correct} block cost at $0$ but the \say{wrong} one at $T$ (or vice versa), and then an equilibrium exists only in the case of zero-net supply (or symmetric initial inventories, respectively); see Remark~\ref{non_existence_for_diff_theta_values}.
In the discrete-time model, a unique equilibrium still exists in these half-grid instantaneous-cost configurations. However, the qualitative behavior of the time-$t$ inventories changes substantially: when there is no instantaneous cost on one half of the grid, exactly one of the two processes $V^{(N)}$ and $W^{(N)}$ develops oscillations on that half of the interval, and the cluster points of the oscillating inventory are no longer the four cluster points from Theorem~\ref{strat_osc_thm}. In both of the configurations described below, the inventories $X^{(N),i}$ converge to the corresponding continuous-time equilibrium in precisely the cases singled out in
Remark~\ref{non_existence_for_diff_theta_values}.

\subsection*{Set-up}

We modify the matrix $\Gamma^\theta$ by turning the instantaneous-cost term on
or off separately on the first and second halves of the grid. Define
\begin{align*}
    H^\theta := \Gamma^0 + 2\theta\widetilde{I},
    \qquad
    J^\theta := \Gamma^0 + 2\theta\overline{I},
\end{align*}
where
\begin{align*}
    \overline{I} := I - \widetilde{I}, \qquad
    \widetilde{I}_{ij}
    :=
    \begin{cases}
        0, & i \neq j,\\[0.2em]
        0, & i = j,\ i\in\{1,\dots,\lceil (N+1)/2 \rceil\},\\[0.2em]
        1, & i = j,\ i\in\{\lceil (N+1)/2 \rceil+1,\dots, N+1\}.
    \end{cases}
\end{align*}
Thus $H^\theta$ corresponds to charging instantaneous costs only on the second
half of the time grid, while $J^\theta$ corresponds to charging instantaneous
costs only on the first half.

It can be shown that, if we replace $\Gamma^\theta$ by $H^\theta$ or $J^\theta$,
the proof of Theorem~\ref{exuniqnasheqdiscrete} carries over. Hence
the equilibrium strategies are still of the form~\eqref{NashEqStratDisc} with
\begin{equation}\label{part_case_sol}
    \bm{v}
    :=
    \frac{(H^{\theta} + (n-1)\widetilde{\Gamma})^{-1}\bm{1}}
         {\bm{1}^\top(H^{\theta} + (n-1)\widetilde{\Gamma})^{-1}\bm{1}},
    \qquad
    \bm{w}
    :=
    \frac{(H^{\theta} - \widetilde{\Gamma})^{-1}\bm{1}}
         {\bm{1}^\top(H^{\theta} - \widetilde{\Gamma})^{-1}\bm{1}},
\end{equation}
and analogously with $H^\theta$ replaced by $J^\theta$.
We then define the time-$t$ inventories $V^{(N)}$ and $W^{(N)}$ from $\bm v$ and $\bm w$ in each case, as in  \eqref{fund_strat_rescaled}. 

\subsection*{Second-half instantaneous cost}

We first charge instantaneous costs only on
the second half of the grid, that is, we use the objective with $H^\theta$
replacing $\Gamma^\theta$. Numerically we observe that
\begin{align*}
    \abs{W^{(N)}_t-\mathbbm{f}(t)} \longrightarrow 0.
\end{align*}
By contrast, $V^{(N)}_t$ does not converge to $\mathbbm{g}(t)$ on the whole
interval $[0,T]$, but it does converge to $\mathbbm{g}(t)$ on $[T/2,T]$. On
$[0,T/2]$, the process $V^{(N)}$ exhibits oscillations and does not have a
limit; see Figure~\ref{second_half_bc_}. In light of
Remark~\ref{non_existence_for_diff_theta_values}, this reflects the
continuous-time situation with $\vartheta_0\neq (n-1)/2$ and
$\vartheta_T=1/2$, where an equilibrium exists only in the zero-net-supply case
$\bar{x}=0$ and is given by $x_i\mathbbm{f}(t)$. If we assume
$\bar{x}=0$, we recover the convergence of the
discrete-time inventories $X^{(N),i}_t$ to the continuous-time equilibrium $x_i\mathbbm{f}(t)$.

\begin{figure}[htbp]
    \centering
    \includegraphics[width=0.6\linewidth]{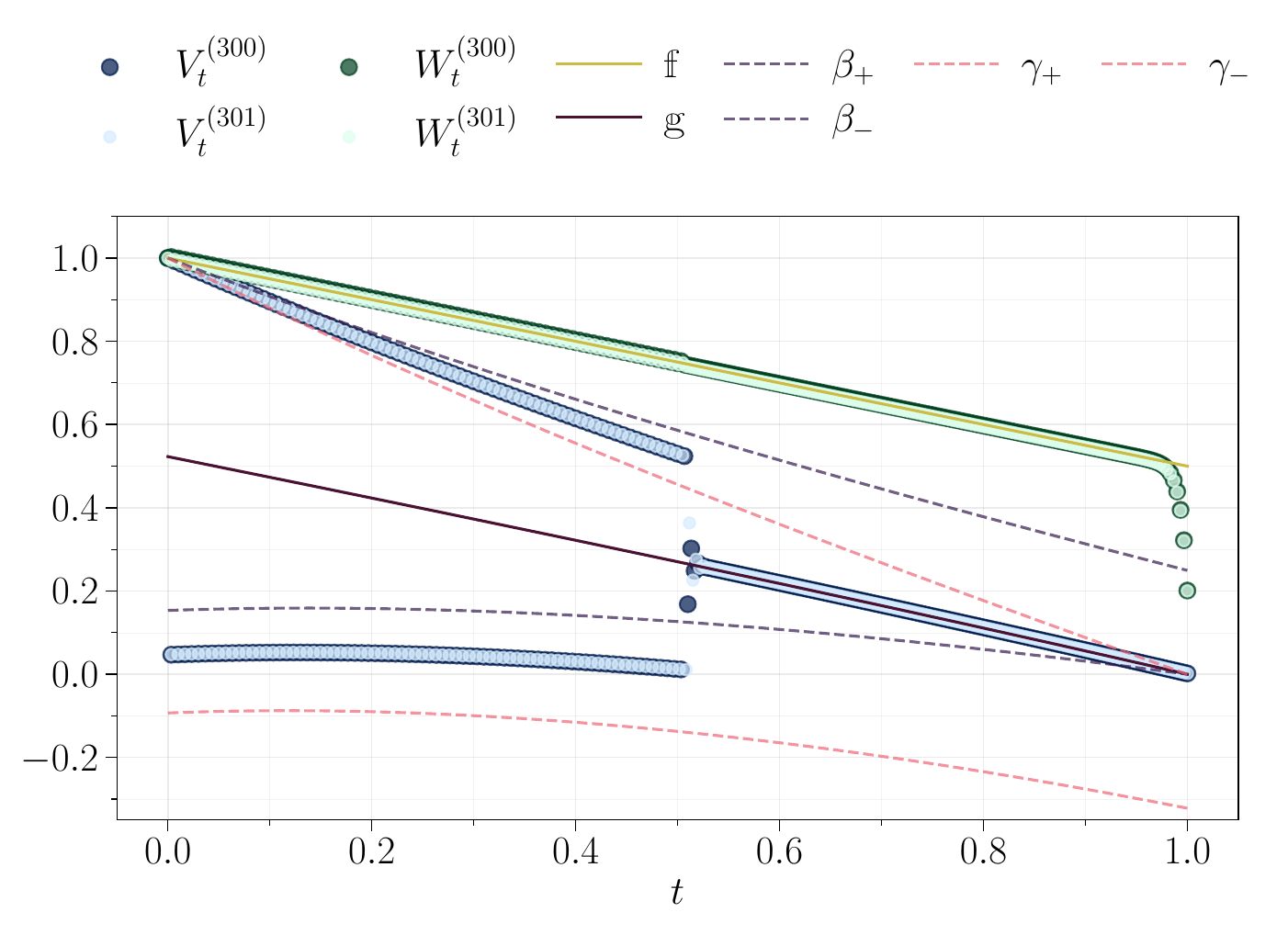}
    \caption{Convergence of inventories for even/odd values of $N$
    in the modified game with cost functional $H^1$ (instantaneous cost charged in second half). We also plot the cluster points~$\beta_\pm$ and~$\gamma_\pm$ from Theorem~\ref{strat_osc_thm} and note how they differ from the envelope of $V^{(N)}_t$.}
    \label{second_half_bc_}
\end{figure}

\subsection*{First-half instantaneous cost}

Next, we charge instantaneous costs only on
the first half of the grid, that is, we use the objective with $J^\theta$
replacing $\Gamma^\theta$. In this case we observe numerically that
\begin{align*}
    \abs{V^{(N)}_t-\mathbbm{g}(t)} \longrightarrow 0.
\end{align*}
By contrast, $W^{(N)}_t$ converges $\mathbbm{f}(t)$ only on $[0,T/2]$. On $[T/2,T]$, the process
$W^{(N)}$ oscillates and fails to converge; see Figure~\ref{first_half_bc_}.
This behavior is consistent with
Remark~\ref{non_existence_for_diff_theta_values}, which states that in the
continuous-time game with $\vartheta_0= (n-1)/2$ and $\vartheta_T\neq 1/2$, an
equilibrium exists only in the symmetric case $x_1=\cdots=x_n$ and is given by
$x_i\mathbbm{g}(t)$. If we assume $x_1=\cdots=x_n$,  we recover the convergence of the discrete-time inventories $X^{(N),i}_t$ to the continuous-time equilibrium $x_i\mathbbm{g}(t)$.

\begin{figure}[htbp]
    \centering
    \includegraphics[width=0.6\linewidth]{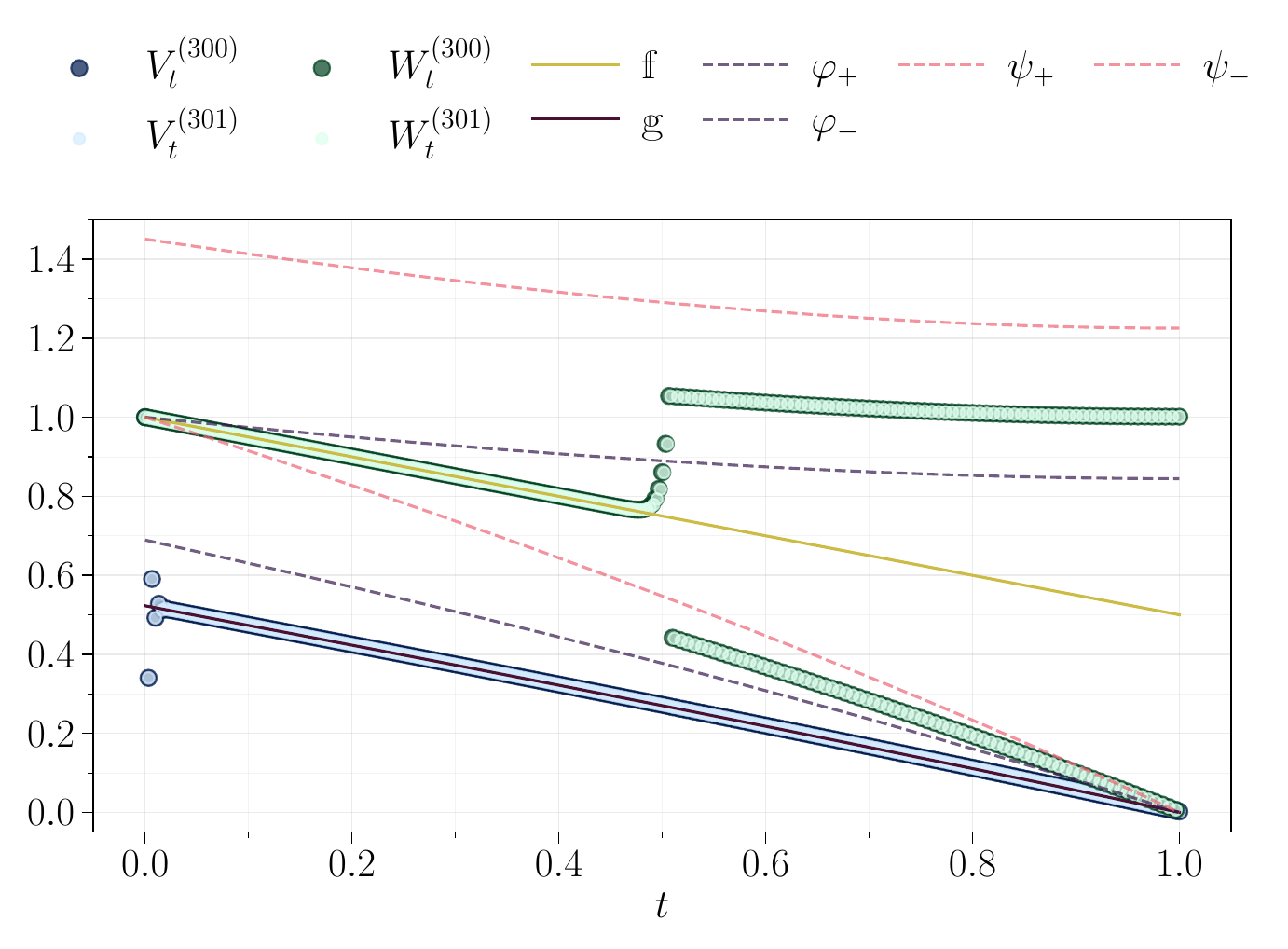}
    \caption{Convergence of inventories for even/odd values of $N$
    in the modified game with cost functional~$J^1$ (instantaneous cost charged in first half). We also plot the cluster points~$\varphi_\pm$ and~$\psi_\pm$ from Theorem~\ref{strat_osc_thm}.}
    \label{first_half_bc_}
\end{figure}

\subsection*{Comparison with the cluster points from Theorem~\ref{strat_osc_thm}}

Finally, we note that the oscillatory envelopes observed in
Figures~\ref{second_half_bc_} and~\ref{first_half_bc_} differ from the cluster points from Theorem~\ref{strat_osc_thm}, which are driven by $N$ being even or odd. In the
half-grid instantaneous-cost setting, the oscillations are localized to the
half of the grid where cost is absent, and the associated cluster points are no
longer determined by the even/odd parity of~$N$ seen in the specification with no instantaneous costs.

\bibliographystyle{abbrv}
\bibliography{stochfin.bib}

\end{document}